\documentclass[12pt,a4paper,oneside]{article}

\usepackage{graphicx}
\usepackage{amssymb}
\usepackage{amsmath}
\usepackage{amsthm}
\usepackage[margin=1.2in]{geometry}
\usepackage{tabularx}
\usepackage[utf8]{inputenc}
\usepackage{setspace}
\usepackage[longnamesfirst]{natbib}
\usepackage{color,soul}
\usepackage{appendix}
\usepackage{multirow}
\usepackage{rotating}
\usepackage{latexsym}
\usepackage{tikz}
\usepackage{pgf}
\usepackage{booktabs}
\usepackage{thmtools}
\usepackage{thm-restate}
\usepackage[nottoc]{tocbibind}
\usepackage{subfig}
\usepackage[short]{optidef}
\usepackage{todonotes}

\usepackage[hidelinks]{hyperref}

\theoremstyle{definition}
\newtheorem{definition}{Definition}

\theoremstyle{plain}
\declaretheorem{proposition}

\declaretheorem{lemma}

\declaretheorem{example}
\declaretheorem{theorem}

\newcommand{\IND}{\text{IND}}
\newcommand{\AGG}{\text{AGG}}

\begin{document}
\title{An Equilibrium Model of the First-Price Auction with Strategic Uncertainty: Theory and Empirics}
\author{Bernhard Kasberger\thanks{D\"{u}sseldorf Institute for Competition Economics (DICE), Heinrich Heine University D\"{u}sseldorf, Germany
		\vskip 0em Email: bernhard.kasberger@hhu.de\vskip 0em
		\noindent I thank David Delacr\'etaz, Matteo Escudé (discussant), Simon Finster, Paul Heidhues, Stefan Hubner, Maarten Janssen, Willemien Kets, Paul Klemperer, Simon Martin, Martin Pesendorfer, Karl Schlag and audiences at Berlin, Bonn, D\"{u}sseldorf, Manchester, Mannheim, NYU Abu Dhabi, Oxford, Vienna, the 2020 Conference on Mechanism and Institution Design, the 2020 INFORMS Annual Meeting, the 2020 European Winter Meeting of the Econometric Society, the 2021 Spring Meeting of Young Economists, and the 2021 Annual Meeting of the Austrian Economic Association for their helpful comments and suggestions. All errors are my own.
		}}
\date{\today}


\maketitle

\begin{abstract}
	In many first-price auctions, bidders face considerable strategic uncertainty: they cannot perfectly anticipate the other bidders' bidding behavior. We propose a model in which bidders do not know the entire distribution of opponent bids but only the expected winning bid and lower and upper bounds on the opponents' bids. We characterize the optimal bidding strategies and prove the existence of a novel equilibrium. Finally, we apply the model to estimate the cost distribution in highway procurement auctions. Our model fits the data better than existing approaches in which the entire distribution of opponent bidders is assumed to be known.
\end{abstract}
\newpage

\onehalfspacing

\section{Introduction}

This paper proposes a novel game-theoretic equilibrium model of the first-price auction in which observable statistics characterize the players' beliefs. In particular, bidders have a belief about the expected winning bid, a statistic that is observable in many auction markets. Indeed, households interested in bidding for a house can look up the average sales price in a neighborhood but typically not the history of transaction prices. In procurement, bidders may have a small sample of comparable transaction prices and corresponding covariates but do not observe losing offers. Bidders can then estimate the expected winning bid but not much more. Finally, while not necessarily a first-price auction, Google informs bidders about expected prices in ad auctions but keeps the bid history secret.\footnote{There are also other settings in which averages may characterize beliefs. For example, central banks communicate that they want an average inflation rate of 2\%. In a different context, there is a (known) target for the average block size on the Ethereum blockchain \citep{roughgarden-eip-1559}.} In addition to the belief about the expected winning bid, bidders also identify lower and upper bounds on the opponents' bids in our model. These can be, for example, prices in the next worst or next best neighborhoods. These three statistics, the lowest opponent bid, the expected winning bid, and the highest opponent bid, independence, and the number of bidders are all a bidder knows about the opponent bid distribution in our model.


A key characteristic of the model is strategic uncertainty, as many bid distributions are consistent with the observed statistics. Due to this uncertainty about the opponent bid distribution, bidders cannot maximize the (subjective) expected payoffs. Instead, any bid will likely lead to a lower payoff than the optimal bid had the bidder known the opponent bid distribution. We quantify this (interim) suboptimality with a loss function; the loss of a bid given the opponents' bid distribution is the difference between the highest possible payoff and the payoff from the bid. It is then natural to desire a low loss. After all, if the loss of a bid is small, the payoff is close to optimal. Therefore, we assume that bidders choose bids that minimize the maximal loss, leading to a close-to-optimal payoff for any possible bid distribution \citep{savage1951}.

The difference to traditional game-theoretic equilibrium models is that the bidders know much less about their opponents, and what they know is, in principle, observable. In a Bayes-Nash equilibrium, the predominant game-theoretic tool for analyzing the first-price auction, the bidders' value distribution is assumed to be common knowledge, and bidders know the opponents' bidding strategies in equilibrium. Neither is typically observable. Usually, only (statistics of) the bids are observable. In our model, there is no common prior as bidders are ignorant of the value distribution and opponents' strategies; they reason directly about the bid distribution. Their beliefs take the form of statistics of the bid distribution. In contrast, bidders know the entire bid distribution in a Bayes-Nash equilibrium; that is, they can compute for each bid the probability that it wins the auction.\footnote{A survey among academic and professional auction consultants shows that many bidders in high-stake auctions struggle to assign probabilities to opponent bids and values \citep{kasberger-schlag}.} We maintain correct beliefs but dispense with the assumption that the entire distribution is known.

We introduce with \emph{aggregate moment equilibrium} a novel equilibrium notion. In an aggregate moment equilibrium, the players play optimally given their beliefs, and the beliefs are correct given the optimal strategies and the latent (independent) value distribution. We prove that an aggregate moment equilibrium exists and that the beliefs induce bidding functions that are strictly increasing in the bidder's private value (Theorem~\ref{theorem:aggregate_moment_equilibrium}). The equilibrium beliefs naturally depend on the latent value distribution and are such that the lowest bid by any type is the belief about the lowest bid, the highest bid by any type is the belief about the highest bid, and the expected winning bid equals the belief about the expected winning bid. The features of optimal behavior given beliefs and consistent beliefs given optimal behavior and the latent value distribution are shared with Bayes-Nash equilibrium.

Instead of observing the average winning bid, bidders may observe the average bid. Hence, we also present a closely related but alternative model in which bidders have a belief not about the expected winning bid but rather the expected opponent bid (in addition to the lowest and highest possible bids). 
Such a belief might come from an observed list of past bids and corresponding auction covariates. Suppose bidders use the sample to estimate a (possibly linear) relationship between the bids and the covariates.\footnote{Due to the auction heterogeneity, it might not be best to naively best respond to the sample as players do in a sampling equilibrium \citep{OSBORNE2003434}. \cite{doi:10.3982/ECTA17105} consider players who use a sample estimate the distribution of actions. In our model, players would only estimate the first moment due to the covariates. We implicitly assume that the sample is small in relation to the number covariates.} Plugging in the covariates of the current auction gives an estimate of the expected opponent bid.\footnote{The survey of \cite{kasberger-schlag} reveals that some real-world bidders estimate the expected opponent bids.} Alternatively, bidders might directly observe the average bid. For concreteness, we assume that each bidder believes that the expected bids of the other bidders are the same. We again define an equilibrium notion (\emph{moment equilibrium}) in which beliefs are consistent with optimal behavior and the latent value distribution. We characterize the minimax-loss bidding function and prove that a moment equilibrium exists in which the minimax loss bidding function is strictly increasing in the bidder's private value (Theorem~\ref{theorem:individual_moment_equilibrium}). To emphasize that beliefs are now directly about individual behavior, we sometimes refer to this moment equilibrium as the individual moment equilibrium.


The difference between the two types of belief is that with individual moment beliefs, bidders do not face the inference problem that they face with aggregate moment beliefs. With aggregate moment beliefs, bidders have to translate their belief about the expected winning bid to a belief about the distribution of the highest opponent bid. To do so, we assume that bidders believe that their belief about the expected winning bid comes from auctions with $n$ bidders whose bids are i.i.d. draws from some individual bid distribution. In the current auction are also $n$ bidders, so a bidder's payoff depends on the distribution of the maximum of the $n-1$ highest bid. Bidders then consider all individual bid distributions as possible for which the expected value of the highest of $n$ independent draws is their belief about the expected winning bid. They then use independence to identify the worst-case distribution of the highest of $n-1$ independent bids. With moment beliefs directly about the individual bid distributions, the first inference is unnecessary.

Next, we show how we can use our equilibrium models with beliefs in the form of observable statistics for the structural estimation of the latent value distribution. We prove analogous results for both types of belief but now discuss the estimation for the model with a belief about the expected winning bid. The backbone of the structural estimation is a non-parametric identification result: the bid distribution identifies the value distribution if bidders play the strictly increasing aggregate moment equilibrium strategies. The first step in the estimation estimates the players' beliefs non-parametrically by computing the sample minimum and maximum bids as well as the average winning bid. The second step uses the estimated beliefs and the characterized minimax-loss bidding function to translate bids into ``pseudo values.'' In contrast, the first step in the leading non-parametric estimator based on the Bayes-Nash equilibrium as proposed by \cite{GPV} [GPV] requires specifying a kernel and a bandwidth to estimate the probability density function of the bid distribution. The second step in GPV uses the first-order condition of expected utility maximization to translate bids into pseudo values. A notable difference is how easy it is to estimate the bidders' beliefs in the moment-equilibrium-based approach as one only needs to compute some summary statistics. While both estimators are consistent (given that the respective theoretical model generates the data), we expect the moment-equilibrium-based estimator to perform well in small samples as the estimation of the beliefs is super-consistent \citep{donald-paarsch-2002}.

Finally, using the data of \cite{BHT-AER-2014}, we estimate the latent cost distribution in Californian highway procurement auctions non-parametrically. Out-of-sample predictions reveal that our equilibrium approaches perform empirically at least as well as Bayes-Nash equilibrium. In these out-of-sample predictions, we compare five estimators for estimating the latent cost distribution. On the one hand, we use the estimators based on aggregate and individual moment beliefs, respectively, each time with and without removing (high losing) outliers based on the interquartile range. On the other hand, we use the Bayes-Nash-equilibrium-based estimator of GPV as a benchmark. For each number of bidders per auction $n$, we use the auctions with a different number of bidders to estimate the cost distribution and then compute the equilibrium bid distribution had there been $n$ bidders. Finally, we assess the fit by computing the distance between the sample and the predicted bid distribution. Across the specifications, we find that robust models that do not make strong assumptions about bidders' knowledge have higher explanatory power of the observed data than BNE. When all bidders are large corporations, the aggregate moment equilibrium performs best. When all bidders are small, the individual moment equilibrium performs best. In that case, removing outliers for estimating the highest possible bid turns out to be particularly useful because there are some unduly high bids made by fringe firms in the sample. 

\subsection{Related Literature}

The paper stands in the tradition of the so-called \cite{wilson-critique} doctrine that seeks to weaken the common knowledge assumptions made in analyzing practical problems with game-theoretic methods. The extensive literature in this tradition has mostly focused on robust mechanism design. Ex-post implementation weakens the common prior but maintains the assumption that players know the other players' strategies \citep{BM05}. Note that the first-price auction does not admit an ex-post equilibrium. 
Other approaches consider non-Bayesian designers dealing with uncertainty (ambiguity) about aspects of the environment by minimizing the maximal regret \citep{BERGEMANN20112527,RobustMonopolyRegulation,guo-shmaya-project-choice} or maximizing worst-case expected utility \citep{Carroll15,Carroll-Segal-resale-2016}. We consider non-Bayesian players and define an equilibrium notion when these players do not have a common prior and do not know other players' strategies.


Other papers that define equilibrium concepts for games with non-Bayesian players facing ambiguity broadly fall into two strands. A first strand considers games of complete information and players being uncertain about the strategies of the other players. A second strand considers games of incomplete information with players being uncertain about the type (i.e., value) distribution. The first strand is more related as the second assumes that players know the other players' strategies in equilibrium.\footnote{\cite{10.2307/25055108}, \cite{chen2007}, and \cite{Gretschko-Mass} analyze the first-price auction and belong to the second category as the players are uncertain about the value distribution, but there is no strategic uncertainty in equilibrium. \cite{hyafill-boutilier} and \cite{Schlag-Zapechelnyuk-best-compromise} do not analyze first-price auctions but use minimax loss to deal with the unknown type distribution.} Within the first strand, \cite{10.1257/mic.4.1.70} defines with \emph{partially specified equilibrium} a related equilibrium notion for games of complete information. In his model, similar to the moment and range, players' beliefs are characterized by expectations of random variables that depend on the mixed strategies of the other players. However, in a first-price auction, forming beliefs about strategies is only useful if one also has a belief about the value distribution. Such knowledge conflicts with our assumption that bidders' beliefs are characterized by observable statistics; the value distribution is typically unobserved.\footnote{Another difference is the decision criterion. \cite{10.1257/mic.4.1.70} uses maximin expected utility, whereas we use minimax loss. Maximin expected utility does not lead to sharp predictions in the first-price auction with moment beliefs. Using smooth ambiguity, \cite{10.1257/aer.20130930} propose a related solution concept \citep[cf.][]{BATTIGALLI201640}.}

There are also other solution concepts for games of complete information with equilibrium strategic uncertainty in the sense of multiple priors. These concepts differ in how they determine the belief---in the form of a set of other players' possible strategies---in equilibrium. \cite{LO1996443} has the restrictive requirement that only optimal mixed strategies are deemed possible. \cite{Klibanoff1996} and \cite{RENOU2010264} demand that the optimal strategy is just one of the strategies deemed possible, which is quite permissive as they do not determine the entire belief set in equilibrium.\footnote{Other solution concepts in this vein are by \cite{DOW1994305}, \cite{EICHBERGER2000183}, and \cite{MARINACCI2000191} for players maximizing Choquet expected utility \citep{schmeidler1989}.} In our model, the belief set does not contain strategies (actions) but rather distributions of actions (bids). Moreover, we parameterize the belief set with the moment and range beliefs. Similar to \cite{LO1996443}, all possible bid distributions reflect the truth in equilibrium by having the same mean as the true bid distribution in our model. At the same time, echoing \cite{Klibanoff1996} and \cite{RENOU2010264}, the equilibrium bid distribution is just one of the possible.

There are other papers in economics \citep{10.1257/aer.20151462,Carrasco-et-al-2018} and in operations research \citep{Smith1995,Popescu2005,Perakis-Roels-2008} that use moment constraints. These papers do not determine them in equilibrium.

\cite{kasberger-schlag} also analyze the first-price auction under strategic uncertainty with minimax loss as the decision criterion. Their objective is not to explain behavior but to find bidding rules that work well for a single bidder. Without using moment beliefs, they construct the set of possible bid distributions for a single bidder; we determine the set in equilibrium. Empirically, they find that their bidding rules would have led to a higher expected utility, on average, which implies that their approach is not well suited for explaining observed behavior. Note that one can use the minimax loss bidding functions of the current paper in such a normative way. \cite{LINHART1989152} is an early paper that uses minimax loss for the analysis of a strategic decision problem (bilateral bargaining). However, they do not define an equilibrium. 
\section{The Model}
\label{sec:model}

In a nutshell, the $n$ bidders in a first-price auction have independent private values and a belief about the lowest possible bid, the highest possible bid, and the expected winning bid. As there are many bid distributions that are consistent with these three statistics, the bidders face (strategic) uncertainty. They deal with the uncertainty by minimizing the maximal loss. The beliefs are correct in equilibrium.


\subsection{Basics}
A single good is for sale in a first-price sealed-bid auction. The auction rules are such that all bidders simultaneously submit sealed bids. The bidder with the highest bid wins and pays her bid. In the case of ties, the good is allocated uniformly between the bidders with the highest bid. There are $n$ bidders in the auction, $n\ge 2$. Let $N$ denote the set of players.

The bidders are risk neutral. Before the auction, each bidder $i$ learns her willingness-to-pay $v_i$ for the auctioned item. These private values are independent, identically distributed, and drawn from a distribution with the cumulative distribution function $F$. The lowest and the highest value in the support are denoted by $\underline v$ and $\overline v$, respectively, with $0\le \underline v<\overline v<\infty$. Let $\text{supp}(F)$ denote the support of $F$. The payoff from losing is normalized to 0. The payoff from winning with bid $b_i$ equals $v_i-b_i$. Hence, conditional on a profile of bids $(b_1,\dots,b_n)$, the expected utility is equal to
\begin{equation*}
	u_i(b_1,\dots,b_n)=
	\begin{cases}
		\frac{1}{k}(v_i-b_i) &\text{if }b_i = \max_{j \in N} b_j\text{ and } k = |\{j\in N\colon b_j= b_i\}|\\
		0&\text{if } b_i < \max_{j\in N} b_j.
	\end{cases}
\end{equation*}

Bidders choose a bid conditional on their willingness-to-pay. From an interim perspective, bidder $i$'s behavior takes the form of choosing a strategy $\beta_i:\text{supp}(F)\rightarrow \mathbb R_+$ that maps values into bids. The strategy is measurable with respect to the underlying Borel $\sigma$-algebra. Importantly, the value distribution $F$ and the bidding strategy $\beta_i$ generate a bid distribution $B_i\in\mathcal P$, where $\mathcal P$ is the set of measures on $\mathbb R_+$ (defined on the Borel $\sigma$-algebra). Bidders bid independently so that the distributions $B_i$ and $B_j$ are independent for $i\neq j$. The set of probability distributions on set $X$ is denoted by $\Delta X$. 

The bid distributions of the other bidders determine the expected utility of a bid $b_i$. Formally, for joint bid distributions $B_{-i}\in \Delta \mathbb R_+^{n-1}$, the expected utility $U_i(b_i, B_{-i})$ of bidder $i$ while bidding $b_i$ is given by 
\begin{equation*}
	U_i(b_i, B_{-i}) = \int u_i(b_i,b_{-i})dB_{-i}(b_{-i}).
\end{equation*}
In the absence of ties and with independent bid distributions, $\Pi_{j\neq i}B_j(b_i)$ is the probability that $b_i$ is higher than all of the other bidders' bids. Then the expected utility is $U_i(b_i, B_{-i}) = \Pi_{j\neq i} B_j(b_i)(v_i-b_i)$.

Bidders know that there are $n$ bidders in total, that values are independent draws, and that everyone bids independently. Each bidder believes that the other players are ex-ante identical, i.e., that their bids come from the same bid distribution.

The bidders do not know the value distribution (not even the support) and do not know the strategies of the other players.\footnote{Bidders using their belief about the value distribution would impose additional constraints on the set of possible bid distributions. It would be interesting to study the interaction of beliefs about the value distribution and the bid distribution in future work.} Instead of forming beliefs about the value distribution and the strategies, they form beliefs directly about the bid distribution, the object that determines the expected utility. From bidder $i$'s perspective, the bid of bidder $j$ is a random variable with distribution $B \in \mathcal P$.

Bidder $i$'s belief about the other bidders' behavior takes the form of the set of all joint bid distributions $B_{-i}$ that she thinks are possible. Let $\mathcal B_i\subseteq\Delta\mathbb R_+^{n-1}$ denote this set of possible bid distributions (the belief set). Thus, the belief is a set of multiple priors. Note that the agents are not Bayesians \citep{savage1954}, i.e., they do not have a probabilistic prior over the set of possible bid distributions. We put structure on the players' belief sets after defining how players choose their bids for a given belief set. 

Bidder $i$ deals with the uncertainty about the opponent bid distribution by minimizing the maximal loss \citep{savage1951}. Given a bid distribution, the loss of a bid is the difference between the highest possible expected utility and the actual payoff. Formally, the loss of bidding $b_i$ is
\begin{equation*}
	\lambda_i(b_i,B_{-i}|v_i) = \sup_{\tilde b_i \in \mathbb R_+} U_i(\tilde b_i,B_{-i}) - U_i(b_i,B_{-i}).
\end{equation*}{}
Player $i$ chooses the bid $b_i$ to minimize the maximal loss, i.e.,
\begin{equation*}
	b_i\in \arg\inf_{\tilde b_i\in \mathbb R_+} \sup_{B_{-i}\in\mathcal B_{i}}\lambda_i(\tilde b_i,B_{-i}|v_i).
\end{equation*}

Minimax loss seeks good performance for all possible bid distributions. As such, it qualifies as \emph{robust}. If the loss of a bid is small, then the bid is close to optimal for the given bid distribution. If maximal loss is small, then one loses little payoff due to not knowing the bid distribution. According to \cite{LINHART1989152}, minimax loss can be used to find compromises in group decisions with diverse subjective probability judgments. Minimax loss bids are easy to justify as one can lay down the objective facts that characterize the belief set and then choose the bid that performs well for all possible bid distributions. Note that we adopt an interim perspective and compute loss conditional on the opponent bid distribution, as in \cite{kasberger-schlag}.\footnote{An alternative would be to compute loss ex-post conditional on the opponents' realized bids. To distinguish these two perspectives, \cite{kasberger-schlag} refer to the interim concept as loss and to the ex-post concept as regret.}

An alternative to minimax loss is maximin expected utility \citep{GILBOA1989141}. In contrast to minimax loss, it does not lead to sharp predictions as all bids are optimal for low values. To see this, consider a bidder with a willingness-to-pay below the expected winning bid. The worst case occurs when one loses the auction. A corresponding worst-case bid distribution puts all the mass on the expected winning bid. Worst-case expected utility is then maximized by any bid below the willingness-to-pay. Another alternative is to assume that bidders have a probabilistic belief about which bid distribution materializes, in which case the beliefs are no longer characterized by statistics.

\subsection{Aggregate Moment Belief}

The belief (set) $\mathcal B_{i}$ of bidder $i$ about the other bidders' bidding behavior is characterized by two types of belief. First, bidder $i$ believes that other bidders only choose bids in $[l_i,u_i]$, where $0\le l_i<u_i$. Note that bidder $i$ forms identical beliefs about the other bidders, i.e., $(l_i,u_i)$ is independent of $j$. Second, bidder $i$ believes that in auctions like the current one the average winning bid is $m_i$. Thus, the other bidders bid in such a way that the expected value of the maximum of $n$ independent draws from their bid distributions is $m_i$. We refer to the first type of belief as \emph{range belief} and to the second type as (aggregate) \emph{moment belief}. 

The range and moment beliefs parameterize the belief set, which is from now on denoted by $\mathcal B_i(l_i,m_i,u_i)$. Let $G$ be the distribution function of a real-valued random variable and let $G^n$ denote the distribution of the maximum of $n$ independent draws from $G$. The set of possible bid distributions is 
\begin{align*}
	\mathcal B_{i}(l_i,m_i,u_i) = \{B_{-i}\in\Delta \mathbb R_+^{n-1}:\exists B\in\mathcal P\text{ with }B_{-i}(b_{-i})=\Pi_{j\neq i}B_j(b_j), \\\int_{l_i}^{u_i}dB=1\text{ and } \int xdB^n(x)=m_i\}.
\end{align*}
Only bid distributions that put mass one on $[l_i,u_i]$ are possible. Moreover, the highest bid of $n$ independent bids must be $m_i$ in expectation.

In first-price auctions, the expected utility depends on the highest competing bid, that is, the highest bid of $n-1$ bids. In the current model, bidders use the belief about the maximum of $n$ bids to form beliefs about the distribution of individual bids and then use the distribution of individual bids to form beliefs about the distribution of the highest bid of $n-1$ bids. The key assumptions are that the number of bidders $n$ is known (that ``generated'' $m_i$ and that participate in the current auction), and that the competing bids are i.i.d. There are clearly interesting alternative models.\footnote{In our model, bidders are sophisticated as they translate their belief about the expected winning to a (worst-case) belief about the distribution of the highest bid of $n-1$ bidders. Alternatively, bidders could be naive in that they ignore the necessary inference and act as if the highest opponent bid is $m$ on average. As the number of bidders grows, the difference between naive and sophisticated bidders disappears.} The current paper focuses on the case of known $n$ and i.i.d. bids to keep the model simple, close to the standard textbook case, and applicable for the structural estimation.




\subsection{Aggregate Moment Equilibrium}

The model is closed by assuming that the beliefs are not systematically wrong. The beliefs are consistent or in equilibrium if the beliefs match the statistics of the bid distribution generated by the latent value distribution and the optimal strategies. 

Let $\beta=(\beta_1,\beta_2,\dots,\beta_n)$, $l=(l_i)_{i\in N}$, $m=(m_i)_{i\in N}$, and $u=(u_i)_{i\in N}$.

\begin{definition}
	\label{def:aggregate_moment_equilibrium}
	The tuple $(\beta,l,m,u)$ is an \emph{aggregate moment equilibrium} if
	\begin{enumerate}
	 	\item $\beta_i$ minimizes the maximal loss for player $i\in N$ and value $v_i\in\text{supp}(F)$, i.e.,
	 	\begin{equation}
	 		\beta_i(v_i|l_i,m_i,u_i) \in\arg\inf_{b_i\in\mathbb R_+}\sup_{B_{-i}\in\mathcal B_i(l_i,m_i,u_i)} \lambda_i(b_i,B_{-i}|v_i);
	 	\end{equation}
	 	\item The aggregate moment belief is consistent across the players and with $\beta$ and $F$, that is,
	 	\begin{equation}
	 		\int\int\dots \int \max_{i\in N}\beta_i(v_i|l_i,m_i,u_i) d F(v_1) dF(v_2) \dots dF(v_n) = m_i
	 		\label{eq:aggregate-moment-equilibrium}
	 	\end{equation}
	 	for all $i\in N$;
	 	\item The range beliefs are consistent, i.e., for all $i\in N$ and $j\neq i$
	 	\begin{align*}
	 		\inf_{v\in\text{supp}(F)}\beta_i(v|l_i,m_i,u_i)&=l_j,\text{ and}\\
	 		\sup_{v\in\text{supp}(F)}\beta_i(v|l_i,m_i,u_i)&=u_j.\\
	 	\end{align*}
	 \end{enumerate} 
\end{definition}

With symmetric strategies and beliefs (dropping the indices), Equation~\eqref{eq:aggregate-moment-equilibrium} becomes
\begin{equation*}
	\int_{\underline v}^{\overline v} \beta(v|l,m,u) dF^n(v)=m,
\end{equation*}
that is, we can compute the expected winning minimax bid with respect to the distribution of the highest value. Moreover, the requirement of consistent range beliefs simplifies to $\inf_{v\in\text{supp}(F)}\beta(v|l,m,u)=l$ and $\sup_{v\in\text{supp}(F)}\beta(v|l,m,u)=u.$


The definition requires that (i) behavior is optimal given the beliefs, and that (ii) the beliefs are consistent with behavior and the latent value distribution. The same two elements are present in a Bayes-Nash equilibrium. The difference is that in a BNE the belief is about the entire bid distribution and not just some (aggregate) statistics. Thus, one can view moment equilibrium as an analogue to Bayes-Nash equilibrium in which bidders know less than the full distribution in equilibrium. 


\section{Existence of an Aggregate Moment Equilibrium}

The section contains the main theoretical result of the paper, namely that an aggregate moment equilibrium exists for all value distributions. The proof of the theorem relies on the characterization of the minimax bidding strategy, with which we begin the analysis.

\subsection{Optimal Bidding Function}

Bidder $i$ faces the problem of not being able to know the bid that maximizes the expected utility (the `best response' bid) as she does not know the competing bid distribution. Thus, any bid is likely to be suboptimal in the sense of not maximizing the expected utility. 
The bid can be suboptimal for two reasons. On the one hand, the bid can be too high relative to the best response or it can be too low. Both types of mistake lead to a loss in payoff relative to the highest possible. 

The bidder's solution to the problem is to choose the bid that minimizes the maximal loss in payoff due to not knowing the bid distribution. Such a bid performs uniformly relatively well as it leads to a payoff that is close to the highest possible for any feasible bid distribution.

Intuitively, loss associated with bidding too high increases in the bid: the higher the bid, the higher the loss in utility if the best response is a very low bid. If the bid is as low as the lowest possible best response, i.e., $b=l$, then the bid cannot be too high. So the loss of bidding too high is 0 at $b=l$. 

Conversely, the loss of bidding too low decreases in the bid: the lower the bid $b$, the higher the loss in utility if the best response is above $b$. If the bid is the highest undominated bid, i.e., $b=v$, then the bid cannot be too low. The loss of bidding too low is 0 for $b=v.$ 


Maximal loss is the maximum of the maximal loss of bidding too high and bidding too low. The maximal loss is minimized by equalizing the two maximal losses. This explains the form of the minimax bidding function in the following proposition, which characterizes the minimax bidding function.

\begin{proposition}
	Suppose bidder $i$ believes that the other bidders' bids are i.i.d$.\,$draws from some distribution with the lowest bid $l$ and the highest bid $u$, $l\le \underline v$. Moreover, the distribution is such that the expected winning bid is $m$ in auctions with $n$ bidders. If $v_i\le m + (m-l)\left(\frac{u-m}{u-l} \right)^{\frac{n-1}n}$, let $b_i$ be the unique solution of  
	\begin{equation}
		\left(\frac{u-m}{u-l}\right)^{\frac{n-1}{n}}(b_i-l) = \max\{\left(\frac{u-m}{u-b_i}\right)^{\frac{n-1}{n}}(v_i-b_i),v_i-m\}.
		\label{eq:agg_minimax_bid_low_value}
	\end{equation}
	For higher values, let $b_i$ be defined as the unique solution of 
	\begin{equation}
		\left(\frac{u-m}{u-l}\right)^{\frac{n-1}{n}}(b_i-l) =(v_i - b_i)\left(1 - \left(\frac{b_i-m}{b_i-l} \right)^{\frac{n-1}{n}}\right)
		\label{eq:agg_minimax_bid_high_value}
	\end{equation}
	The optimal bid of bidder $i$ is $\beta^\AGG(v_i|l,m,u):=\min\{b_i,u\}$. The optimal bidding function $\beta^{\AGG}$ is continuous in $v$, $l$, $m$, and $u$.
	\label{prop:aggregate-minimax-bdding-function}
\end{proposition}

The implicit equations in \eqref{eq:agg_minimax_bid_low_value} and \eqref{eq:agg_minimax_bid_high_value} equate the maximal loss from bidding too high and too low. The respective left-hand side of the equations give the worst-case loss of bidding too high (which is increasing in $b_i$), and the right-hand side state the maximal loss when the best response is higher than $b_i$ (which decreases in $b_i$).

The minimax bidding function takes a different form below $m$ and above $m$. This kink is due the constraints of the moment belief on the worst-case distributions. In particular, the implicit equations \eqref{eq:agg_minimax_bid_low_value} and \eqref{eq:agg_minimax_bid_high_value} show that the worst case when bidding too low depends on the bid being below or above the moment belief $m$.

We now explain the reason for the kink in the bidding function and the shape of worst-case loss. The proof of the proposition shows that the worst-case bid distribution has at most two elements in the support. Two elements are enough to satisfy the moment constraint and accentuate what bidding too high and bidding too low means. The moment constraint requires one mass point below $m$ and one mass point above $m$. The best response to such a discrete bid distribution is to bid marginally above one of these mass points. The worst case never places all the mass above the value so that any bid is optimal.

In the case of bidding too high, one leaves most money on the table if lowering the bid does not change the probability of winning (while raising the surplus conditional on winning). This is pushed to the extreme if one mass point of the bid distribution is at the lowest possible bid, i.e., $x_1=l$. The worst case then maximizes the probability with which others bid $l$. The second mass point that maximizes this probability is $x_2=u$, implying the probability that satisfies the moment constraint. This maximizes the loss of bidding too high.

In the case of bidding too low, the worst that can happen is if raising the bid marginally leads to a discrete increase in the winning probability. Note that this decreases the surplus conditional on winning only marginally. If one bids below $m$, then one mass point of the bid distribution is just marginally above $b$. To maximize the probability with which one could win by raising the bid marginally, the second mass point is at $u$. However, it can be the case for high values bidding below $m$ that loss is higher if all the mass is placed on $m$. This follows from the trade-off between the probability of winning and the surplus conditional on winning. For bids above $m$, the mass point is again just above the bid. But now some mass below $m$ is needed and this mass is to be minimized as it positively impacts the payoff of the bid $b$. The worst that can happen is a mass point at $l$.

Note that the equalization of the two formulas for loss may lead to a bid above $u$. It is never optimal to bid above $u$ as such a bid is believed to win with certainty; a marginally lower bid would increase the payoff. Thus, the bidding function is capped at $u$.

Probabilities take the form of $q^{\frac{n-1}n}$ in the minimax bidding function. This form results from the inference from the aggregate moment belief. If the bid distribution puts mass on two bids $x_1$ and $x_2$, then the moment constraint $q\cdot x_1+(1-q)\cdot x_2=m$ pins down $q$. The probability that a single bidder bids $x_1$ is $q^{\frac{1}n}$, and the probability that the maximum of $n-1$ independent draws is $x_1$ is $q^{\frac{n-1}n}$.

The following proposition discusses how the minimax bid reacts to changes in the parameters.

\begin{proposition}
	\label{prop:comparative_statics_aggregate_moment_belief}
	Let $v\in[\underline v,\overline v]$ be such that $\beta^\AGG(v|l,m,u)\in(l,u)$. The minimax bidding function increases in $v$, $m$, and $n$. The minimax bidding function decreases in $u$. For $v<m + (m-l)\left(\frac{u-m}{u-l} \right)^{\frac{n-1}n}$, the minimax bid increases in $l$. The minimax bid reacts ambiguously to changes in $l$ for higher types.
\end{proposition}
Minimax loss bids increase in value. The other comparative statics are also as expected: More competing bidders and a higher expected winning bid increase the bids. An increase in $u$ increases the probability with which other bidders bid $l$ in the worst case when bidding too high; the optimal reaction is to lower the bid. Similarly, low types can bid higher as $l$ increases. For types bidding above $m$ a higher $l$ decreases the maximal loss of bidding too high (which has a positive effect on the minimax bid) and decreases the maximal loss of bidding too low (which has a negative effect on the minimax bid). Depending on the type, both effects can dominate.

\subsection{Equilibrium Existence}

The minimax bidding function of Proposition~\ref{prop:aggregate-minimax-bdding-function} takes arbitrary, possibly misspecified, beliefs as input. We now show that there are consistent (equilibrium) beliefs.

The following theorem establishes the existence of an aggregate moment equilibrium. The beliefs, the bidding functions, and the value distribution induce a bid distribution such that the range and the expected winning bid match the players' beliefs. Moreover, it shows that there is a \emph{separating} moment equilibrium, i.e., the beliefs $(l,m,u)$ are such that the bidding functions are strictly increasing in $v$. 

\begin{theorem}
	\label{theorem:aggregate_moment_equilibrium}
	 There is an aggregate moment equilibrium $(\beta^{\AGG},l^{\star},m^{\star},u^{\star})$ for any value distribution $F$ such that the bidding function $\beta^\AGG(v|l^{\star},m^{\star},u^{\star})$ is strictly increasing in $v$.
\end{theorem}

We first discuss the result and then the main points of the proof. The theorem proves that a pure-strategy separating moment equilibrium exists for any value distribution. This stands in contrast to the Bayes-Nash equilibrium, where pure equilibria tend to exist in first-price auctions only for continuous distributions. The difference is due to the nature of beliefs. In a BNE, bidders have a (correct) belief about the other bidders' strategies. If a bidder knows that the other bidders make a certain bid with positive probability, then it is optimal for some types to bid marginally higher. This leads to a mixed strategy in equilibrium. No such detailed information is available to the bidders in a moment equilibrium.

The moment equilibrium is similar to the Bayes-Nash equilibrium as the beliefs are correct in equilibrium. The difference is the level of detail. In an aggregate moment equilibrium the beliefs are statistics: the lowest and highest bids and the expected winning bid. In a BNE, the beliefs are about the entire bid distribution and not just statistics. A monotone (separating) equilibrium can exist in both types of equilibrium.

Coming to the main points of the proof, the second point of the definition of an aggregate moment equilibrium (Def.~\ref{def:aggregate_moment_equilibrium}) implies that the moment beliefs of the players must be the same in equilibrium, i.e., $m_1=m_2=\dots=m_n$. The third point of the definition implies that the bidders must have identical range beliefs. 

The first step in the proof shows that $l^\star=\underline v$ is the unique consistent lower bound belief. This can be inferred from the minimax bidding function. As the minimax bidding function increases in $v$, the minimum is attained at $\underline v$. Suppose $l<\underline v$ and $\beta^\AGG(\underline v|l,m,u)=l$. The equality
\begin{equation*}
	\left(\frac{u-m}{u-l}\right)^{\frac{n-1}{n}}(b_i-l) = \max\{\left(\frac{u-m}{u-b_i}\right)^{\frac{n-1}{n}}(v_i-b_i),v_i-m\}
\end{equation*}
does not hold for $b_i=l$ as the left-hand side is 0 and the right-hand side is strictly positive.

Requiring separation pins down a unique consistent upper bound belief $u$ for each moment belief $m$. The requirement is that the highest type $\overline v$ is the unique type that bids $u$. This is the case if $u$ solves
\begin{equation*}
	\left(\frac{u-m}{u-l}\right)^{\frac{n-1}{n}}(u-l) =(\overline v - u)\left(1 - \left(\frac{u-m}{u-l} \right)^{\frac{n-1}{n}}\right).
\end{equation*}
There is a unique $u$ that solves the equation for each $m$. Let $u(m)$ denote the solution as a function of $m$. This function increases in $m$ (as shown in the appendix).

The final step of the proof studies the function
\begin{equation*}
	\varphi^\AGG(m) = \int_{\underline v}^{\overline v}\beta^\AGG(v|l=\underline v, m ,u=u(m))dF^n(v)
\end{equation*}
for $m\in[\underline v,\overline v]$. The function $\varphi^\AGG$ returns for each aggregate moment belief $m$ the corresponding expected winning bid. A fixed point of this function is thus an aggregate moment equilibrium belief. The domain of the function is the set of feasible moment beliefs. This set is $[\underline v,\overline v]$ as all types bid above $l$, which is equal to $\underline v$, and no type bids above value, so all bidders bid below $\overline v$. The same arguments establish that the interval $[\underline v,\overline v]$ is also the codomain. The existence of a fixed point is shown by arguing that $\varphi^\AGG(m)>m$ for $m$ in a neighborhood of $\underline v$ and $\varphi^\AGG(m)<m$ for $m$ in a neighborhood of $\overline v$. For $m$ close to $\underline v$, the vast majority of types bid above $m$, so the expected winning bid is higher than the moment belief. Similarly, for very high moment beliefs the vast majority bid below $m$ so that the expected winning bid is below $m$. As the function is continuous, it must intersect the 45-degree line at least once. The aggregate moment belief at the intersection is consistent by construction and it induces separating range beliefs.

Note that if there is only one possible value ($\underline v = \overline v)$, then a separating equilibrium exists trivially ($l^\star=m^\star=u^\star=\underline v$). If there are at least two types, then the previous limit arguments apply. 

The following example illustrates the equilibrium construction and behavior.
\begin{example}
	We consider two bidders, $n=2$, and the uniform value distribution, i.e., $F(v)=v$ for $v\in[0,1]$. Figure~\ref{fig:var_phi} shows the function $\varphi^\AGG(m)$ that returns for the aggregate moment belief $m$ the corresponding expected winning bid if the range beliefs induce separating bidding functions. The aggregate moment equilibrium is at the intersection of $\varphi^\AGG$ and the identity function $\text{id}(m)=m$. The equilibrium beliefs are $l^{\star}=0$, $m^{\star}=0.37$, and $u^\star = 0.5$. The average bid is 0.28, which is higher than the corresponding average bid of 0.25 in the Bayes-Nash equilibrium. Figure~\ref{fig:beta} displays the moment equilibrium bidding function. The kink at $m^\star$ is clearly visible. The figure also shows the BNE bidding function for comparison. BNE bids are lower than moment equilibrium bids.
\end{example}


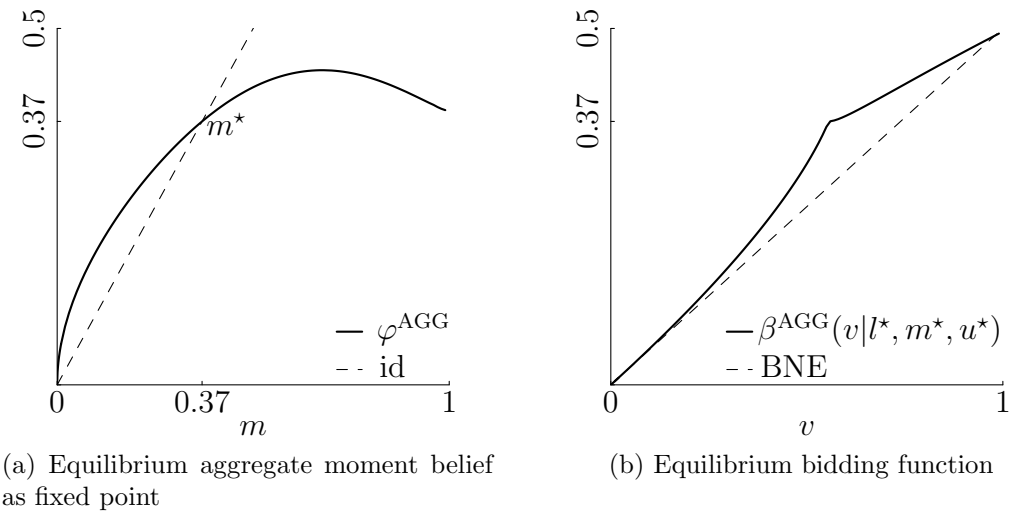
\begin{figure}
     \centering
     \subfloat[Equilibrium aggregate moment belief as fixed point]{\label{fig:var_phi}
\begin{tikzpicture}[x=1pt,y=1pt, scale=0.34]
\definecolor{fillColor}{RGB}{255,255,255}
\path[use as bounding box,fill=fillColor,fill opacity=0.00] (0,0) rectangle (505.89,505.89);
\begin{scope}
\path[clip] ( 49.20, 61.20) rectangle (480.69,456.69);
\definecolor{drawColor}{RGB}{0,0,0}

\path[draw=drawColor,line width= 0.8pt,line join=round,line cap=round] ( 53.51,106.57) --
	( 57.83,125.77) --
	( 62.14,140.67) --
	( 66.46,153.36) --
	( 70.77,164.63) --
	( 75.09,174.91) --
	( 79.40,184.43) --
	( 83.72,193.34) --
	( 88.03,201.77) --
	( 92.35,209.78) --
	( 96.66,217.44) --
	(100.98,224.79) --
	(105.29,231.87) --
	(109.61,238.70) --
	(113.92,245.31) --
	(118.24,251.71) --
	(122.55,257.93) --
	(126.87,263.97) --
	(131.18,269.85) --
	(135.50,275.57) --
	(139.81,281.15) --
	(144.13,286.58) --
	(148.44,291.88) --
	(152.76,297.06) --
	(157.07,302.10) --
	(161.39,307.03) --
	(165.70,311.83) --
	(170.02,316.52) --
	(174.33,321.09) --
	(178.65,325.55) --
	(182.96,329.90) --
	(187.28,334.14) --
	(191.59,338.26) --
	(195.91,342.28) --
	(200.22,346.19) --
	(204.54,349.99) --
	(208.85,353.68) --
	(213.17,357.26) --
	(217.48,360.73) --
	(221.80,364.05) --
	(226.11,367.34) --
	(230.43,370.46) --
	(234.74,373.51) --
	(239.06,376.42) --
	(243.37,379.22) --
	(247.69,381.90) --
	(252.00,384.44) --
	(256.32,386.92) --
	(260.63,389.25) --
	(264.94,391.46) --
	(269.26,393.55) --
	(273.57,395.52) --
	(277.89,397.36) --
	(282.20,399.09) --
	(286.52,400.69) --
	(290.83,402.16) --
	(295.15,403.52) --
	(299.46,404.75) --
	(303.78,405.86) --
	(308.09,406.85) --
	(312.41,407.72) --
	(316.72,408.46) --
	(321.04,409.07) --
	(325.35,409.57) --
	(329.67,409.94) --
	(333.98,410.19) --
	(338.30,410.34) --
	(342.61,410.34) --
	(346.93,410.18) --
	(351.24,410.02) --
	(355.56,409.67) --
	(359.87,409.23) --
	(364.19,408.66) --
	(368.50,407.98) --
	(372.82,407.20) --
	(377.13,406.31) --
	(381.45,405.30) --
	(385.76,404.19) --
	(390.08,402.99) --
	(394.39,401.69) --
	(398.71,400.30) --
	(403.02,398.83) --
	(407.34,397.26) --
	(411.65,395.62) --
	(415.97,393.90) --
	(420.28,392.10) --
	(424.60,390.24) --
	(428.91,388.32) --
	(433.23,386.40) --
	(437.54,384.34) --
	(441.86,382.28) --
	(446.17,380.20) --
	(450.49,378.09) --
	(454.80,375.98) --
	(459.12,373.87) --
	(463.43,371.89) --
	(467.75,369.50) --
	(472.06,367.53) --
	(476.38,366.02);
\end{scope}
\begin{scope}
\path[clip] (  0.00,  0.00) rectangle (505.89,505.89);
\definecolor{drawColor}{RGB}{0,0,0}

\node[text=drawColor,anchor=base,inner sep=0pt, outer sep=0pt, scale=  1.00] at (264.94, 5.60) {$m$};
\end{scope}
\begin{scope}
\path[clip] ( 49.20, 61.20) rectangle (480.69,456.69);
\definecolor{drawColor}{RGB}{0,0,0}

\path[draw=drawColor,line width= 0.4pt,dash pattern=on 4pt off 4pt ,line join=round,line cap=round] ( 49.20, 61.20) --
	(291.78,505.89);

\path[draw=drawColor,line width= 0.8pt,line join=round,line cap=round] ( 49.20, 62.62) --
	( 49.64, 75.48) --
	( 50.07, 81.38) --
	( 50.50, 85.93) --
	( 50.93, 89.78) --
	( 51.36, 93.18) --
	( 51.79, 96.26) --
	( 52.22, 99.10) --
	( 52.66,101.75) --
	( 53.09,104.24) --
	( 53.52,106.60);
\end{scope}
\begin{scope}
\path[clip] (  0.00,  0.00) rectangle (505.89,505.89);
\definecolor{drawColor}{RGB}{0,0,0}

\path[draw=drawColor,line width= 0.4pt,line join=round,line cap=round] ( 49.20, 61.20) -- (480.69, 61.20);

\path[draw=drawColor,line width= 0.4pt,line join=round,line cap=round] ( 49.20, 61.20) -- ( 49.20, 65.15);

\path[draw=drawColor,line width= 0.4pt,line join=round,line cap=round] (208.67, 61.20) -- (208.67, 65.15);

\path[draw=drawColor,line width= 0.4pt,line join=round,line cap=round] (480.69, 61.20) -- (480.69, 65.15);

\node[text=drawColor,anchor=base,inner sep=0pt, outer sep=0pt, scale=  1.00] at ( 49.20, 30.60) {0};

\node[text=drawColor,anchor=base,inner sep=0pt, outer sep=0pt, scale=  1.00] at (208.67, 30.60) {0.37};

\node[text=drawColor,anchor=base,inner sep=0pt, outer sep=0pt, scale=  1.00] at (480.69, 30.60) {1};

\path[draw=drawColor,line width= 0.4pt,line join=round,line cap=round] ( 49.20, 61.20) -- ( 49.20,456.69);

\path[draw=drawColor,line width= 0.4pt,line join=round,line cap=round] ( 49.20, 61.20) -- ( 53.15, 61.20);

\path[draw=drawColor,line width= 0.4pt,line join=round,line cap=round] ( 49.20,353.52) -- ( 53.15,353.52);

\path[draw=drawColor,line width= 0.4pt,line join=round,line cap=round] ( 49.20,456.69) -- ( 53.15,456.69);


\node[text=drawColor,rotate= 90.00,anchor=base,inner sep=0pt, outer sep=0pt, scale=  1.00] at ( 34.80,353.52) {0.37};

\node[text=drawColor,rotate= 90.00,anchor=base,inner sep=0pt, outer sep=0pt, scale=  1.00] at ( 34.80,456.69) {0.5};
\end{scope}
\begin{scope}
\path[clip] ( 49.20, 61.20) rectangle (480.69,456.69);
\definecolor{drawColor}{RGB}{0,0,0}




\node[right, text=drawColor,anchor=base west,inner sep=0pt, outer sep=0pt, scale=  1.00] at (213.67, 338.52) {$m^\star$};

\path[draw=drawColor,line width= 0.8pt,line join=round,line cap=round] (356.88, 120.20) -- (384.88, 120.20);
\node[text=drawColor,anchor=base west,inner sep=0pt, outer sep=0pt, scale=  1.00] at (400.88, 111.76) {$\varphi^\AGG$};

\path[draw=drawColor,line width= 0.4pt,dash pattern=on 4pt off 4pt ,line join=round,line cap=round] (356.88, 78.20) -- (384.88, 78.20);
\node[text=drawColor,anchor=base west,inner sep=0pt, outer sep=0pt, scale=  1] at (403.88, 69.76) {id};

\end{scope}
\end{tikzpicture}
     }\qquad
     \subfloat[Equilibrium bidding function]{\label{fig:beta}
\begin{tikzpicture}[x=1pt,y=1pt, scale = 0.34]
\definecolor{fillColor}{RGB}{255,255,255}
\path[use as bounding box,fill=fillColor,fill opacity=0.00] (0,0) rectangle (505.89,505.89);
\begin{scope}
\path[clip] ( 49.20, 61.20) rectangle (480.69,456.69);
\definecolor{drawColor}{RGB}{0,0,0}

\path[draw=drawColor,line width= 0.8pt,line join=round,line cap=round] ( 53.51, 65.16) --
	( 57.83, 69.15) --
	( 62.14, 73.15) --
	( 66.46, 77.17) --
	( 70.77, 81.22) --
	( 75.09, 85.28) --
	( 79.40, 89.42) --
	( 83.72, 93.53) --
	( 88.03, 97.66) --
	( 92.35,101.83) --
	( 96.66,106.02) --
	(100.98,110.24) --
	(105.29,114.49) --
	(109.61,118.77) --
	(113.92,123.07) --
	(118.24,127.40) --
	(122.55,131.77) --
	(126.87,136.16) --
	(131.18,140.59) --
	(135.50,145.05) --
	(139.81,149.55) --
	(144.13,154.08) --
	(148.44,158.65) --
	(152.76,163.26) --
	(157.07,167.91) --
	(161.39,172.60) --
	(165.70,177.33) --
	(170.02,182.12) --
	(174.33,186.95) --
	(178.65,191.82) --
	(182.96,196.76) --
	(187.28,201.74) --
	(191.59,206.79) --
	(195.91,211.89) --
	(200.22,217.06) --
	(204.54,222.30) --
	(208.85,227.60) --
	(213.17,232.99) --
	(217.48,238.46) --
	(221.80,244.02) --
	(226.11,249.68) --
	(230.43,255.43) --
	(234.74,261.30) --
	(239.06,267.29) --
	(243.37,273.41) --
	(247.69,279.68) --
	(252.00,286.11) --
	(256.32,292.72) --
	(260.63,299.54) --
	(264.94,306.60) --
	(269.26,313.94) --
	(273.57,321.61) --
	(277.89,329.67) --
	(282.20,338.28) --
	(286.52,347.56) --
	(290.83,353.68) --
	(295.15,354.61) --
	(299.46,356.04) --
	(303.78,357.80) --
	(308.09,359.75) --
	(312.41,361.83) --
	(316.72,364.00) --
	(321.04,366.24) --
	(325.35,368.52) --
	(329.67,370.83) --
	(333.98,373.18) --
	(338.30,375.55) --
	(342.61,377.94) --
	(346.93,380.34) --
	(351.24,382.74) --
	(355.56,385.16) --
	(359.87,387.57) --
	(364.19,389.99) --
	(368.50,392.41) --
	(372.82,394.81) --
	(377.13,397.24) --
	(381.45,399.65) --
	(385.76,402.06) --
	(390.08,404.46) --
	(394.39,406.86) --
	(398.71,409.25) --
	(403.02,411.64) --
	(407.34,414.02) --
	(411.65,416.39) --
	(415.97,418.75) --
	(420.28,421.11) --
	(424.60,423.46) --
	(428.91,425.80) --
	(433.23,428.13) --
	(437.54,430.46) --
	(441.86,432.77) --
	(446.17,435.08) --
	(450.49,437.38) --
	(454.80,439.67) --
	(459.12,441.96) --
	(463.43,444.23) --
	(467.75,446.50) --
	(472.06,448.76) --
	(476.38,451.01);
\end{scope}
\begin{scope}
\path[clip] (  0.00,  0.00) rectangle (505.89,505.89);
\definecolor{drawColor}{RGB}{0,0,0}

\node[text=drawColor,anchor=base,inner sep=0pt, outer sep=0pt, scale=  1.00] at (264.94, 5.60) {$v$};
\end{scope}
\begin{scope}
\path[clip] ( 49.20, 61.20) rectangle (480.69,456.69);
\definecolor{drawColor}{RGB}{0,0,0}

\path[draw=drawColor,line width= 0.4pt,dash pattern=on 4pt off 4pt ,line join=round,line cap=round] ( 49.20, 61.20) --
	(480.69,456.69);
\end{scope}
\begin{scope}
\path[clip] (  0.00,  0.00) rectangle (505.89,505.89);
\definecolor{drawColor}{RGB}{0,0,0}

\path[draw=drawColor,line width= 0.4pt,line join=round,line cap=round] ( 49.20, 61.20) -- (480.69, 61.20);

\path[draw=drawColor,line width= 0.4pt,line join=round,line cap=round] ( 49.20, 61.20) -- ( 49.20, 65.15);


\path[draw=drawColor,line width= 0.4pt,line join=round,line cap=round] (480.69, 61.20) -- (480.69, 65.15);

\node[text=drawColor,anchor=base,inner sep=0pt, outer sep=0pt, scale=  1.00] at ( 49.20, 30.60) {0};


\node[text=drawColor,anchor=base,inner sep=0pt, outer sep=0pt, scale=  1.00] at (480.69, 30.60) {1};

\path[draw=drawColor,line width= 0.4pt,line join=round,line cap=round] ( 49.20, 61.20) -- ( 49.20,456.69);

\path[draw=drawColor,line width= 0.4pt,line join=round,line cap=round] ( 49.20, 61.20) -- ( 53.15, 61.20);

\path[draw=drawColor,line width= 0.4pt,line join=round,line cap=round] ( 49.20,353.52) -- ( 53.15,353.52);

\path[draw=drawColor,line width= 0.4pt,line join=round,line cap=round] ( 49.20,456.69) -- ( 53.15,456.69);


\node[text=drawColor,rotate= 90.00,anchor=base,inner sep=0pt, outer sep=0pt, scale=  1.00] at ( 34.80,353.52) {0.37};

\node[text=drawColor,rotate= 90.00,anchor=base,inner sep=0pt, outer sep=0pt, scale=  1.00] at ( 34.80,456.69) {0.5};
\end{scope}
\begin{scope}
\path[clip] ( 49.20, 61.20) rectangle (480.69,456.69);
\definecolor{drawColor}{RGB}{0,0,0}

\path[draw=drawColor,line width= 0.8pt,line join=round,line cap=round] ( 49.20, 61.20) --
	( 53.51, 65.15) --
	( 57.83, 69.11);

\path[draw=drawColor,line width= 0.8pt,line join=round,line cap=round] (176.88, 120.20) -- (204.88, 120.20);
\node[text=drawColor,anchor=base west,inner sep=0pt, outer sep=0pt, scale=  1.00] at (210.88, 111.76) {$\beta^\AGG(v|l^\star,m^\star,u^\star)$};

\path[draw=drawColor,line width= 0.4pt,dash pattern=on 4pt off 4pt ,line join=round,line cap=round] (176.88, 78.20) -- (204.88, 78.20);
\node[text=drawColor,anchor=base west,inner sep=0pt, outer sep=0pt, scale=  1] at (213.88, 69.76) {BNE};

\end{scope}
\end{tikzpicture}
     }
     \caption{Aggregate moment equilibrium with uniform values and two bidders}
     \label{fig:consistent-moment-beliefs}
     \begin{minipage}{.85\textwidth}
    \footnotesize
    \emph{Notes: The equilibrium aggregate moment belief $m^{\star}$ is at the intersection of $\varphi(m)$, which gives the expected winning bid when players have the belief $m$, and the identity function $\text{id}(m)$.}
    \end{minipage}
\end{figure}

Figure~\ref{fig:var_phi} reveals that there is also a pooling equilibrium at $l=m=u=\underline v$. This pooling equilibrium exists for all value distributions. It exists for technical reasons. Suppose a bidder thinks that all other bidders bid $l$ with certainty. There is no uncertainty and the (supremum) best response is $l$, confirming the belief that everyone bids $l$. The equilibrium is unstable in the sense that marginally different beliefs ($l<m$) lead to bids above $l$. We do not consider the pooling equilibrium to be of economic interest as, typically, there is variation in the bids.

\section{Individual Moment Beliefs}
\label{sec:individual_theory}

The model of the previous sections was slightly unconventional from a modeling perspective as players formed beliefs about an aggregate statistic: the expected winning bid. This section presents an alternative model that follows the game-theoretic convention of players forming beliefs about the other players' individual behavior. The beliefs take the form of statistics again. Note that both types of belief may come from what is observable, or, as mentioned in the introduction, what can be learned from a small sample of past (winning) bids.\footnote{For simplicity, we confine ourselves to bidders learning either the average winning bid or the average opponent bid. In a sample of losing and winning bids, one can estimate both statistics. We leave this case and the one in which bidders also learn higher moments for future research.} 

\subsection{Individual Moment Beliefs and Equilibrium}

Bidder $i$ now forms beliefs directly about the behavior of bidder $j$; there is no detour via the distribution of the winning bid. Bidder $i$'s belief about bidder $j$'s behavior is that the bid distribution of bidder $j$ belongs to the set $\mathcal B_{ij}$, where $\mathcal B_{ij}\subseteq \mathcal P$. The set $\mathcal B_{i}=\times_{j\neq i}\mathcal B_{ij}$ is the set of possible competing bid distributions or the belief set.

Bidder $i$ believes that bidder $j$'s bid is $\mu_i$ on average. We refer to $\mu_i$ as an (individual) moment belief. As in the previous section, bidder $i$ believes that bidder $j$ does not bid below $l_i$ and does not bid above $u_i$. Thus, bidder $i$ believes that the bid of bidder $j$ is drawn from a distribution in the set $\mathcal B_{ij}(l_i,\mu_i,u_i)$, where this set is
\begin{equation*}
	\mathcal B_{ij}(l_i,\mu_i,u_i) = \{B\in\mathcal P:\int_{l_i}^{u_i}dB=1 \text{ and } \int xdB(x)=\mu_i\}.
\end{equation*}
Note that the expected bids of the other bidders are identical. Their bids do not necessarily come from the same distribution.

The new model requires adapting the definition of a moment equilibrium.

\begin{definition}
	\label{def:individual_moment_equilibrium}
	The tuple $(\beta,l,\mu,u)$ is a \emph{moment equilibrium} if
	\begin{enumerate}
	 	\item $\beta_i$ minimizes the maximal loss for player $i\in N$ and value $v_i\in\text{supp}(F)$, i.e.,
	 	\begin{equation}
	 		\beta_i(v_i|l_i,\mu_i,u_i) \in\arg\inf_{b_i\in\mathbb R_+}\sup_{B_{-i}\in\mathcal B_i(l_i,\mu_i,u_i)} \lambda_i(b_i,B_{-i}|v_i);
	 	\end{equation}
	 	\item The moment beliefs are consistent with $\beta$ and $F$, that is,
	 	\begin{equation}
	 		\int \beta_i(v_i|l_i,\mu_i,u_i) d F(v_i) = \mu_j
	 		\label{eq:individual_moment-equilibrium}
	 	\end{equation}
	 	for all $i\in N$ and $j\neq i$;
	 	\item The range beliefs are consistent, i.e., for all $i\in N$ and $j\neq i$
	 	\begin{align*}
	 		\inf_{v\in\text{supp}(F)}\beta_i(v|l_i,\mu_i,u_i)&=l_j,\text{ and}\\
	 		\sup_{v\in\text{supp}(F)}\beta_i(v|l_i,\mu_i,u_i)&=u_j.\\
	 	\end{align*}
	 \end{enumerate} 
\end{definition}
In an (individual) moment equilibrium, the same kind of consistency requirements as in an aggregate moment equilibrium are met. The range beliefs are consistent with optimal behavior and the value distribution. Moreover, the expected optimal bid of a player equals the other bidders' beliefs about the expected bid. 

Bayesian priors can be thought of as an infinite sequence of moment beliefs that pin down a unique distribution. A Bayes-Nash equilibrium then requires an infinite number of moment beliefs to be consistent. In a moment equilibrium we only require a finite number of moment beliefs to be consistent.

\subsection{Optimal Bidding Function}

We again start the analysis by characterizing the minimax bidding function. The bidders face the same type of decision problem as above. They engage in worst-case reasoning to find the bid that performs uniformly well in the sense of worst-case loss for all possible bid distributions. Loss again comes from bidding too high or low relative to the best response bid. The bid that minimizes maximal loss equates the maximal loss when bidding too high with the maximal loss when bidding too low.

The difference between individual and aggregate moment beliefs is that bidders now do not need to infer something from the distribution of the highest bid of $n$ bidders about the highest bid of $n-1$ bidders. Bidders now form beliefs directly about the distribution of each other bidder's bid distribution and can use these beliefs to form a belief about the distribution of the highest bid of $n-1$ opponent bidders.

It is possible to state the minimax bidding function in closed form when $n=2$. We define the following two cutoff values
\begin{equation*}
	\hat v_1 = \frac{u(\mu-l)+\mu(u-\mu)}{u-l}\, \text{ and } \hat v_2 = \frac{u(u-l) - l(u-\mu)}{\mu-l} 
\end{equation*}
for the two-bidder case. 

\begin{proposition}
	Suppose bidder $i$ believes that the other bidders' bids are independent draws from bid distributions with lowest bid $l$ and highest bid $u$, $l\le \underline v$. Moreover, each bidder $j$'s bid distribution is such that the expected bid is $\mu$ in auctions with $n$ bidders. 

	Let $n= 2$. The minimax bidding function is
	\begin{equation}
		\beta^\IND(v_i|l,\mu,u):=
		\begin{cases}
			u-\sqrt{(u-l)(u-v_i)}&\text{if } v_i < \hat v_1\\
			l+\frac{-(\mu-l) (u-l)+\sqrt{(\mu-l) (u-l) ((4 v-3 l) (u-\mu)+u (\mu-l)-l (u-l))}}{2 (u-\mu)} & \text{if }\hat v_1\le v_i< \hat v_2\\
			u&\text{if } v_i\ge \hat v_2.
		\end{cases}
		\label{eq:minimax-bid-mean-and-range-n-2}
	\end{equation}

	Let $n\ge 3$. Let $\tilde x_1(b)=\max\{l,\min\{\mu,\frac{(n-1)b-u}{n-2}\}\}$. If $v_i\le \mu + (\mu-\tilde x_1(\mu))\left(\frac{u-\mu}{u-\tilde x_1(\mu)} \right)^{n-1}$, let $b_i$ be the unique solution of  
	\begin{equation}
		\left(\frac{u-\mu}{u-\tilde x_1(b_i)} \right)^{n-1}(b_i-\tilde x_1(b_i)) = \left(\frac{u-\mu}{u-b_i} \right)^{n-1}(v_i - b_i).
		\label{eq:ind_minimax_bid_low_value}
	\end{equation}
	For higher values, let $b_i$ be defined as the unique solution of 
	\begin{equation}
		\left(\frac{u-\mu}{u-\tilde x_1(b_i)} \right)^{n-1}(b_i-\tilde x_1(b_i)) =(v_i - b_i)\left(1 - \left(\frac{b_i-\mu}{b_i-l} \right)^{n-1}\right)
		\label{eq:ind_minimax_bid_high_value}
	\end{equation}
	The optimal bid of bidder $i$ is $\beta^\IND(v_i|l,\mu,u):=\min\{b_i,u\}$. The optimal bidding function $\beta^\IND$ is continuous in $v$, $l$, $\mu$, and $u$.
	\label{prop:individual-minimax-bdding-function}
\end{proposition}

The first step of the proof shows that the worst-case bid distributions are identical across the competing bidders. This was imposed in the previous section to infer something from the expected winning bid about the individual bid distributions. Now the bid distributions could be different (up to the range and the first moment), but they are identical in the worst case. Loss is maximized by certain extreme cases, making possible asymmetries irrelevant.

The second step shows that the worst-case bid distributions have at most two elements in the support. The result holds more generally. A moment belief is the expectation of a measurable function. In the current model, the measurable function is $g(b_j)=b_j$ and the moment belief is $\mu_i$. Let $M_{ij}$ be the number of moment beliefs that bidder $i$ has about bidder $j$'s bid distribution. The worst-case bid distribution has at most $M_{ij}+1$ elements in the support \citep{Winkler_1988}. 

Let $x_1$ and $x_2$ denote the two bids in the support, where $l\le x_1\le\mu\le x_1\le u$. The moment constraint then implies that $\pi\cdot x_1+(1-\pi)x_2=\mu$, so the moment constraint pins down the mass $\pi$ on $x_1$. Bidder $i$'s payoff depends on the distribution of the highest bid of $n-1$ independent bids. The probability that the highest bid of the competing bidders is $x_1$ is then $\pi^{n-1}$. Worst-case loss is now convex in the probability, whereas it was concave in the previous section.

The minimax bid again equalizes the maximal loss from bidding too high with the maximal loss from bidding too low. The implicit equations \eqref{eq:ind_minimax_bid_low_value} and \eqref{eq:ind_minimax_bid_high_value} state the maximal loss of bidding too high on the left-hand side and the maximal loss of bidding too low on the right-hand side. 

The maximal loss of bidding too high is more elaborate than in the previous section. Depending on $b$, it can be maximized by any $x_1\in[l,\mu]$. Why is the worst case of bidding too high more involved with individual moment beliefs than with aggregate moment beliefs? In both cases, maximal loss takes the form $\lambda^H=\pi^{\alpha}(b-x_1)$, where $\pi=(x_2-m)/(x_2-x_1)$ or $\pi=(x_2-\mu)/(x_2-x_1)$ and $\alpha\in\{\frac{n-1}{n},n-1 \}$. Loss trades off the probability with which $b$ becomes winning and the ex-post loss $b-x_1$ conditional on winning. A higher $x_1$ raises the probability while it decreases the difference. If $\alpha=(n-1)/n$, then loss is concave in the probability. This means that raising $\pi$ to the power of $\alpha$ ``boosts'' the probability as $\pi^\alpha>\pi$ for $\alpha<1$. This is not the case when $\alpha = n-1 \ge 1$. Then raising $\pi$ to the power of $\alpha$ depresses the probability, which can then be increased by raising $x_1$. In the latter case, there is a non-trivial interaction between $\pi$ and $b-x_1$. In the former case, loss is always maximized by putting as much mass as possible on $l$, as it anyway leads to a high probability.

The maximal loss of bidding too low is simpler as the worst case always puts a mass point slightly above $b$. In contrast to the case of aggregate moment beliefs, the worst case when bidding the minimax bid never puts all the mass on $\mu$. The proof shows that this can be the worst case for bids that are not the minimax bids.

\begin{proposition}
	\label{prop:comparative_statics_individual_moment_belief}
	Let $v\in[\underline v,\overline v]$ such that the minimax bid $\beta^\IND(v|l,\mu,u)\in(l,u)$. The minimax bidding function increases in $v$, $\mu$, and $n$. 
\end{proposition}

The comparative statics are again as one would expect. The minimax bid is strictly increasing in $v$ and $n.$

\subsection{Equilibrium Existence}

A separating moment equilibrium exists also in the case of individual moment beliefs for any value distribution.

\begin{theorem}
	\label{theorem:individual_moment_equilibrium}
	 There is a moment equilibrium $(\beta^{\IND},l^{\star},\mu^{\star},u^{\star})$ for any value distribution $F$ such that the bidding function $\beta^\IND(v|l^{\star},\mu^{\star},u^{\star})$ is strictly increasing in $v$.
\end{theorem}

The unique consistent lower bound belief is again $l^\star = \underline v$. In terms of the upper bound belief, separation and consistency requires $\bar v = \hat v_2$ if $n=2$ and 
\begin{equation*}
	\bar v = u+ \frac{(u-\mu)(u-l)^{n-1}}{(u-l)^{n-1}-(u-\mu)^{n-1}}
\end{equation*}
if $n\ge 3$. In either case, there is a unique $u^{\star}(\mu)$ for each $l$ and $\mu$ that solves the respective equation. Uniqueness follows from the right-hand side being strictly increasing in $u$. 

The function $\psi^\IND$ returns for each moment belief $\mu$ the expected bid if the range beliefs are consistent and separating, i.e., 
\begin{equation}
	\psi^\IND(\mu)=\int_{\underline v}^{\bar v} \beta^\IND(v|l^{\star},\mu,u^{\star}(\mu))dF(v).
\end{equation}
The consistent moment belief $\mu^{\star}$ is a fixed point of $\psi^\IND$. 

The domain of $\psi^\IND$ is $[\underline v,\bar v]$. If $\mu=\underline v$, then all the mass is on $l^{\star}$. Likewise, if $\mu=\bar v$, then $u^{\star}=\mu=\bar v$ and all the mass is on $u^{\star}$. The codomain of $\psi^\IND$ is also $[\underline v,\bar v]$ as no type bids below $l^{\star}$ and no type bids above value. The function $\psi^\IND$ is continuous by Lebesgue's dominated convergence theorem \citep{Elstrodt}. A fixed point exists according to Brouwer's fixed point theorem. We argue that the fixed point is in the interior so that the moment equilibrium is separating. If $\mu=\underline v$, then the minimax bid of all types is $\underline v$. If $\mu$ is marginally higher, then almost all values bid strictly higher than $\underline v$. Similarly, if $\mu=\bar v$, then almost all values $v$ bid less than $\mu$. The function $\psi^\IND$ must intersect the identity function at least once in the interior of $[\underline v,\bar v]$.


\begin{example}
	The value distribution is uniform with support $[0,1]$ and there are two players, i.e., $n=2$. The upper belief is $u(\mu)=\sqrt{\mu}$ in a separating equilibrium. Figure~\ref{fig:var_phi_ind} shows the function $\psi(\mu)=\int_{0}^{1} \beta^\IND(v|l=0,\mu,u(\mu))dv$. The concave shape of $\psi$ implies that there is a unique separating moment equilibrium with consistent range beliefs. The equilibrium beliefs are $l^{\star}=0$, $m^{\star} = 0.3$, and $u^{\star} = 0.55$. For comparison, the expected bid in the unique BNE is $0.25$ and the highest bid is $0.5$. Figure~\ref{fig:consistent-moment-beliefs_individual} illustrates the respective equilibrium bidding functions.

\end{example}

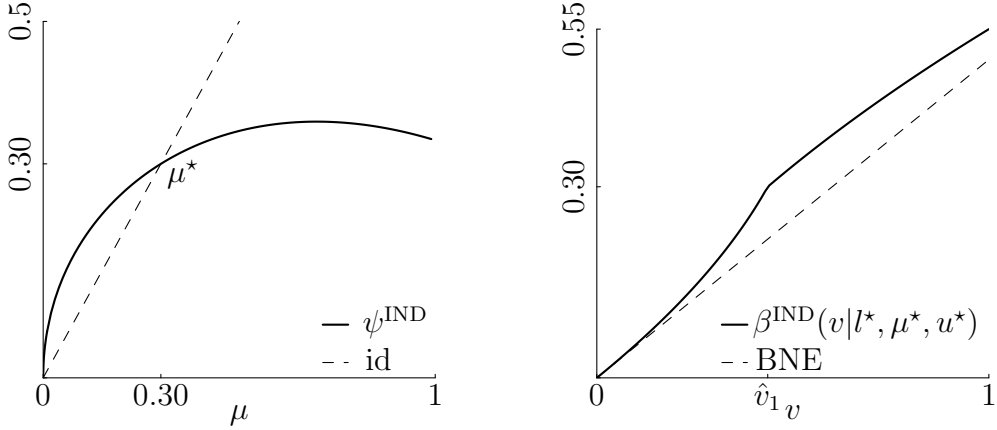
\begin{figure}
     \centering
     \subfloat[Equilibrium moment belief as fixed point]{\label{fig:var_phi_ind}
\begin{tikzpicture}[x=1pt,y=1pt, scale=.34]
\definecolor{fillColor}{RGB}{255,255,255}
\path[use as bounding box,fill=fillColor,fill opacity=0.00] (0,0) rectangle (505.89,505.89);
\begin{scope}
\path[clip] ( 49.20, 61.20) rectangle (480.69,456.69);
\definecolor{drawColor}{RGB}{0,0,0}

\path[draw=drawColor,line width= 0.8pt,line join=round,line cap=round] ( 53.51,112.60) --
	( 57.83,133.08) --
	( 62.14,148.46) --
	( 66.46,161.18) --
	( 70.77,172.19) --
	( 75.09,181.98) --
	( 79.40,190.84) --
	( 83.72,198.96) --
	( 88.03,206.46) --
	( 92.35,213.45) --
	( 96.66,219.99) --
	(100.98,226.15) --
	(105.29,231.96) --
	(109.61,237.46) --
	(113.92,242.69) --
	(118.24,247.66) --
	(122.55,252.40) --
	(126.87,256.92) --
	(131.18,261.25) --
	(135.50,265.32) --
	(139.81,269.35) --
	(144.13,273.15) --
	(148.44,276.80) --
	(152.76,280.30) --
	(157.07,283.65) --
	(161.39,286.89) --
	(165.70,290.00) --
	(170.02,292.98) --
	(174.33,295.85) --
	(178.65,298.61) --
	(182.96,301.26) --
	(187.28,303.81) --
	(191.59,306.26) --
	(195.91,308.62) --
	(200.22,310.88) --
	(204.54,313.05) --
	(208.85,315.13) --
	(213.17,317.13) --
	(217.48,319.05) --
	(221.80,320.88) --
	(226.11,322.64) --
	(230.43,324.32) --
	(234.74,325.92) --
	(239.06,327.46) --
	(243.37,328.92) --
	(247.69,330.31) --
	(252.00,331.63) --
	(256.32,332.89) --
	(260.63,334.07) --
	(264.94,335.20) --
	(269.26,336.26) --
	(273.57,337.26) --
	(277.89,338.19) --
	(282.20,339.07) --
	(286.52,339.88) --
	(290.83,340.64) --
	(295.15,341.34) --
	(299.46,341.99) --
	(303.78,342.57) --
	(308.09,343.10) --
	(312.41,343.58) --
	(316.72,344.00) --
	(321.04,344.37) --
	(325.35,344.69) --
	(329.67,344.95) --
	(333.98,345.16) --
	(338.30,345.33) --
	(342.61,345.44) --
	(346.93,345.50) --
	(351.24,345.51) --
	(355.56,345.47) --
	(359.87,345.39) --
	(364.19,345.25) --
	(368.50,345.07) --
	(372.82,344.84) --
	(377.13,344.57) --
	(381.45,344.25) --
	(385.76,343.88) --
	(390.08,343.47) --
	(394.39,343.01) --
	(398.71,342.51) --
	(403.02,341.97) --
	(407.34,341.37) --
	(411.65,340.74) --
	(415.97,340.06) --
	(420.28,339.34) --
	(424.60,338.58) --
	(428.91,337.77) --
	(433.23,336.92) --
	(437.54,336.03) --
	(441.86,335.10) --
	(446.17,334.12) --
	(450.49,333.11) --
	(454.80,332.05) --
	(459.12,330.95) --
	(463.43,329.80) --
	(467.75,328.63) --
	(472.06,327.42) --
	(476.38,326.16);
\end{scope}
\begin{scope}
\path[clip] (  0.00,  0.00) rectangle (505.89,505.89);
\definecolor{drawColor}{RGB}{0,0,0}

\node[text=drawColor,anchor=base,inner sep=0pt, outer sep=0pt, scale=  1.00] at (264.94, 15.60) {$\mu$};
\end{scope}
\begin{scope}
\path[clip] ( 49.20, 61.20) rectangle (480.69,456.69);
\definecolor{drawColor}{RGB}{0,0,0}

\path[draw=drawColor,line width= 0.4pt,dash pattern=on 4pt off 4pt ,line join=round,line cap=round] ( 49.20, 61.20) --
	(480.69,852.18);

\path[draw=drawColor,line width= 0.8pt,line join=round,line cap=round] ( 49.20, 62.87) --
	( 49.64, 77.83) --
	( 50.07, 84.58) --
	( 50.50, 89.73) --
	( 50.93, 94.06) --
	( 51.36, 97.86) --
	( 51.79,101.28) --
	( 52.22,104.42) --
	( 52.66,107.33) --
	( 53.09,110.05) --
	( 53.52,112.63);
\end{scope}
\begin{scope}
\path[clip] (  0.00,  0.00) rectangle (505.89,505.89);
\definecolor{drawColor}{RGB}{0,0,0}

\path[draw=drawColor,line width= 0.4pt,line join=round,line cap=round] ( 49.20, 61.20) -- (480.69, 61.20);

\path[draw=drawColor,line width= 0.4pt,line join=round,line cap=round] ( 49.20, 61.20) -- ( 49.20, 65.15);

\path[draw=drawColor,line width= 0.4pt,line join=round,line cap=round] (178.74, 61.20) -- (178.74, 65.15);

\path[draw=drawColor,line width= 0.4pt,line join=round,line cap=round] (480.69, 61.20) -- (480.69, 65.15);

\node[text=drawColor,anchor=base,inner sep=0pt, outer sep=0pt, scale=  1.00] at ( 49.20, 30.60) {0};

\node[text=drawColor,anchor=base,inner sep=0pt, outer sep=0pt, scale=  1.00] at (178.74, 30.60) {0.30};

\node[text=drawColor,anchor=base,inner sep=0pt, outer sep=0pt, scale=  1.00] at (480.69, 30.60) {1};

\path[draw=drawColor,line width= 0.4pt,line join=round,line cap=round] ( 49.20, 61.20) -- ( 49.20,456.69);

\path[draw=drawColor,line width= 0.4pt,line join=round,line cap=round] ( 49.20, 61.20) -- ( 53.15, 61.20);

\path[draw=drawColor,line width= 0.4pt,line join=round,line cap=round] ( 49.20,298.67) -- ( 53.15,298.67);

\path[draw=drawColor,line width= 0.4pt,line join=round,line cap=round] ( 49.20,456.69) -- ( 53.15,456.69);


\node[text=drawColor,rotate= 90.00,anchor=base,inner sep=0pt, outer sep=0pt, scale=  1.00] at ( 34.80,298.67) {0.30};

\node[text=drawColor,rotate= 90.00,anchor=base,inner sep=0pt, outer sep=0pt, scale=  1.00] at ( 34.80,456.69) {0.5};
\end{scope}
\begin{scope}
\path[clip] ( 49.20, 61.20) rectangle (480.69,456.69);
\definecolor{drawColor}{RGB}{0,0,0}
\node[right, text=drawColor,anchor=base west,inner sep=0pt, outer sep=0pt, scale=  1.00] at (185.74, 278.67) {$\mu^\star$};

\path[draw=drawColor,line width= 0.8pt,line join=round,line cap=round] (356.88, 120.20) -- (384.88, 120.20);
\node[text=drawColor,anchor=base west,inner sep=0pt, outer sep=0pt, scale=  1.00] at (400.88, 111.76) {$\psi^\IND$};

\path[draw=drawColor,line width= 0.4pt,dash pattern=on 4pt off 4pt ,line join=round,line cap=round] (356.88, 78.20) -- (384.88, 78.20);
\node[text=drawColor,anchor=base west,inner sep=0pt, outer sep=0pt, scale=  1] at (403.88, 69.76) {id};




\end{scope}
\end{tikzpicture}
     }\qquad
     \subfloat[Equilibrium bidding function]{\label{fig:beta_ind}
\begin{tikzpicture}[x=1pt,y=1pt, scale=.34]
\definecolor{fillColor}{RGB}{255,255,255}
\path[use as bounding box,fill=fillColor,fill opacity=0.00] (0,0) rectangle (505.89,505.89);
\begin{scope}
\path[clip] ( 49.20, 61.20) rectangle (480.69,456.69);
\definecolor{drawColor}{RGB}{0,0,0}

\path[draw=drawColor,line width= 0.8pt,line join=round,line cap=round] ( 49.20, 61.20) --
	( 53.51, 64.75) --
	( 57.83, 68.33) --
	( 62.14, 71.94) --
	( 66.46, 75.59) --
	( 70.77, 79.28) --
	( 75.09, 83.00) --
	( 79.40, 86.76) --
	( 83.72, 90.56) --
	( 88.03, 94.41) --
	( 92.35, 98.29) --
	( 96.66,102.22) --
	(100.98,106.19) --
	(105.29,110.21) --
	(109.61,114.28) --
	(113.92,118.39) --
	(118.24,122.56) --
	(122.55,126.79) --
	(126.87,131.07) --
	(131.18,135.41) --
	(135.50,139.81) --
	(139.81,144.27) --
	(144.13,148.80) --
	(148.44,153.40) --
	(152.76,158.07) --
	(157.07,162.82) --
	(161.39,167.65) --
	(165.70,172.57) --
	(170.02,177.57) --
	(174.33,182.67) --
	(178.65,187.87) --
	(182.96,193.17) --
	(187.28,198.59) --
	(191.59,204.12) --
	(195.91,209.79) --
	(200.22,215.59) --
	(204.54,221.54) --
	(208.85,227.65) --
	(213.17,233.94) --
	(217.48,240.42) --
	(221.80,247.10) --
	(226.11,254.02) --
	(230.43,261.19) --
	(234.74,268.64) --
	(239.06,274.73) --
	(243.37,278.41) --
	(247.69,282.06) --
	(252.00,285.68) --
	(256.32,289.28) --
	(260.63,292.85) --
	(264.94,296.38) --
	(269.26,299.90) --
	(273.57,303.38) --
	(277.89,306.85) --
	(282.20,310.28) --
	(286.52,313.70) --
	(290.83,317.08) --
	(295.15,320.45) --
	(299.46,323.79) --
	(303.78,327.11) --
	(308.09,330.41) --
	(312.41,333.69) --
	(316.72,336.95) --
	(321.04,340.18) --
	(325.35,343.40) --
	(329.67,346.59) --
	(333.98,349.77) --
	(338.30,352.92) --
	(342.61,356.06) --
	(346.93,359.18) --
	(351.24,362.28) --
	(355.56,365.37) --
	(359.87,368.43) --
	(364.19,371.48) --
	(368.50,374.51) --
	(372.82,377.53) --
	(377.13,380.53) --
	(381.45,383.51) --
	(385.76,386.47) --
	(390.08,389.43) --
	(394.39,392.36) --
	(398.71,395.28) --
	(403.02,398.19) --
	(407.34,401.08) --
	(411.65,403.95) --
	(415.97,406.82) --
	(420.28,409.66) --
	(424.60,412.50) --
	(428.91,415.32) --
	(433.23,418.12) --
	(437.54,420.92) --
	(441.86,423.70) --
	(446.17,426.47) --
	(450.49,429.22) --
	(454.80,431.96) --
	(459.12,434.69) --
	(463.43,437.41) --
	(467.75,440.12) --
	(472.06,442.81) --
	(476.38,445.49) --
	(480.69,448.16);
\end{scope}
\begin{scope}
\path[clip] (  0.00,  0.00) rectangle (505.89,505.89);
\definecolor{drawColor}{RGB}{0,0,0}

\node[text=drawColor,anchor=base,inner sep=0pt, outer sep=0pt, scale=  1.00] at (264.94, 15.60) {$v$};
\end{scope}
\begin{scope}
\path[clip] (  0.00,  0.00) rectangle (505.89,505.89);
\definecolor{drawColor}{RGB}{0,0,0}

\path[draw=drawColor,line width= 0.4pt,line join=round,line cap=round] ( 49.20, 61.20) -- (480.69, 61.20);

\path[draw=drawColor,line width= 0.4pt,line join=round,line cap=round] ( 49.20, 61.20) -- ( 49.20, 65.15);

\path[draw=drawColor,line width= 0.4pt,line join=round,line cap=round] (237.31, 61.20) -- (237.31, 65.15);

\path[draw=drawColor,line width= 0.4pt,line join=round,line cap=round] (480.69, 61.20) -- (480.69, 65.15);

\node[text=drawColor,anchor=base,inner sep=0pt, outer sep=0pt, scale=  1.00] at ( 49.20, 30.60) {0};

\node[text=drawColor,anchor=base,inner sep=0pt, outer sep=0pt, scale=  1.00] at (237.31, 30.60) {$\hat v_1$};

\node[text=drawColor,anchor=base,inner sep=0pt, outer sep=0pt, scale=  1.00] at (480.69, 30.60) {1};

\path[draw=drawColor,line width= 0.4pt,line join=round,line cap=round] ( 49.20, 61.20) -- ( 49.20,448.16);

\path[draw=drawColor,line width= 0.4pt,line join=round,line cap=round] ( 49.20, 61.20) -- ( 53.15, 61.20);

\path[draw=drawColor,line width= 0.4pt,line join=round,line cap=round] ( 49.20,273.23) -- ( 53.15,273.23);

\path[draw=drawColor,line width= 0.4pt,line join=round,line cap=round] ( 49.20,448.16) -- ( 53.15,448.16);


\node[text=drawColor,rotate= 90.00,anchor=base,inner sep=0pt, outer sep=0pt, scale=  1.00] at ( 34.80,273.23) {0.30};

\node[text=drawColor,rotate= 90.00,anchor=base,inner sep=0pt, outer sep=0pt, scale=  1.00] at ( 34.80,448.16) {0.55};
\end{scope}
\begin{scope}
\path[clip] ( 49.20, 61.20) rectangle (480.69,456.69);
\definecolor{drawColor}{RGB}{0,0,0}

\path[draw=drawColor,line width= 0.4pt,dash pattern=on 4pt off 4pt ,line join=round,line cap=round] ( 49.20, 61.20) --
	(480.69,414.32);
\path[draw=drawColor,line width= 0.8pt,line join=round,line cap=round] (186.88, 120.20) -- (214.88, 120.20);
\node[text=drawColor,anchor=base west,inner sep=0pt, outer sep=0pt, scale=  1.00] at (220.88, 111.76) {$\beta^\IND(v|l^\star,\mu^\star,u^\star)$};

\path[draw=drawColor,line width= 0.4pt,dash pattern=on 4pt off 4pt ,line join=round,line cap=round] (186.88, 78.20) -- (214.88, 78.20);
\node[text=drawColor,anchor=base west,inner sep=0pt, outer sep=0pt, scale=  1] at (223.88, 69.76) {BNE};

\end{scope}
\end{tikzpicture}
     }
     \caption{Moment equilibrium with uniform values and two bidders}
     \label{fig:consistent-moment-beliefs_individual}
     \begin{minipage}{.85\textwidth}
    \footnotesize
    \emph{Notes: The equilibrium moment belief $\mu^{\star}$ is at the intersection of $\varphi(\mu)$, which gives the expected bid when players have the belief $\mu$, and the identity function $\text{id}(\mu)$.}
    \end{minipage}
\end{figure}

A final note is that there is a multiplicity of equilibria with inconsistent range beliefs. For any $l\in [0,\underline v]$ there is a consistent moment belief $\mu^\star(l)$ and an upper bound belief such that the bidding function is strictly increasing.

\section{Identification and Estimation}

This second part of the paper is concerned with the estimation of the latent value distribution $F$ from observed bids. We first provide non-parametric identification results and then non-parametric estimators of the value distribution. The next section applies these to data from highway procurement auctions and assesses the fit in out-of-sample predictions.

\subsection{Non-parametric Identification}

Suppose we observe a distribution $G$ of bids and know that these bids come from auctions with $n$ bidders. Our objective is to infer the value distribution $F$ from the bid distribution $G$. To do so, we need a (game-theoretic) model that informs us about the relation between $F$ and $G$. \cite{GPV} [GPV], the landmark paper on the structural analysis of first-price auctions, assumes that the data is generated by bidders playing a Bayes-Nash equilibrium. We assume in contrast that the bidders play an aggregate moment equilibrium or a moment equilibrium. In any case, the value distribution $F$ is identified if it is uniquely determined by the bid distribution and the assumption that the data is generated by a certain model.

Theorems~\ref{theorem:aggregate_moment_equilibrium} and \ref{theorem:individual_moment_equilibrium} show that there are equilibrium beliefs that induce strictly increasing bidding functions. The bidding functions depend on the bidder's value and on statistics that can be inferred from the bid distribution. The inverse minimax bidding functions then translate the bid distribution into a unique value distribution. Let $v^\AGG$ denote the inverse minimax bidding function when bidders have aggregate moment beliefs, i.e., $v^\AGG(\beta^\AGG(v|l^\star, m^\star, u^\star))=v$. Let $v^\IND$ denote the inverse minimax bidding function when bidders have individual moment beliefs.

The next theorem shows that the value distribution $F$ is non-parametrically identified by the bid distribution if the bid distribution is generated by the bidders playing an aggregate moment equilibrium. \cite{GPV} prove an analogous theorem when BNE is the underlying behavioral model.

\begin{theorem}
	Let $G$ be a distribution with support $[l,u]$ and expected winning bid $m$, i.e., $m=\int b d G^n(b)$. There exists a distribution of bidders' private values $F$ such that $G$ is the bid distribution of a separating aggregate moment equilibrium if and only if $l<m<u$. Moreover, when $F$ exists, it is unique with support $[\underline v,\overline v]$ and satisfies $F(v) = G(\beta^\AGG(v|l,m,u))$ for all $v\in[\underline v,\overline v]$.
	\label{thm:identification_aggregate_meq}
\end{theorem}

The theorem relies on the minimax bidding function $\beta^\AGG$ being strictly increasing in $v$ in a separating moment equilibrium (Proposition~\ref{prop:comparative_statics_aggregate_moment_belief}) because this allows transforming the value distribution to a bid distribution and vice versa. In a separating equilibrium, the unique type bidding $u$ is $\overline v$. The necessary condition for the bidding function to be strictly increasing is then $l<m<u$. 

The value distribution is non-parametrically identified from the observed bid distribution whenever $l<m<u$. This is a minimal prerequisite as it only requires that there is some dispersion in the bids. If all the mass were on one point, then one could not tell whether the data came from a point mass value distribution or a pooling equilibrium.


We highlight a difference to the Bayes-Nash equilibrium-based approach by GPV. In contrast to GPV, we do not require $G$ to be absolutely continuous. This is because a pure moment equilibrium exists for all value distributions whereas there is no guarantee that a pure Bayes-Nash equilibrium exists if the value distribution has mass points. We also do not need $G$ to admit a density as we do not consider first-order conditions.

The identification is non-parametric, which implies that the estimates are robust to possible misspecification of the underlying parametric class of distributions. Clearly, the identification also holds if one makes parametric assumptions.

The following theorem shows that the value distribution $F$ is non-parametrically identified by the bid distribution if one assumes that the bidders play a separating moment equilibrium as in Section~\ref{sec:individual_theory}.

\begin{theorem}
	Let $G$ be a distribution with support $[l,u]$ and expected value $\mu$, i.e., $\mu=\int b d G(b)$. There exists a distribution of bidders' private values $F$ such that $G$ is the bid distribution of a separating moment equilibrium if and only if $l<\mu<u$. Moreover, when $F$ exists, it is unique with support $[\underline v,\bar v]$ and satisfies $F(v) = G(\beta^\IND(v|l,m,u))$ for all $v\in[\underline v,\bar v]$.
	\label{thm:identification_individual_meq}
\end{theorem}
We omit the statement of the proof as it is analogous to the proof of Theorem~\ref{thm:identification_aggregate_meq}.

\subsection{Non-parametric Estimation}
\label{subsec:estimation}

In practice, we observe a sample of the underlying bid distribution $G$. The sample can be structured in the sense of identifying the bids that come from the same auction. A structured sample $H$ takes the form of a $T\times n$ matrix, where $T$ is the number of observed auctions and $n$ is the number of bidders. A typical element (bid) of $H$ is denoted by $h_{tj}$, $1\le t\le T$ and $1\le j\le n$. An unstructured sample $h$ takes the form of a vector of bids. Abusing notation, the sample has length $T$ and we use $h_t$ to denote a typical element, $1\le t\le T$. We do not know which bids were winning or losing, we only assume that the bids in $h$ are i.i.d. draws from $G$. We assume that we see the number of bidders $n$ that generated sample $h$. One can immediately go from a structured sample $H$ to an unstructured sample $h$ by listing all rows of $H$ in one vector. One cannot translate an unstructured sample into a structured sample. 

We use the inverse minimax bidding function to estimate $F$ from the sample. All we need to do is estimate the players' beliefs. An immediate estimator for the lower bound belief $l$ is the minimum bid in the sample, i.e., $\hat l = \min_t h_t$. Similarly, an estimator for $u$ is the highest bid in the sample, that is, $\hat u = \max_t h_t$. \cite{donald-paarsch-2002} show that the sample extreme values are super-consistent estimators. In a structured sample, the estimator for the expected winning bid is the average winning bid, i.e., $\hat m = \frac{1}{T} \sum_{t=1}^T \max_j h_{tj}$. If the data does not identify the winning bids, then one can estimate $m$ by computing the expected value of the highest of $n$ independent draws from $h$. To do so, we compute the probability mass function implied by the sample $h$ and then compute the expectation of the distribution of the highest bid that only takes values in $h$. In the moment equilibrium with individual moment beliefs, one estimates $\mu$ by taking the sample mean, i.e., $\hat \mu = \frac{1}{T}\sum_t h_t$. For the bidders' belief about the number of player $n$, we use the observed number of bidders.

The two estimators are consistent given the data is generated by the respective equilibrium notion as the beliefs $(l,m,u)$ and $(l,\mu,u)$ are estimated consistently. 

The estimation is fully non-parametric. At a high level, there are two differences to BNE-based approaches. The first difference is that the beliefs take the form of statistics in the moment equilibrium approach, whereas they are about entire continuous distributions in the BNE approach. No kernel estimation is necessary in contrast to the BNE-based approach of \cite{GPV}. A second difference is that the beliefs are about the bid distribution in the moment equilibrium approach, so about something that is observed. BNE bases primitive beliefs on the unobserved value distribution.

In the empirical application, we use a total of four estimators (in addition to GPV). The first two estimators are those already discussed. They use the sample minimum and maximum as estimates for $l$ and $u$, respectively, but differ in the underlying game-theoretic model. While the first set of estimators is super-consistent, it may be sensitive to outliers.

The two other estimators follow a different approach to estimating $l$.\footnote{As the application considers procurement auctions, it is actually the highest bid $u$ that receives different treatment. Note that in procurement auctions the highest bid $u$ takes the role of the lowest bid $l$ in standard bidders-as-buyers auctions.} In a parametric setting, \cite{donald-paarsch-2002} use an estimator for the lower bound of the support that is higher than the sample minimum. As we estimate the value distribution non-parametrically, we cannot use their estimator. Instead, we estimate the minimum bid after removing outliers. Specifically, we use the interquartile range so that outliers are bids below $Q_1 - k (Q_3-Q_1)$, where $Q_1$ and $Q_3$ are the lower and upper quartiles, respectively, and $k=1.5$ \citep{tukey}. The estimator for $l$ is then the minimum sample bid that is higher than $Q_1 - k (Q_3-Q_1)$. To minimize the effect of removing outliers, we estimate the belief $\mu$ before removing the outliers. Note that we only change the estimator for $l$. Our rationale is that (potentially) winning bids are probably more meaningful than losing bids. Sample extrema are not sensitive to outliers, which may be obtained by mistakes in the data collection, the treatment of covariates, or just noisy bidding. By focusing on low bids we implicitly assume that these issues matter more at the bottom of the distribution. In particular, noisy bidding arguably matters more for losing than for winning bids.


\section{Structural Analysis of Highway Procurement Auctions}

We now use the moment equilibrium approach to estimate the distribution of private cost in highway procurement auctions. The section first introduces the data and then discusses how we deal with the observed and unobserved auction heterogeneity. We then estimate the cost distribution and assess the estimates in out-of-sample predictions. Our benchmark is the estimator based on the BNE-based approach of \cite{GPV}. Appendix~\ref{app:empirical} shows how the theory has to be adapted for procurement auctions.

\subsection{Data}

We analyze procurement auctions of California's Department of Transportation for highway paving contracts from 1999 to 2005. The data set was originally prepared and analyzed by \citet*{BHT-AER-2014} (BHT).\footnote{I thank \citeauthor*{BHT-AER-2014} for making their data publicly available.}$^{,}$\footnote{The only minor difference is that we correct the number of bidders when the variable giving the number of bidders (NBIDDERS) does not match the number of recorded bids. We adjust the number of bidders to match the number of recorded bids. There are two such auctions.}

The timing of these auctions is that government engineers first prepare information on the project and a cost estimate. The projects typically involve different work and cost items. The government engineers specify the expected quantities for each item, e.g., the tons of asphalt needed. Interested contractors can submit bids for these individual items until a certain set date. The bidder with the lowest total bid wins the contract. We assume that bidders directly submit the total bid. The total cost of undertaking the project is assumed to be privately known to the bidder at the time of bidding. 

In the data, 347 firms submitted a total of 3,661 bids in 819 auctions. There is substantial heterogeneity in these firms. BHT report that the 20 largest firms won 73.4\% of the total contract dollars awarded. There is a sizable difference even within the top 20 contractors, as only two companies participated in more than 10\% of the auctions. More than half of the firms never won a paving contract in the data.

BHT deal with the bidder heterogeneity by classifying bidders into two categories: fringe and non-fringe bidders. A fringe firm is a contractor that won less than 1\% of the value of contracts awarded. Using this definition, 20 firms are non-fringe, while 327 firms are classified as fringe. Similar to BHT, we assume that fringe and non-fringe firms have different cost distributions.

The data contains the bids and many covariates, of which we use the following. The engineer's estimate (ENG) is known to the bidders before bidding and reflects an estimate of a ``fair and reasonable price.'' The variable DIST is the firm's distance to the project's location. Distances are measured in 100 miles. As firms probably have capacity constraints, BHT provide with UTIL a measure of the estimated capacity utilization rate. The capacity utilization rate is the ratio of backlog and capacity, where capacity is the maximum backlog in the sample. Assuming a linear depreciation, backlog is the remaining dollar value of previously won projects in the data. BHT also use RDIST and RUTIL to provide the minimum of the rivals' distance and capacity utilization rate, respectively. 

The level of the bids is affected by factors that are unique to the auction, such as the engineer's estimate. We control for the observed and unobserved auction heterogeneity by homogenizing the bids with a random effects regression. A similar homogenization procedure was used, for example, by \cite{hong-haile-shum}, \cite{Krasnokutskaya-20110-restud}, and \cite{BHT-AER-2014}. 

The theoretical underpinning of the homogenization regression is the following proposition that shows the moment equilibrium scales analogously with affine transformations of the values. \cite{athey-haile-2007-handbook} prove an analogous result for BNE. More formally, and in the framing of bidders-as-buyers auctions, we assume that the bidders' values are an affine function of their private value (signal), i.e., their willingness-to-pay is $\alpha(v;\mathbf y)=p(\mathbf y)v+q(\mathbf y)$, where $\mathbf y$ is the vector of observables and $p$ and $q$ are functions.

\begin{proposition}
	If $\beta^\IND$ is the minimax bidding function for $l_0,\mu_0,u_0$, and values $v$, then $\beta^\IND_\alpha$,
	\begin{equation*}
		\beta^\IND_\alpha(\alpha(v;\mathbf y)|\alpha(l_0;\mathbf y),\alpha(\mu_0;\mathbf y), \alpha(u_0;\mathbf y)) = p(\mathbf y)\beta^\IND(v|l_0,\mu_0,u_0) + q(\mathbf y),
	\end{equation*}
	is the minimax bidding function for $\alpha(l_0;\mathbf y),\alpha(\mu_0;\mathbf y),\alpha(u_0;\mathbf y),$ and  values $\alpha(v;\mathbf y)$.

	Let values and covariates $\mathbf y$ be independent. If $(\beta^\IND,l_0,\mu_0,u_0)$ is a moment equilibrium for valuations $v$ distributed according to distribution $F$, then $(\beta^\IND_\alpha, \alpha(l_0;\mathbf y),$ $\alpha(\mu_0;\mathbf y), \alpha(u_0;\mathbf y))$ is a moment equilibrium for valuations $\alpha(v;\mathbf y)$ distributed according to distribution $F_\alpha$, $F_\alpha(\alpha(v;\mathbf y))=F(v)$ for all $v\in[\underline v,\bar v]$.

	Analogous statements hold for minimax bidding functions with aggregate moment beliefs and aggregate moment equilibria.
	\label{prop:affine}
\end{proposition}

Relying on Proposition~\ref{prop:affine}, we homogenize the data by first running the random effects regression
\begin{align}
	\frac{b_{ij}}{\text{ENG}_{ij}} = \alpha_1 + \beta_i + \sum_{k=2}^9 \alpha_k \chi_{\{n_i=k\}} &+ \alpha_{10} \text{FRINGE}_{ij}+ \alpha_{11} \text{DIST}_{ij} + \alpha_{12}\text{UTIL}_{ij}\nonumber\\& + \alpha_{13}\text{RUTIL}_{ij}+ \alpha_{14}\text{RDIST}_{ij}  + \varepsilon_{ij},\label{eq:homogenization_reg}
\end{align}
where $i$ identifies the auction and $j=1,2,\dots,n_i$ the identity of the bidder in auction $i$. The variable $\chi_{\{n_i=k\}}$ is a dummy variable for the number of bidders. The variable is 0 if $n_i\ge 10$. Note that there are no auctions with one bidder. Intuitively, the coefficient estimates should be positive and decreasing as one would expect that more competition in a procurement auction in the form of more participating firms brings bids down. The variable $\text{FRINGE}$ is a dummy variable that is 1 if firm $j$ is classified as a fringe firm and 0 otherwise. Following \cite{athey-haile-2007-handbook}, the homogenized bid $\hat b_{ij}$ is then $\hat \varepsilon_{ij}+\hat \alpha_1+\hat \alpha_{n_i}$, where $\hat\alpha_{n_i}=0$ for $n_i\ge 10$. 

Table~\ref{tab:re} presents the estimated coefficients. One can see that the effect of the number of bidders on the mean bid is not always as theory would predict. Bids increase when moving from four to five bidders, which is counterintuitive. The structural estimation only uses data with at most seven bidders per auction. 

\begin{table}[ht]
\centering
\caption{Homogenization random effects regression \eqref{eq:homogenization_reg}: coefficient estimates}
\begin{tabular}{rrrrrrrr}
  \toprule
  \toprule
  Intercept & FRINGE & DIST & UTIL & RDIST & RUTIL  & & \\ 
  0.967 	& 0.040  & 0.000& 0.007& 0.000 & -0.042 & &  \\ 
  \midrule
  $n=2$ 	&$n=3$ 	& $n=4$ 	& $n=5$ 	& $n=6$ 	& $n=7$ 	& $n=8$ 	& $n=9$\\
  0.117		&0.071  & 0.058 	& 0.081 	& 0.009 	& 0.012 	& -0.066 	& 0.019 \\ 
  \midrule
  \# of bids&3661&\\
   \bottomrule
   \bottomrule
\end{tabular}
\label{tab:re}
\end{table}

We assume that the costs of fringe and non-fringe bidders are drawn from different distributions. We focus on symmetric auctions for simplicity, and thus divide the data into two data sets. In one data set all bidders in all auctions are fringe firms, while in the other set in all auctions all bidders are non-fringe firms. Table~\ref{tab:summary_homogenized_bids} summarizes the homogenized data. The number of observed bids lies between 10 and 108 for the various categories. For the non-fringe data, there are auctions with two, three, four, and five bidders. For the fringe data, we only use the auctions with at most seven bidders. With the exception of five bidders, the mean homogenized bid decreases in the number of bidders. Next, consider the lowest bids $\hat l$ in the sample. Theory predicts that the lowest bid $l$ decreases as the auction becomes more competitive (higher $n$). The sample does not reflect this property; we take this into account in the estimation. Theory also predicts that the highest bid $u$ equals the highest cost $\overline c$ for any number of bidders ($n \ge 2)$. The sample also exhibits variation in the highest observed bids. Moreover, in auctions with fringe bidders there are, for $n=3$ and $n=5$, bids that are significantly higher than the other maximal bids. As discussed in Section~\ref{subsec:estimation}, we take these outliers into account in the estimation. 


\begin{table}[ht]
\centering
\caption{Summary of the homogenized bids}
\begin{tabular}{llrrrrr}
  \toprule
  \toprule
  && $\hat l$ & $\hat m$ & $\hat \mu$ & $\hat u$ & \# of bids \\
  \cline{3-7}
  \multicolumn{7}{l}{\bfseries{Auctions with exclusively non-fringe bidders}}\vspace{2pt}\\
  &$n = 2$ & 0.93 & 1.02 & 1.07 & 1.29 & 108 \\ 
  &$n = 3$ & 0.85 & 0.96 & 1.05 & 1.38 & 90  \\ 
  &$n = 4$ & 0.90 & 0.93 & 1.02 & 1.26 & 20  \\ 
  &$n = 5$ & 0.94 & 0.94 & 1.04 & 1.19 & 10 \\ 
  \multicolumn{7}{l}{\bfseries{Auctions with exclusively fringe bidders}}\vspace{2pt}\\
  &$n = 2$ & 0.91 & 1.03 & 1.10 & 1.36 & 48 \\ 
  &$n = 3$ & 0.64 & 0.91 & 1.05 & 2.01 & 54 \\ 
  &$n = 4$ & 0.77 & 0.90 & 1.03 & 1.53 & 92 \\ 
  &$n = 5$ & 0.80 & 0.90 & 1.05 & 1.97 & 70 \\ 
  &$n = 6$ & 0.62 & 0.81 & 0.97 & 1.28 & 48 \\ 
  &$n = 7$ & 0.75 & 0.82 & 0.97 & 1.46 & 77 \\ 
  \bottomrule
  \bottomrule
\end{tabular}     
  \begin{minipage}{.85\textwidth}
    \footnotesize
    \emph{Notes: The minimum (maximum) bid is denoted by $\hat l$ ($\hat u$). The average winning bid is $\hat m$. The average bid is $\hat \mu$.}
  \end{minipage}
\label{tab:summary_homogenized_bids}
\end{table}

\subsection{Estimation and Out-of-sample Predictions}



We now come to the estimation of the cost distribution and the out-of-sample assessment of the estimates. We take the following approach. Suppose we want to predict the bid distribution of auctions with $n$ bidders. We take the data from auctions with $n'$ bidders, $n'\neq n$. We estimate the cost distribution for each $n'$. We pool the cost estimates across $n'$ to estimate the cost distribution and compute the counterfactual equilibrium had there been $n$ bidders. We compare the counterfactual equilibrium distribution to the sample with $n$ bidders (without removing outliers). 

We first estimate the cost distribution with each of the four non-parametric estimators based on the moment equilibrium approach. As discussed in Section~\ref{subsec:estimation}, these estimators differ in the underlying game-theoretic model (aggregate or individual moment equilibrium) and the treatment of the range beliefs (with or without removing outliers). Let $(b_t)_{t=1}^{T_n}$ denote the sample of bids from auctions with $n$ bidders, where $T_n$ denotes the sample size. To be clear, the $n$ bidders are all either fringe or non-fringe firms and the sample size differs when removing and not removing outliers. Let $\hat l(n) = \min_{t:1\le t \le T_n} b_t$ denote the sample minimum and $\hat u(n)$ the sample maximum. 

We estimate the moment and range beliefs when we want to predict the bid distribution with $n$ bidders as follows. In the treatment of the auctions with $n'$ bidders, $n'\neq n$, we take the estimates for $m$ and $\mu$, denoted by $\hat m_{n'}$ and $\hat \mu_{n'}$, from Table~\ref{tab:summary_homogenized_bids}. The estimate for $l$ is the lowest $\hat l$ for at most $n'$ bidders. Formally, we estimate $l$ by $\hat l_{n'} = \min_{k\in \{2,3,\dots,n'\}\setminus\{n\}}\hat l(k)$. The idea behind this approach is that this estimate for $l$ is consistent with theory as it decreases in $n'$. The estimate for $u$ is the pooled maximum of the auctions without $n$ bidders, i.e., $\hat u_{n'} = \min_{k\neq n} \hat u(k)$. Here, we do not restrict attention to auctions with at most $n'$ bidders as the theoretical upper bound is the same for any $n$. By pooling across the different auctions we can estimate the support more reliably. In the case of estimating $u$, we distinguish the maximum before and after removing outliers as based on the interquartile range. We compute the interquartile range with the pooled bids, where we discard bids above the cutoff $Q_3 + 1.5 (Q_3 - Q_1)$.

We plug the estimated beliefs into the inverse minimax bidding functions $c^\AGG$ and $c^\IND$ as described in Appendix~\ref{app:empirical}. The appropriate inverse minimax bidding function maps the sample $(b_t)_{t=1}^{T_{n'}}$ into a sample of pseudo costs $(c_t)_{t=1}^{T_{n'}}$, where $c_t = c^\AGG(b_t|\hat l_{n'}, \hat m_{n'}, \hat u_{n'})$ or $c_t = c^\IND(b_t|\hat l_{n'}, \hat m_{n'}, \hat u_{n'})$. 

We use auctions with at most four bidders for auctions in which all bidders are non-fringe firms and auctions with at most seven bidders in auctions in which all bidders are fringe firms. We limit the number of bidders per auction due to the sample size. For example, we only observe two auctions in which there are five non-fringe bidders and no fringe bidders. Let $\bar n^{NF}=4$ and $\bar n^F=7$ be the respective maxima of the number of bidders for auctions with non-fringe and fringe bidders. We estimate the pseudo costs for each $n \in \{2,3,\dots, \bar n^\theta\}$ and $\theta\in \{NF,F\}$. 

We assess the estimates out-of-sample. Let $\theta\in \{NF,F\}$ and $n\in \{2,3,\dots,\bar n^\theta\}$. Suppose we want to predict the bid distribution for auctions with $n$ bidders. We pool the pseudo costs for all $n'\in \{2,3,\dots,\bar n^\theta\}\setminus\{n\}$ to get an estimate of the cost distribution. We then compute the bid distributions predicted by the respective moment equilibrium had there been $n$ bidders. Let $T$ denote the number of pseudo costs in the estimated cost distribution, i.e., $T=\sum_{n'\neq n, \,2 \le n'\le \bar n^\theta} T_{n'}$. Let $(c_t)_{t=1}^T$ denote the pooled pseudo costs. The theory implies that the unique consistent upper bound belief $u^\star$ is the highest possible cost, i.e., $u^\star_{n\theta} = \max c_t$. For the moment equilibrium with beliefs about individual behavior, we find a fixed point of the function
\begin{equation*}
    \psi^\IND(l,\mu)=\left(\beta^\IND(\min c_t|l,\mu,u^\star_{n\theta}), \frac{1}{T} \sum_{t=1}^T\beta^\IND(c_t|l,\mu,u^\star_{n\theta})\right).
\end{equation*}
The function returns for each $(l,\mu)$ the corresponding minimum bid and expected bid given the sample of pseudo costs. Note that the computation of the counterfactual is fully non-parametric; no kernel estimation is used. The fixed point of the function $(l^\star_{n\theta},\mu^\star_{n\theta})$ forms the basis of the out-of-sample prediction because we compute $\hat b_t^\IND = \beta^\IND(c_t|l^\star_{n\theta}, \mu^\star_{n\theta}, u^\star_{n\theta})$ for all $c_t$, $1\le t\le T$. 

In the case of aggregate moment equilibrium, we also use $u^\star_{n\theta}=\max c_t$. To compute the fixed point of the sample analogue of $\varphi$, we need to estimate the distribution of the lowest bid among $n$ bidders. We could use the pseudo cost sample to estimate the density and use it, together with the empirical c.d.f. of the cost distribution, to compute the estimate of the density of the distribution of the lowest of $n$ independent cost draws. Instead, we implemented a fully non-parametric and computationally less demanding approach. We compute the probability mass function of the distribution of the minimum of $n$ independent draws from the discrete distribution of $(c_t)_{t=1}^{T_n}$. Let $p_1 = \mathbb P(C < x)$, $p_2= \mathbb P(C=x)$, and $p_3=\mathbb P(C > x)$. Then, the probability that the minimum of $n$ independent draws is $x$ is
\begin{equation*}
    \mathbb P(\min \{C_1,C_2,\dots,C_n\}=x) = \sum_{j=0}^{n-1} \binom{n}{j} (p_3^j(p_1 + p_2)^{n-j} - (p_2 + p_3)^jp_1^{n-j}).
\end{equation*}
Let $p^\text{min}(x)=\mathbb P(\min \{C_1,C_2,\dots,C_n\}=x)$.
We then find a fixed point $(l^\star_{n\theta},m^\star_{n\theta})$ of the function
\begin{equation*}
		\varphi^\AGG(l,m) = \left(\beta^\AGG(\min c_t|l,m,u^\star_{n\theta}), \frac{1}{T}\sum_{t=1}^T p^\text{min}(c_t)\beta^\AGG(c_t|l,m,u^\star_{n\theta}) \right).
\end{equation*}
The predicted bid distribution is then $(\hat b_t^\AGG)_{t=1}^T$, where $\hat b_t^\AGG = \beta^\AGG(c_t|l^\star_{n\theta}, m^\star_{n\theta}, u^\star_{n\theta})$ for all pseudo costs $c_t$.

The benchmark for the estimation of the cost distribution is the estimator proposed by \cite{GPV}. They assume that bidders play the BNE. Each bid is understood to be the best response to the empirical bid distribution. The estimator depends on the number of bidders, as well as the density and cumulative distribution function of the bid distribution. The estimator for the c.d.f. is the empirical c.d.f. GPV propose a density estimation for the density of the bid distribution. We use the triweight kernel and parameters that GPV use in their Section 2.4. Bidders best respond to the empirical bid distribution $G$, with density $g$, which clearly depends on whether the auction consists of fringe or non-fringe bidders. The first-order condition of the maximization of expected utility leads to
\begin{equation}
	c_{ij} = b_{ij} - \frac{1}{n-1} \frac{1-G(b_{ij})}{g(b_{ij})}.
	\label{eq:GPV}
\end{equation}
One can use the first-order condition to transform a sample of bids into a sample of pseudo costs. For a set of pseudo costs, we estimate the density of the cost distribution with a kernel density estimation. We only apply the BNE estimator to the entire sample. Note that the kernel density estimation removes extreme observations to correct for the bias on the boundaries of the support.

For the BNE, we plug in the empirical c.d.f. $\hat F_{(n)}$ of all estimated pseudo costs from all bids from auctions with $n'\neq n$ bidders. The closed-form solution for a BNE of symmetric procurement auctions is
\begin{equation}
	\beta^{BNE}(c_{ij}) = c_{ij} + \frac{\int_c^{\bar c}(1-F(x))^{n-1}dx}{(1-F(c_{ij}))^{n-1}}.
\end{equation}

Table~\ref{tab:summary-out-of-sample-statistics} reports the mean bid and the standard deviation (SD) of the sample bid distribution and the predictions of the various theories and estimators. The table suggests that aggregate moment equilibrium predicts the bid distributions rather accurately in auctions in which the bidders are large firms. In auctions in which the bidders are small firms, it performs less well. 

\begin{table}
  \begin{center}
    \caption{Summary statistics in the out-of-sample predictions}\label{tab:summary-out-of-sample-statistics}

    \begin{tabular}{lrrrrrrrr}
      \toprule
      \toprule
      & Mean & SD && Mean & SD && Mean & SD \\ 
      \multicolumn{9}{l}{\bfseries{Auctions with non-fringe bidders}}\\
      & \multicolumn{2}{c}{$n=2$}&& \multicolumn{2}{c}{$n=3$}&&\multicolumn{2}{c}{$n=4$}\\
      \cline{2-3}\cline{5-6}\cline{8-9}
        ~~Sample      & 1.071 & 0.079 && 1.046 & 0.090 && 1.018 & 0.096 \\ 
        ~~AGG         & 1.062 & 0.073 && 1.059 & 0.088 && 1.046 & 0.097 \\ 
        ~~IND         & 1.106 & 0.074 && 1.027 & 0.096 && 1.012 & 0.087 \\ 
        ~~BNE         & 1.082 & 0.062 && 1.027 & 0.066 && 1.019 & 0.080 \\ 
        ~~AGG Out.    & 1.053 & 0.065 && 1.050 & 0.079 && 1.043 & 0.087 \\ 
        ~~IND Out.    & 1.078 & 0.065 && 1.021 & 0.090 && 1.016 & 0.087 \\ 
      \midrule
      \midrule
      \multicolumn{9}{l}{\bfseries{Auctions with fringe bidders}}\\
      & \multicolumn{2}{c}{$n=2$}&& \multicolumn{2}{c}{$n=3$}&&\multicolumn{2}{c}{$n=4$}\\
      \cline{2-3}\cline{5-6}\cline{8-9}
        ~~Sample      & 1.098 & 0.100 && 1.049 & 0.219 && 1.032 & 0.144 \\ 
        ~~AGG         & 1.111 & 0.121 && 1.066 & 0.128 && 1.066 & 0.145 \\ 
        ~~IND         & 1.241 & 0.143 && 1.067 & 0.156 && 1.018 & 0.158 \\ 
        ~~BNE         & 1.104 & 0.109 && 1.045 & 0.101 && 1.023 & 0.116 \\ 
        ~~AGG Out.    & 1.035 & 0.089 && 1.028 & 0.111 && 1.022 & 0.116 \\ 
        ~~IND Out.    & 1.086 & 0.089 && 1.043 & 0.123 && 1.013 & 0.129 \\ 
        & \multicolumn{2}{c}{$n=5$}&& \multicolumn{2}{c}{$n=6$}&&\multicolumn{2}{c}{$n=7$}\\
        \cline{2-3}\cline{5-6}\cline{8-9}
        ~~Sample      & 1.050 & 0.165 && 0.969 & 0.120 && 0.973 & 0.133 \\ 
        ~~AGG         & 1.050 & 0.145 && 1.058 & 0.155 && 1.054 & 0.161 \\ 
        ~~IND         & 0.993 & 0.145 && 0.998 & 0.150 && 0.993 & 0.148 \\ 
        ~~BNE         & 1.003 & 0.122 && 1.011 & 0.126 && 1.004 & 0.129 \\ 
        ~~AGG Out.    & 1.014 & 0.120 && 1.024 & 0.125 && 1.026 & 0.129 \\ 
        ~~IND Out.    & 0.991 & 0.127 && 0.994 & 0.128 && 0.990 & 0.127 \\ 
        \bottomrule
        \bottomrule
    \end{tabular}
  \end{center}
    \begin{minipage}{1\textwidth}
    \footnotesize
    \emph{Notes: AGG Out. and IND Out. refer to the predictions based on aggregate and individual moment equilibrium after removing outliers in the estimation.}
  \end{minipage}
\end{table}

To better interpret Table~\ref{tab:summary-out-of-sample-statistics}, we compute the difference between the mean and the standard deviation. Specifically, the ``moment distance'' (MD) is the distance between two vectors of the mean and the standard deviation. Let $EQ\in\{\AGG, \IND, \text{BNE}\}$ denote the type of equilibrium. Let $(\hat m^{EQ}_n,\hat{s}^{EQ}_n)$ be the estimated mean and standard deviation of the predicted equilibrium bid distribution and $(\hat m_n,\hat{s}_n)$ the sample mean and standard deviation of the bid distribution with $n$ bidders. The statistic $MD$ is formally defined as
\begin{equation*}
  MD = \sqrt{(\hat m^{EQ}_n-\hat m_n)^2 +(\hat{s}^{EQ}_n-\hat{s}_n)^2}.
\end{equation*}
We also report a second criterion that compares not only the first two moments but the entire cumulative distribution functions. Specifically, let $L^1$ denote the distance between the respective estimated equilibrium bid distribution $\hat G^{EQ}_{n}$ and the empirical bid distribution $\hat G_{n}$ as given by
\begin{equation*}
  L^1 = {\int |\hat G_{n}^{EQ}(x)-\hat G_n(x)| dx}.
\end{equation*}
We also aggregate the two distances across the different $n$. We do so by computing the weighted average, where the weights are proportional to the number of observations in Table~\ref{tab:summary_homogenized_bids}.

\begin{table}[ht]
\centering
\caption{Distances between the sample and predicted bid distribution}\label{tab:metrics}
\begin{tabular}{lrrrrrrrrrrr}
  \toprule
  \toprule
 	& MD & $L^1$ && MD & $L^1$ && MD & $L^1$ &&MD&$L^1$\\ 
    		\multicolumn{9}{l}{\bfseries{Auctions with non-fringe bidders}}\\
  & \multicolumn{2}{c}{$n=2$}&& \multicolumn{2}{c}{$n=3$}&&\multicolumn{2}{c}{$n=4$}&&\multicolumn{2}{r}{Weighted AVG}\\
 		\cline{2-3}\cline{5-6}\cline{8-9}\cline{11-12}
  ~~AGG      & 0.011 & 0.019 && 0.013 & 0.016 && 0.028 & 0.035 && 0.013 & 0.019 \\ 
  ~~IND      & 0.035 & 0.035 && 0.020 & 0.024 && 0.011 & 0.018 && 0.027 & 0.029 \\ 
  ~~BNE      & 0.020 & 0.018 && 0.030 & 0.022 && 0.017 & 0.016 && 0.024 & 0.019 \\ 
  ~~AGG Out. & 0.023 & 0.022 && 0.012 & 0.013 && 0.026 & 0.034 && 0.019 & 0.019 \\ 
  ~~IND Out. & 0.016 & 0.017 && 0.024 & 0.026 && 0.009 & 0.015 && 0.019 & 0.021 \\ 
  \midrule
  \midrule
	\multicolumn{9}{l}{\bfseries{Auctions with fringe bidders}}\\
	  & \multicolumn{2}{c}{$n=2$}&& \multicolumn{2}{c}{$n=3$}&&\multicolumn{2}{c}{$n=4$}&&&\\
	   	\cline{2-3}\cline{5-6}\cline{8-9}
	~~AGG       & 0.025 & 0.033 && 0.093 & 0.064 && 0.034 & 0.045  &&  &  \\ 
  ~~IND       & 0.149 & 0.144 && 0.066 & 0.053 && 0.020 & 0.022  &&  &  \\ 
  ~~BNE	      & 0.011 & 0.016 && 0.118 & 0.059 && 0.029 & 0.021  &&  &  \\ 
  ~~AGG Out.  & 0.064 & 0.064 && 0.110 & 0.048 && 0.029 & 0.025  &&  &  \\ 
  ~~IND Out.  & 0.017 & 0.015 && 0.096 & 0.054 && 0.024 & 0.026  &&  &  \\ 
    & \multicolumn{2}{c}{$n=5$}&& \multicolumn{2}{c}{$n=6$}&&\multicolumn{2}{c}{$n=7$}&&\multicolumn{2}{r}{Weighted AVG}\\
     	\cline{2-3}\cline{5-6}\cline{8-9}\cline{11-12}
  ~~AGG 				& 0.020 & 0.019 && 0.096 & 0.090 && 0.085 & 0.084 && 0.056 & 0.055 \\  
  ~~IND 				& 0.061 & 0.059 && 0.042 & 0.031 && 0.024 & 0.031 && 0.053 & 0.051 \\  
  ~~BNE 				& 0.064 & 0.047 && 0.042 & 0.042 && 0.030 & 0.036 && 0.047 & 0.036 \\  
  ~~AGG Out.    & 0.058 & 0.036 && 0.056 & 0.056 && 0.052 & 0.061 && 0.058 & 0.046 \\  
  ~~IND Out.    & 0.071 & 0.060 && 0.026 & 0.026 && 0.017 & 0.030 && 0.040 & 0.035 \\  
   \bottomrule
   \bottomrule
\end{tabular}

  \begin{minipage}{1\textwidth}
    \footnotesize
    \emph{Notes: AGG Out. and IND Out. refer to the predictions based on aggregate and individual moment equilibrium after removing outliers in the estimation. MD is the Euclidean distance between the sample mean bid and standard deviation and the predicted mean and standard deviation. $L^1$ is the $L^1$ norm between the sample and the predicted bid distribution. Weighted AVG is the average distance weighted by the number of observations in the sample.}
  \end{minipage}
\end{table}

Table~\ref{tab:metrics} reports the two notions of distance: the distance in the first two moments (MD) and the $L^1$ norm. In auctions with non-fringe firms, aggregate moment equilibrium without removing outliers leads to the smallest distances overall. In terms of the (weighted average of the) MD metric, aggregate moment equilibrium performs clearly better than the other theories. There is not much difference between the theories in terms of the $L^1$ norm. Individual moment equilibrium benefits from removing the outliers. In auctions with small firms, individual moment equilibrium after removing the outliers performs best overall. Note that this is not the first empirical study that finds that the behavior of small and large firms can be rationalized by different models of behavior \citep{10.1257/aer.20172015}. With small firms, both types of moment equilibrium benefit from removing outliers. This suggests that it is the small and not the large firms that submit more noisy bids at the top of the distribution. Overall, moment equilibrium performs well in the data and is not outperformed by Bayes-Nash equilibrium.

\section{Conclusion}

The paper proposes a new way of modeling bidders' beliefs in first-price auctions. The idea is that the bidders' beliefs are characterized by statistics of the bid distribution that are typically observable, such as the expected winning bid or the expected bid. The idea that beliefs are characterized by observables can be taken to other settings. For example, firms might not know the distribution of competitors' costs and inventories but might only observe their prices. In a job search, a worker might know the average wage but not the entire wage distribution. 

In our setting, bidders have a belief about the first moment of the distribution of opponents' (aggregate) actions. One can generalize the single belief to a finite number of moments. Intuitively, the model would then nest Bayes-Nash equilibrium, which is as if bidders have an infinite sequence of moment beliefs, as a limiting case. It would be interesting to formalize the relation between moment equilibrium and Bayes-Nash equilibrium for a general class of games. Moreover, assuming that beliefs are characterized by data from past auctions, which moments of the bid distribution would an auctioneer want to reveal?

A characteristic of moment equilibrium is the presence of strategic uncertainty in equilibrium: Bidders cannot perfectly anticipate the distribution of opponents' actions. Bidders consider many bid distributions to be possible in a moment equilibrium and all these distributions reflect the truth as captured by the moment and range beliefs. However, bidders do not know the (unobservable) value distributions and strategies as in traditional Bayesian treatments of the first-price auction. Players directly forming beliefs about the distribution of other players' actions makes the model more parsimonious, based on observables, and eliminates the need to invert strategies. Inverting the bidding function is often mathematically challenging, particularly in models with multidimensional types or multiple goods. Moment equilibrium might pose a tractable alternative. 

We illustrated that (aggregate) moment equilibrium can form the basis of the structural estimation of the latent value distribution. The estimates performed well in out-of-sample predictions in comparison to Bayes-Nash equilibrium. The estimation is simple as only some statistics of observables need to be computed. Hence, moment equilibrium might perform better than Bayes-Nash equilibrium in small samples.



\appendix
\section{Omitted Proofs}

\subsection{Proof of Proposition~\ref{prop:aggregate-minimax-bdding-function}}

Equipped with a willingness-to-pay $v$, bidder $i$'s loss of not knowing the distribution $B$ of the highest bid among $n$ independent bidders while facing $n-1$ of these is
\begin{align}
	\sup_{\tilde b\in\mathbb R_+}\sup_{B\in\Delta\mathbb R_+} & B(\tilde b)^{\frac{n-1}n}(v-\tilde b) - B(b)^{\frac{n-1}n}(v-b) \label{eq:new_loss}\\\notag
	\text{s.t. }&\int x\,dB(x)=m\\
	&\int_l^u \,dB(x)=1.\notag
\end{align}
Note that we have swapped the order of taking the supremum, that is, the ``best response'' $\tilde b$ is chosen and then the bid distribution. The supremum is not affected by the order.

A first observation is that it is never optimal to bid strictly above $u$ or strictly lower than $l$. To see this, note that bidding below $l$ means that one loses with certainty. As the lowest value $\underline v$ is at least $l$, it is weakly better to bid $l$. Similarly, bidding above $u$ means that one bids strictly higher than all bids of the competitors. One can lower the bid to (marginally above) $u$ without decreasing the likelihood of winning. Formally, for all possible bid distributions $B_{-i}$ we have $\lambda(u,B_{-i}|v_i)<\lambda(b,B_{-i}|v_i)$ for $u<b$ as $v_i-u=U_i(u,B_{-i})>U_i(b,B_{-i}) = v_i-b.$

The following lemma shows that the worst-case bid distribution has at most two elements in the support.

\begin{lemma}
	For any fixed pair $(b,\tilde b)$ in equation \eqref{eq:new_loss}, the bid distribution that maximizes loss has at most two elements in the support.
\end{lemma}
\begin{proof}
	Let $\tilde b < b$. We observe that loss is higher when there is no mass in $(\tilde b,b)$. Loss then equals $B(\tilde b)^{\frac{n-1}{n}}(b-\tilde b)$. We distinguish $m\le \tilde b$ and $\tilde b < m$.

	If $m\le \tilde b$, then the worst case puts all the mass on $m$. Loss then equals $b-\tilde b$. The support is a singleton. 

	Let $\tilde b<m$ and let $B$ denote the worst-case bid distribution. Let $\text{supp}(B)$ denote the support of $B$. We show that $\text{supp}(B)=\{\tilde b,u\}$. 

	We first consider the mass above $\tilde b$. Let $y_1$ be the mass $B$ places on $(\tilde b,u)$, and let $m_1$ be the contribution of that interval to the expected value. Formally, let
	\begin{equation*}
		y_1 = \lim_{\epsilon\searrow 0} \int_{\tilde b+\epsilon}^{u-\epsilon}dB = \lim_{\epsilon\searrow 0}B(u-\epsilon) - B(\tilde b+\epsilon)
	\end{equation*}
	and 
	\begin{equation*}
		m_1 = \lim_{\epsilon\searrow 0} \int_{\tilde b+\epsilon}^{u-\epsilon}x dB(x).
	\end{equation*}

	By way of contradiction, suppose $y_1>0$. Then, the mass $y_1$ can be shifted from $(\tilde b,u)$ to $\{\tilde b,u\}$ without changing the expected value while (strictly) increasing $B(\tilde b)$ and (weakly) decreasing $B(b)$. Such a shift would increase loss. To see that such a reallocation of mass is feasible, let $p_1$ denote the mass that is additionally put on $\tilde b$. We choose $p_1$ such that $p_1\cdot \tilde b+(y_1-p_1)\cdot u=m_1$, i.e., $p_1 = (y_1u-m_1)/(u-\tilde b)$. Note that $0<p_1$ if and only if $m_1<y_1u$, which is true whenever $y_1>0$. Moreover, $p_1\le y_1$ if and only if $y_1\tilde b\le m_1$, which is also true. Hence, $B$ cannot be the worst-case distribution. The worst-case bid distribution cannot put mass on $(\tilde b,u)$. The argument also establishes that $u$ is the only element in the support of the worst-case bid distribution $B$ above the average $m$, i.e., $\text{supp}(B)\cap[m,\infty)=\{u\}$.

	We now show that the worst-case bid distribution also cannot put positive mass on $[l,\tilde b)$. We show that if $\text{supp}(B)\cap[l,\tilde b)\neq\emptyset$, then one can decrease the mass on $u$ and increase the mass on $\tilde b$ by removing the mass on $[l,\tilde b)$. There is such a shift in mass that it does not alter the expected value but instead increases loss. Formally, let $y_2$ denote the mass below $\tilde b$, i.e., $y_2 = \lim_{\epsilon\searrow 0}\int_l^{\tilde b-\epsilon}dB$, and let $m_2$ denote the contribution of $m_2$ to the expected value, i.e., $m_2 = \lim_{\epsilon\searrow 0}\int_l^{\tilde b-\epsilon}xdB(x)$. Let $p_u$ denote the mass on $u$. In accordance with the previous paragraph we have 
	\begin{equation*}
	 	m = m_2 + p_uu+(1-y_2-p_u)\tilde b.
	\end{equation*}
	If there was only mass on $\{\tilde b,u\}$, then the distribution puts mass $p$ on $\tilde b$, where $p=(u-m)/(u-\tilde b)$. By way of contradiction, assume that $y_2>0$. The claim is that reducing the mass below $\tilde b$ decreases the mass on $u$, that is, $1-p\le p_u$. The inequality is equivalent to $-\tilde b \le -m_2-(1-y_2)\tilde b$, which is true as $m_2\le y_2\tilde b.$ Hence, $y_2>0$ cannot be the case in the worst-case bid distribution. The worst-case bid distribution has only the two elements $\{\tilde b,u\}$ in its support.

	Let $b<\tilde b$. We distinguish three sub-cases: $b<\tilde b<m$, $b<m<\tilde b$, and $m\le b$. We first consider $b<\tilde b<m$. As the proof is similar to the above arguments, we keep it short. The first insight is that if the distribution puts mass on $(\tilde b,u)$, loss can be made higher while maintaining the moment constraint, and shift the mass from $(\tilde b,u)$ to $\{\tilde b,u\}$. This raises the mass on $\tilde b$ and, therefore, loss. The second insight is that one can then shift the mass below $\tilde b$ to $\{\tilde b,u\}$. This increases the mass on $\tilde b$ and decreases the mass on $u$. Moreover, it decreases the mass below $b$. The only elements in the support of the worst-case bid distribution can be $\{\tilde b,u\}$.

	The second sub-case is characterized by $b<m<\tilde b$. All the mass can be put on $m$. This maximizes loss as it keeps $B(b)$ down to 0 and raises $B(\tilde b)$ to 1.

	The third sub-case has $m\le b<\tilde b$. Now there must be mass below $b$ as there must be mass below $m$. The mass that is required below $m$ to satisfy the moment constraint is lower the lower the elements below $m$ in the support. Hence, we can increase $B(\tilde b)$ and lower $B(b)$ by reallocating the mass from $(l,\tilde b)$ to $\{l,\tilde b\}$. This can be done without changing the expected value. What is left to show is that there is no mass above $\tilde b$ in the worst case. Such a mass could be reallocated to $\{l,\tilde b\}$. This would decrease the mass required on $l$ to satisfy the moment constraint, and therefore decrease $B(b).$ Moreover, it would increase $B(\tilde b)$. 
\end{proof}
	
	Now that we know that the worst-case bid distribution only has two bids $x_1$ and $x_2$ in its support, we can maximize loss for a given bid $b$. The first thing to note is that the ``best response'' $\tilde b$ is either $x_1$ or $x_2$ in the worst case. The nature of the first-price auction ensures that the optimal response to a discrete distribution is always to bid just slightly above a mass point. Note that, strictly speaking, there is no best response. This poses no problem as we operate in the payoff space and consider suprema. The best response $\tilde b$ is, therefore, either $x_1$ or $x_2$. If it was neither of them, then loss would be zero as one does not wish to win the auction.

	With two bids in its support, the distribution of the maximum of $n$ bids is $q\cdot [x_1] + (1-q)\cdot [x_2]$ with $l\le x_1\le m<x_2\le u$. The moment constraint $q\cdot x_1 + (1-q)\cdot x_2=m$ implies
	\begin{equation*}
		q = \frac{x_2 - m}{x_2 - x_1}. 
	\end{equation*}
	A single competing bidder bids $x_1$ with probability $q^{\frac{1}{n}}$ and $x_2$ with probability $1-q^{\frac{1}{n}}$. Bidder $i$ wins if she bids higher than the maximum of $n-1$ draws. The maximum of $n-1$ draws is $x_1$ with probability $q^{\frac{n-1}{n}}$ and $x_2$ with probability $1-q^{\frac{n-1}{n}}$.

	To simplify notation, we define the probability $p$ as
	\begin{equation*}
		p(x_1,x_2) = q^{\frac{n-1}{n}}=\left(\frac{x_2 - m}{x_2 - x_1}\right)^{\frac{n-1}{n}}.
	\end{equation*}
	We will often drop the arguments. The first derivatives are
	\begin{align*}
		\frac{\partial p(x_1,x_2)}{\partial x_1} &= \frac{n-1}{n} \frac{p}{x_2-x_1} ,\text{ and } \\
		\frac{\partial p(x_1,x_2)}{\partial x_2} &= \frac{n-1}{n} p \frac{m-x_1}{(x_2-x_1)(x_2-m)}. 
	\end{align*}
	As both derivatives are positive for the non-degenerate case $x_1<m<x_2$, the probability $p$ is increasing in $x_1$ and $x_2$, respectively.

	We distinguish two primary cases. The first case features $\tilde b<b$ and the second case is characterized by $\tilde b>b$. Note that $\tilde b=b$ does not maximize loss as it leads to a loss of 0.

	The bid $b$ can be too high, i.e., the best response $\tilde b$ is less than $b$. There are the following three exhaustive cases: (i) $x_1<b<x_2$ and $x_1$ is the best response, (ii) $x_2<b$ and $x_1$ is the best response, and (iii) $x_2<b$ and $x_2$ is the best response.

	In the first case ($x_1<b<x_2$ and $x_1$ the best response), loss equals
	\begin{equation*}
		\lambda_1^H(x_1,x_2) = p\cdot(v-x_1) - p\cdot(v-b) = p\cdot (b-x_1).
	\end{equation*}
	Loss is increasing in $x_2$ as the first derivative of loss with respect to $x_2$ is positive. The first derivative of loss with respect to $x_1$ equals
	\begin{equation*}
		\frac{\partial \lambda^H_1}{\partial x_1} = \frac{\partial p}{\partial x_1}(b-x_1)-p = \frac{n-1}{n}p\frac{b-x_1}{x_2-x_1} - p.
	\end{equation*}
	Assuming that $p>0$, this expression is negative if and only if $(n-1)(b-x_1)<n(x_2-x_1)$, which is true as $b<x_2$ by assumption. Note that $p=0$ does not maximize loss. Hence, the worst-case bid distribution has $x_1=l$ and $x_2=u$. Conditional worst-case loss equals
	\begin{equation}
		\bar \lambda^H_1(b) = \left(\frac{u-m}{u-l}\right)^{\frac{n-1}{n}}(b-l).
		\label{eq:loss-high-agg}
	\end{equation}

	In the second case ($x_2<b$ and $x_1$ the best response), loss equals
	\begin{equation}
		\lambda^H_2(x_1,x_2) = p\cdot(v-x_1) - (v-b).
	\end{equation}
	This case does not maximize loss. By placing $x_2$ above $b$ loss is made higher and changed to $\lambda_1^H$. Note that loss $\lambda_2^H$ increases in $x_2$, so $x_2<b$ poses a binding constraint. The constraint can be avoided by placing $x_2$ above $b$, which leads to the first case.

	In the third case ($x_2<b$ and $x_2$ the best response), loss equals
	\begin{equation*}
		\lambda_3^H(x_1,x_2) = b - x_2.
	\end{equation*}
	Loss is maximized by $x_2=m$. We now show that worst-case loss of the first case is higher than $b-m$, i.e., $\bar\lambda^H_1\ge b-m$. We show that
	\begin{equation*}
		b-m\le \frac{u-m}{u-l}(b-l)\le \left(\frac{u-m}{u-l} \right)^{\frac{n-1}{n}}(b-l)=\bar\lambda^H_1.
	\end{equation*}
	The second inequality is true as $q\le q^{\frac{n-1}{n}}$ for $q\in(0,1)$. The first inequality is true for $b\le u$. To summarize, the worst-case loss conditional on $b$ being too high is $\bar\lambda_1^H$.

	The bid $b$ can also be too low, i.e., the best response is higher than $b$. There are again three exhaustive cases to consider: (i) $b<x_1$ and $x_1$ the best response, (ii) $b<x_1$ and $x_2$ the best response, and (iii) $x_1<b<x_2$ and $x_2$ the best response. Note that the first two cases require $b<m$, so only the third applies to the case in which one bids above the average $m$.

	In the first case of bidding too low ($b<x_1$ and $x_1$ the best response), loss equals
	\begin{equation*}
		\lambda^L_1(x_1,x_2) = p\cdot(v-x_1).
	\end{equation*}
	Loss increases in $x_2$, so $x_2 = u$ and $p>0$ in the worst case. The first derivative with respect to $x_1$ equals
	\begin{equation*}
		\frac{\partial\lambda^L_1}{\partial x_1} = \frac{n-1}{n} \frac{p}{x_2-x_1} (v-x_1) - p.
	\end{equation*}
	The first derivative is negative if $x_1 < nx_2 - (n-1)v$ and positive if $x_1>nx_2-(n-1)v$. Hence, loss is convex in $x_1$. Recall that we have $x_2=u$ in the worst case. Loss is decreasing in $x_1$ if $x_1\le m\le n u - (n-1)v$ (low value $v$). Loss is then maximized by $x_1$ as low as possible, that is, by $x_1 = b$. For higher values, the convexity of loss matters. The maximum of the loss is attained on the boundary, so $x_1$ is either $b$ or $m$.

	In the second case ($b<x_1$ and $x_2$ the best response), loss equals $v-x_2$. It is maximized by $x_2=m$. Note that the same loss can be achieved in the first case by $x_1=m$. This is, however, not always optimal. Thus, the second case is subsumed by the first.

	In the third case ($x_1<b<x_2$ and $x_2$ the best response), loss equals
	\begin{equation*}
		\lambda^L_3(x_1,x_2) = v - x_2 - p\cdot(v-b).
	\end{equation*}
	Loss decreases in $x_1$ and $x_2$, respectively. Hence, the worst-case bid distribution has the two mass points $x_1=l$ and $x_2=\max\{b,m\}$.

	To summarize, the worst-case loss conditional on bidding too low is
	\begin{equation*}
	\bar\lambda^L(b)=
		\begin{cases}
			\max\{\left(\frac{u-m}{u-b}\right)^{\frac{n-1}{n}}(v-b),v-m\}&\text{if } b<m\\
			(v - b)\left(1 - \left(\frac{b-m}{b-l} \right)^{\frac{n-1}{n}}\right) &\text{if }m\le b.
		\end{cases}
		\label{eq:loss-low-agg}
	\end{equation*}

	The maximal loss when bidding $b$ is then $\max\{\bar \lambda^H_1(b),\bar \lambda^L(b)\}$. The maximal loss of bidding too high $\lambda^H_1$ is 0 when bidding $l$ and increases in $b$. The maximal loss of bidding too low is strictly positive at $b=l$ and is weakly decreasing in $b$. Thus, maximal loss is minimized by equalizing the two conditional maximal losses. The monotonicity implies that there is a unique solution.


\subsection{Proof of Proposition~\ref{prop:comparative_statics_aggregate_moment_belief}}

	To study the comparative statics, we introduce the three implicit equations 
	\begin{align*}
		F_1(v,b,l,m,u,n) &:= p(l,u)\cdot(b-l) - p(b,u)\cdot (v-b)=0,\\ 
		F_2(v,b,l,m,u,n) &:= p(l,u)\cdot(b-l) - (v-m)=0, \text{ and}\\
		F_3(v,b,l,m,u,n) &:= p(l,u)\cdot(b-l) - (v-b)\cdot(1-p(l,b))=0,
	\end{align*}
	all of which are $\lambda^H-\lambda^L$ for different values $v$. The first two apply if $b<m$ and the last applies if $m\le b$. The roots of these equations are the minimax bids. The implicit function theorem states that
	\begin{equation*}
		\frac{\partial \beta}{\partial x} = - \frac{\partial F_k}{\partial x}\cdot\left( \frac{\partial F_k}{\partial b} \right)^{-1},
	\end{equation*}
	where $k=1,2,3$ defines the relevant region for $v$ and $x\in\{v,l,m,u,n\}$. Note that $\frac{\partial p}{\partial m}<0.$ 

	We have
	\begin{align*}
		\frac{\partial F_1}{\partial v} &= -p(b,u)&<0\\
		\frac{\partial F_1}{\partial b} &= p(l,u) - \frac{\partial p(b,u)}{\partial b}(v-b) + p(b,u) &>0\\
		\frac{\partial F_1}{\partial l} &= \frac{\partial p(l,u)}{\partial l}(b-l) - p(l,u) &<0 \\
		\frac{\partial F_1}{\partial m} &= \frac{\partial p(l,u)}{\partial m} (b-l) - \frac{\partial p(b,u)}{\partial m} (v-b)  &=0\\
		\frac{\partial F_1}{\partial u} &= \frac{\partial p(l,u)}{\partial u}(b-l) - \frac{\partial p(b,u)}{\partial u}(v-b)  &>0.
	\end{align*}

	To see why the derivative with respect to $b$ is positive, recall that loss associated with bidding too low is decreasing in $x_1$. Hence, the partial derivative of the loss of bidding too low with respect to $b$ ($x_1$) is negative.

	We show that $\frac{\partial F_1}{\partial m}=0$. We consider the case in which the minimax bid is found by solving
	\begin{equation*}
		\left(\frac{u-m}{u-l} \right)^{\frac{n-1}{n}}(b-l) - \left(\frac{u-m}{u-b} \right)^{\frac{n-1}{n}}(v-b)  = 0.
	\end{equation*}
	We divide by $(u-m)^{\frac{n-1}{n}}$ and observe that the derivative with respect to $m$ is then 0. 

	We now show that $\frac{\partial F_1}{\partial u} > 0$. Note that for these relatively low values the inverse bid equals
	\begin{equation*}
		v = b + \left(\frac{u-b}{u-l}\right)^{\frac{n-1}{n}}(b-l)
	\end{equation*}
	We plug the inverse minimax bid into the first partial derivative and obtain
	\begin{align*}
		\frac{\partial F_1}{\partial u} &= \frac{\partial p(l,u)}{\partial u}(b-l) - \frac{\partial p(b,u)}{\partial u}(v-b) \\
		&= \frac{n-1}{n}\frac{b-l}{u-m} \cdot \left(p(l,u) \frac{m-l}{u-l} - p(b,u) \frac{m-b}{u-b}\left(\frac{u-b}{u-l}\right)^{\frac{n-1}{n}} \right)\\
		&= \frac{n-1}{n}\frac{b-l}{u-m} \cdot \left(p(l,u) \frac{m-l}{u-l} \left(\frac{u-l}{u-l}\right)^{\frac{n-1}{n}} - p(b,u) \frac{m-b}{u-b}\left(\frac{u-b}{u-l}\right)^{\frac{n-1}{n}} \right).
	\end{align*}
	The sign of the partial derivative is determined by the sign of $h(l)-h(b)$, where $h$ is defined as
	\begin{equation*}
		h(x) := p(x,u) \left(\frac{u-x}{u-l}\right)^{\frac{n-1}{n}}\frac{m-x}{u-x} = \left(\frac{u-m}{u-l} \right)^{\frac{n-1}{n} }\frac{m-x}{u-x}.
	\end{equation*}
	The first derivative of $h$ with respect to $x$ is $h'(x) = \left(\frac{u-m}{u-l} \right)^{\frac{n-1}{n} }\frac{-(u-m)}{(u-x)^2}<0$. We conclude that $h(l)-h(b)>0$, so $\frac{\partial F_1}{\partial u}>0$.

	Note that $\frac{\partial F_1}{\partial l}<0$ as $\lambda^H_1$ decreases in $x_1$. 

	We recall that $p=q^{\frac{n-1}n}$ for taking the partial derivative with respect to $n$, i.e., 
	\begin{align*}
		\frac{\partial F_1}{\partial n} &= \frac{1}{n^2}\cdot p(l,u)\cdot \log(q(l,u)) \cdot (b-l) - \frac{1}{n^2}\cdot p(b,u)\cdot \log(q(b,u))\cdot (v-b)\\
		&= \frac{1}{n^2}\cdot \left(p(l,u)\cdot \log(q(l,u)) \cdot (b-l) - p(b,u)\cdot \log(q(b,u))\cdot (v-b)\right)\\
		&= \frac{1}{n^2}\cdot \left(p(l,u)\cdot \log(q(l,u)) \cdot (b-l) - p(b,u)\cdot \log(q(b,u))\cdot (b-l)\cdot\left(\frac{u-b}{u-l} \right)^{\frac{n-1}n}\right)\\
		&= \frac{1}{n^2}\cdot (b-l)\cdot p(l,u)\cdot\left( \log(q(l,u)) - \log(q(b,u))\right)
	\end{align*}
	Note that we plug in the inverse minimax bid. As $\log(q(x,u))$ increases in $x$, the first partial derivative is negative.

	For intermediate values $v$, the partial derivatives are
	\begin{align*}
		\frac{\partial F_2}{\partial v} &:= -1 &<0\\
		\frac{\partial F_2}{\partial b} &:= p(l,u) &>0\\
		\frac{\partial F_2}{\partial l} &:= \frac{\partial p(l,u)}{\partial l}(b-l) - p(l,u)&<0\\
		\frac{\partial F_2}{\partial m} &:= \frac{\partial p(l,u)}{\partial m}(b-l) + 1 &<0 \\
		\frac{\partial F_2}{\partial u} &:= \frac{\partial p(l,u)}{\partial u}(b-l) &>0 \\
		\frac{\partial F_2}{\partial n} &:= \frac{1}{n^2} \cdot (b-l) \cdot p(l,u) \cdot \log(q(l,u))& < 0.
	\end{align*}

	The partial derivative with respect to $l$ is negative as the loss of bidding too high decreases in $x_1=l$. We show that $\frac{\partial F_2}{\partial m}<0$. Recall that this case applies for the intermediate $v$ that bid below the mean. In particular, loss of bidding too low is decreasing and increasing in $x_1$ for $x_1\in[b,m]$ only if $v\ge (n\cdot u-m)/(n-1)$. Moreover, note that choosing $b$ such that $p(l,u)(b-l)=v-m$ implies that $b-l=(v-m)/p(l,u)$. The first partial derivative with respect to $m$ simplifies to
	\begin{align*}
		\frac{\partial F_2}{\partial m} &= \frac{\partial p(l,u)}{\partial m}(b-l) + 1\\
		&= \frac{n-1}{n} \left(\frac{u-m}{u-l}\right)^{\frac{n-1}n-1}\frac{-1}{u-l}(b-l) + 1\\
		&= \frac{n-1}{n} \left(\frac{u-m}{u-l}\right)^{\frac{n-1}n-1}\frac{-1}{u-l}(v-m) \left(\frac{u-l}{u-m}\right)^{\frac{n-1}{n}}  + 1\\
		&= -\frac{n-1}{n} \frac{v-m}{u-m} + 1\\
		&\le - \frac{n-1}{n} \frac{1}{u-m} \left(\frac{n\cdot u-m}{n-1} -m\right) + 1\\
		&= - \frac{n-1}{n} \frac{1}{u-m} \frac{n(u-m)}{n-1} + 1 = 0,
	\end{align*}
	where we use the lower bound on $v$ for the inequality. 

	The partial derivatives for high types that bid above the expected winning bid $m$ are
	\begin{align*}
		\frac{\partial F_3}{\partial v} &:= - (1-p(l,b)) &< 0,\\
		\frac{\partial F_3}{\partial b} &:= p(l,u) + (1-p(l,b)) +(v-b) \frac{\partial p(l,b)}{\partial b} &> 0 \\
		\frac{\partial F_3}{\partial m} &:= \frac{\partial p(l,u)}{\partial m} (b-l)+ (v-b) \frac{\partial p(l,b)}{\partial m}  &<0\\
		\frac{\partial F_3}{\partial u} &:= \frac{\partial p(l,u)}{\partial u}(b-l) &>0.
	\end{align*}
	The partial derivative with respect to $n$ is negative, $\frac{\partial F_3}{\partial n}<0$, as $p$ decreases in $n$.

	The following numerical example demonstrates the ambiguous impact of $l$ on the minimax bid for high types. Let $n=u=2$ and $m=1$. We consider two values: $v=2$ and $v'=5$. We have $\beta(v=2|l=0)=1.05$ but $\beta(v=2|l=0.1)=1.06$. The minimax bid increases in $l$. On the other hand, we have $\beta(v=5|l=0)=1.69$ but $\beta(v=5|l=0.1)=1.68$. The minimax bid decreases in $l$.

	To summarize, the minimax bid increases in $v,$ $m$, and $n$. It decreases in $u$. For types that bid below the average winning bid $m$, the minimax bid increases in $l$. The impact of $l$ on high values' bids is ambiguous.


\subsection{Proof of Theorem~\ref{theorem:aggregate_moment_equilibrium}}
	The only consistent lower bound belief is $l^{\star}=\underline v.$ To see this, we first look for the value $v$ whose minimax bid is $l$ given arbitrary beliefs $(l,m,u)$, $l<m<u$. The worst-case loss associated with bidding too high (Eq.~\eqref{eq:loss-high-agg}) is 0 if $b=l$. The maximal loss associated with bidding too low (Eq.~\eqref{eq:loss-low-agg}) is $\max\{v-m,(v-l)\cdot p(l,u)\}$. The bid $b=l$ minimizes worst-case loss if the maximal loss associated with bidding too low is also 0. This cannot be the case for $v>m$. The only way that the worst-case loss of bidding too low is 0 is if $v=l$. The only way to satisfy this equality in a monotone equilibrium is $l^{\star}=\underline v$.

	For $l^{\star}=\underline v$ and any $m\le \bar v$, there exists a unique upper bound belief $u(m)$ such that the bidding function $\beta(v|l^{\star},m,u(m))$ is strictly increasing in $v$. The bidding function is strictly increasing if the bid $u$ equalizes the maximal loss of bidding too low and the maximal loss of bidding too high, i.e., if $\bar\lambda^L(u)=\bar\lambda^H(u)$ for the highest type $\bar v$. We solve for a root of the function
	\begin{align*}
		\bar\lambda^L(u)-\bar\lambda^H(u) &= (\bar v - u)(1-p(l,u)) - p(l,u)(u-l)=\bar v - u - p(l,u)(\bar v - l). 
	\end{align*}
	in u. The function is clearly strictly decreasing in $u$. The function equals $\bar v-m>0$ for $u$ as small as possible ($u=m$) and is negative for $u$ as large and plausible values (e.g., $u=\bar v$). Hence, for any $m<\bar v$ there is exactly one consistent upper bound belief $u(m)$. The implicit function theorem allows us to show that
	\begin{align*}
		u'(m) &= - \left(\frac{\partial \bar v-u-p(l,u)(\bar v-l)}{\partial m} \right)\cdot\left(\frac{\partial \bar v-u-p(l,u)(\bar v-l)}{\partial u} \right)^{-1} \\
		&= -\frac{\partial p(l,u)}{\partial m}(\bar v - l) \left(1 + \frac{\partial p(l,u)}{\partial u}(\bar v - l) \right)^{-1}>0.
	\end{align*}
	The consistent upper bound belief $u(m)$ is increasing in $m$.

	Suppose there is a single type, i.e., $\underline v = \overline v$. Then the equalities $\underline v = l^\star = m^\star = u^\star$ must hold in an aggregate moment equilibrium. 

	Let $\underline v < \overline v$. We now show that $\varphi^\AGG$ must have an interior fixed point, that is, $m^\star$ is such that $\underline v < m^\star < \overline v$. Recall from Proposition~\ref{prop:aggregate-minimax-bdding-function} that there is a cutoff type, now denoted by $\hat v(m)$,
	\[
	\hat v(m) = m + (m-l)\left(\frac{u-m}{u-l} \right)^{\frac{n-1}n},
	\]
	such that types below the cutoff bid below $m$ and types above the cutoff bid above $m$. Observe that as $m$ converges to $l$, which equals $\underline v$ in equilibrium, the cutoff converges to $\underline v$. Hence, as $m$ converges to $\underline v$, an increasing share of types bid above $m$. As only types close to $\underline v$ bid close to $m$, the expected bid must be higher than $m$.

\subsection{Proof of Proposition~\ref{prop:individual-minimax-bdding-function}}

The proof first shows that the worst-case bid distributions only put mass on at most two bids (Lemma~\ref{lemma:individual_two_bids}). These bids are the same for all other bidders in the worst case (Lemma~\ref{lemma:individual_identical}). We then maximize loss with respect to these two bids. Loss may again come from bidding too high and too low. The minimax bid equalizes the highest loss from bidding too high and the highest loss from bidding too low.

\begin{lemma}
	\label{lemma:individual_two_bids}
	The bid distribution of bidder $j$ that maximizes bidder $i$'s loss has at most two elements in the support.
\end{lemma}
\begin{proof}
	Recall that we maximize loss with respect to $\tilde b_i$ and the bid distributions $(B_k)_{k\neq i}$. We first change the order of taking suprema, i.e., we fix $(b_i,\tilde b_i)$. Then we consider any bid distributions of bidders $k\neq i,j$ and obtain the maximization problem
	\begin{align}
		\sup_{B_j\in\mathcal B_{ij}(l_i,\mu_i,u_i)}&\int \mathbb E_{B_{-i,j}}u_i(\tilde b_i,b_j,b_{-i,j}) - u_i(b_i,b_j,b_{-i,j})dB_j(b_j) \label{eq:individual_loss_swapped_order}
	\end{align}
	The extreme points of the constraint set $\mathcal B_{ij}(l_i,\mu_i,u_i)$ are distributions with at most two elements in the support \citep{Winkler_1988}. The objective function is linear in $B_j$. The maximum of a linear function is attained at an extreme point.
\end{proof}

The lemma states that the worst-case distribution $B_j$ has at most two elements in its support.\footnote{One can generalize the lemma. Let $M$ be the number of moment beliefs, where a moment belief is the expectation of a measurable function. The worst-case bid distribution has at most $M+1$ elements in the support.} Thus, it suffices to consider bid distributions of the form $\pi \cdot[x_1] + (1-\pi)\cdot [x_2]$, where $\pi \in[0,1]$ and $l\le x_1\le x_2\le u.$ The first moment constraint implies $\pi \cdot x_1+(1-\pi)\cdot x_2 = \mu$, which clearly requires $x_1\le\mu\le x_2$. The moment constraint implies 
\begin{equation}
	\pi(x_1,x_2) = \frac{x_2-\mu}{x_2-x_1}. \label{eq:pi}
\end{equation}
We sometimes drop the arguments.

Let $\pi_1$ and $\pi_2$ denote the partial derivatives of $\pi$ with respect to $x_1$ and $x_2$, respectively. These derivatives are equal to
\begin{equation*}
 	\pi_1 = \frac{\partial \pi}{\partial x_1} = \frac{x_2-\mu}{(x_2-x_1)^2} \text{ and }
 	\pi_2 = \frac{\partial \pi}{\partial x_2} = \frac{\mu-x_1}{(x_2-x_1)^2} .
\end{equation*} 
Observe that $\pi_1 > 0$ and $\pi_2 > 0$ for $x_1<\mu<x_2$.

\begin{lemma}
	\label{lemma:individual_identical}
	The worst-case bid distributions of the other bidders are identical.
\end{lemma}
\begin{proof}
	The worst-case bid distribution for each other bidder $j$ has at most two elements, $x_{j1}$ and $x_{j2}$, in the support. The best response to a discrete bid distribution is to bid (slightly above) one of the mass points. Tie-breaking does not play a role as suprema are taken. With the current moment constraints and for a given pair $(\tilde b,b)$ as in \eqref{eq:individual_loss_swapped_order}, the maximization of loss equals 
	\begin{align*}
		\sup_{x_{j1},x_{j2}:l\le x_{j1}\le \mu\le x_{j2}\le u}(v-\tilde b)\prod_{j\neq i} (\pi(x_{j1},x_{j2})\chi_{x_{j1}\le \tilde b}+(1-\pi(x_{j1},x_{j2}))\chi_{x_{j2}\le \tilde b})\\ - (v- b)\prod_{j\neq i} (\pi(x_{j1},x_{j2})\chi_{x_{j1}< b}+(1-\pi(x_{j1},x_{j2}))\chi_{x_{j2}< b}),
	\end{align*}
	where we break ties in favor of higher loss and $\chi_A$ denotes the indicator function, i.e., the function that takes the value 1 if $A$ is true and 0 otherwise. 

	If $\tilde b< \mu$, then loss is maximized by $x_{j1} = \tilde b$ as this maximizes the mass on $\tilde b$. The second mass point $x_{j2}$ is then chosen equal to $u$ as this maximizes the mass on $\tilde b$ and is higher than $b$.

	Let $\mu \le \tilde b$. If $\tilde b<b$, then loss is maximized by $x_{j2}=\tilde b$ and $x_{j1}=l$ as this maximizes the mass on $\tilde b$. If $\mu<b<\tilde b$, then the same construction maximizes loss. If $b<\mu<\tilde b$, then loss is again maximized by $x_{j2}=\tilde b$, but now $x_{j1}=b$. The bid $b$ always loses, while the bid $\tilde b$ wins with certainty. Thus, loss is always highest with symmetric worst-case distributions.
\end{proof}

Due to symmetric worst-case beliefs, the maximization of loss takes the form
\begin{align*} 
	\sup_{x_1,x_2}& \max\{\pi(x_1,x_2)^{n-1}(v-x_1), v-x_2,0\} - (v-b)(\chi_{x_1<b<x_2} \pi(x_1,x_2)^{n-1}+ \chi_{x_2<b})\\
	&\text{s.t. } l\le x_1\le \mu\le x_2\le u,
\end{align*}
for a given $b\in [l,u]$. Bidding above value is dominated, so we consider $b\le v$. The case in which the highest possible utility is 0 is then irrelevant as it never maximizes loss.

The maximization of loss requires the discussion of six exhaustive cases. There are three cases in which the bid $b$ is too high as the best response is lower. There are also six cases in which the bid $b$ is too low as the best response is higher. In what follows, we maximize loss for each of these six cases with respect to $x_1$ and $x_2$. We then find the minimax bids.

\paragraph{Loss of bidding too high.}
There are three cases in which one bids too high, that is, when the best response to the bid distribution is lower than $b$. The first case has $x_1<b<x_2$ and $x_1$ is the best response. The second case has $x_2<b$ with $x_1$ as the best response. The third case has $x_2<b$ with $x_2$ being the best response. 

Consider the first case of bidding too high, i.e., $x_1 < b < x_2$ and $x_1$ the best response. Loss equals $\lambda^{H1}(x_1,x_2)=\pi^{n-1}(b-x_1)$, which is to be maximized with respect to $x_1$ and $x_2$ under the constraints $l\le x_1\le \mu<x_2\le u$ and $x_1<b<x_2$.

We show that $\tilde x_1 = \max\{l,\min\{\hat x_1(b,u),\mu\}\}$ and $x_2 =u$ maximize loss, where
	\begin{equation*}
		\hat x_1(b,x_2)=\begin{cases}
			\frac{(n-1)b-x_2}{n-2} &\text{for } n\ge 3\\
			l&\text{for }n=2.
		\end{cases}
	\end{equation*}
Note that $\hat x_1(b,x_2)<b$ for $b<x_2.$ 
The partial derivatives of $\lambda^{H1}$ with respect to $x_1$ and $x_2$ equal
\begin{align*}
	\lambda^{H1}_1(x_1,x_2) &= \pi^{n-1}\left((n-1) \frac{b-x_1}{x_2-x_1} - 1\right) \text{ and}\\
	\lambda^{H1}_2(x_1,x_2) &= (n-1)\pi^{n-2}\pi_2(b-x_1),
\end{align*}
respectively. Loss increases in $x_2$ as $\lambda_2^{H1}>0$. Thus, loss is maximized by $x_2=u.$ For $n=2$, loss decreases in $x_1$ as $\lambda_1^{H1}<0$. It is, therefore, maximized by $\tilde x_1 = l = \max\{l,\min\{\hat x_1(b,u),\mu\}\}$. For $n\ge 3$, the root of $\lambda^{H1}_1$ is $\hat x_1(b,x_2).$ The second derivative of $\lambda^{H1}$ with respect to $x_1$ equals
\begin{equation*}
	\frac{\partial^2 \lambda^{H1}}{\partial x_1^2} = (n-1)\pi^{n-1}\frac{ b n-(n-2) {x_1}-2 {x_2}}{({x_2}-{x_1})^2},
\end{equation*}
which is negative at $x_1 = \hat x_1(b,x_2)$. To see this, note that $b n-(n-2) {\hat x_1}-2 {x_2}<0$ simplifies to $b < x_2$, which is true. Hence, for $\hat x_1(b,u)<l$, loss decreases in $x_1$ and is maximized by $x_1=l.$ For $l\le \hat x_1<\mu$, loss is maximized by $x_1=\hat x_1(b,u)$. For $\mu\le \hat x_1(b,u),$ loss increases in $x_1$ and is maximized by $x_1=\mu$. To summarize, loss is maximized by $\tilde x_1 = \max\{l,\min\{\hat x_1(b,u),\mu\}\}$ and $x_2=u$.

The second and the third case of bidding too high both lead to a lower loss than in the first case. To see this, note that loss in the second case equals $\pi^{n-1}(v-x_1)-(v-b)\le \pi^{n-1}(b-x_1)=\lambda^{H1}(x_1,x_2).$ The worst-case distribution never puts a mass point below $b$ when $x_1$ is the best response. In the third case, loss equals $b-x_2$, which is maximized by $x_2=\mu$. Note that this loss equals the loss of the first case when choosing $x_1=\mu$. This choice is optimal only if $\hat x_1(b,u)>\mu$, in which case the two cases yield the same maximal loss.

To summarize, the worst-case loss when bidding too high $\bar \lambda^H(b)$ is given by
\begin{equation*}
	\bar \lambda^H(b) = \left(\frac{u-\mu}{u-\max\{l,\min\{\hat x(b,u),\mu\}\}} \right)^{n-1}(b-\max\{l,\min\{\hat x(b,u),\mu\}\}).
\end{equation*}
Observe that $\bar \lambda^H$ is continuous in $b$. 

The range and slope of $\bar \lambda^H$ are as follows. First, the loss of bidding too high when bidding the lowest bid is 0, that is, $\bar \lambda^H(l)=0$ as $\hat x(l,u)\le l.$ Second, $\bar \lambda^H$ strictly increases in $b$. The statement is obvious unless $\hat x(b,u)=\max\{l,\min\{\hat x(b,u),\mu\}\}$. In this case, the derivative is 
\begin{equation*}
	\frac{d\,\bar \lambda^H(b)}{db} =\frac{d}{db}\max_{x_1} \lambda^{H1}(x_1,u)= \pi(\hat x(b,u),u)^{n-1}
\end{equation*}
by the envelop theorem.

\paragraph{Loss of bidding too low.}

There are also three cases when one bids too low, i.e., when the best response to the bid distribution is higher than $b$. The first case has $b<x_1$ with $x_1$ as the best response. The second case has $b<x_1$ and $x_2$ is the best response. The third case has $x_1<b<x_2$ with $x_2$ being the best response.

In the first case of bidding too low, loss equals $\lambda^{L1}(x_1,x_2)=\pi^{n-1}(v-x_1)$ as $b < x_1$ and $x_1$ is the best response. Clearly, the case applies only if $b<\mu$. Observe that we can draw from the analysis of $\lambda^{H1}$ as $\lambda^{L1}$ is $\lambda^{H1}$ when substituting $v$ for $b$. It follows that the partial derivative $\lambda_2^{L1}$ is positive, so that loss is maximized by $x_2=u.$ For $n=2$ and $v\le u$, loss $\lambda^{L1}$ is decreasing in $x_1$ as $\lambda_1^{L1}\le 0$ so that loss is maximized by $x_1=b$ and $x_2=u$. For $u < v$, loss is increasing in $x_1$ so that it is maximized by $x_1=\mu$. Maximal loss then equals $v-\mu$. For $n\ge 3$, the root of $\lambda_1^{L1}$ is $\hat x_1(v,x_2)$. The calculations above show that the second derivative $\lambda_{11}^{L1}$ is negative at $\hat x_1(v,x_2)$. Note that $\hat x_1(v,u)\le v$ if and only if $v\le u.$ Hence, for $v\le u$, loss $\lambda^{L1}$ is maximized by $x_1=\max\{b,\min\{\hat x_1(v,u),\mu\}\}$ and $x_2=u$. Observe that maximal loss is weakly larger than $v-\mu=\lambda^{L1}(\mu,u)$. For $u<v$, loss increases in $x_1$ and is maximized by $x_1=\mu$ and $x_2=u$. Note that then $\mu<u<v\le \hat x_1(v,u)$ so that $\max\{b,\min\{\hat x_1(v,u),\mu\}\} =\mu$ as $b<\mu$.

In the second case of bidding too low, loss equals $v-x_2$ as $b<x_1$ and $x_2$ is the best response. Maximal loss equals $v-\mu$, which is weakly less than the worst case of the previous case.

In the third case associated with bidding too low, loss equals $\lambda^{L3}(x_1,x_2)=v-x_2 - \pi^{n-1}(v-b)$ as $x_1<b<x_2$ with $x_2$ as the best response. We show that loss decreases in $x_1$ and $x_2$. The first partial derivatives of $\lambda^{L3}$ are $\lambda^{L3}_1(x_1,x_2) = -(n-1)(v-b) \pi^{n-2} \pi_1$ and $\lambda^{L3}_2(x_1,x_2) = -1-(n-1)(v-b)\pi^{n-2}\pi_2$. Both derivatives are negative. As a result, if $\mu\le b$, then loss is maximized by $x_1 = l$ and $x_2 = b$. If $b < \mu$, then the worst-case has $x_2 = \mu$, so that loss then equals $v-\mu.$ This loss is less than the worst-case loss $\lambda^{L1}(\max\{b,\min\{\mu,\hat x_1(v,u)\}\},u)$. 

To summarize, let $\bar \lambda^L$ denote the worst-case loss when bidding too low. The maximal loss $\bar \lambda^L$ equals
\begin{equation*}
	\bar \lambda^L(b) =
	\begin{cases}
		\left(\frac{u-\mu}{u-\max\{b,\min\{\hat x_1(v,u),\mu\}\}} \right)^{n-1}(v-\max\{b,\min\{\hat x_1(v,u),\mu\}\})&\text{for } b<\mu\\
		(v-b)\left(1-\left(\frac{b-\mu}{b-l} \right)^{n-1}\right)&\text{for } \mu\le b.
	\end{cases}
\end{equation*}
Note that $\bar \lambda^L$ is continuous in $b$. 

The range and slope of $\bar \lambda^L$ are as follows. First, the maximal loss of bidding too low is positive when bidding the lowest bid, i.e., $\bar \lambda^L(l)>0$. Second, loss $\bar \lambda^L$ is constant in $b$ if $b<\min\{\hat x(v,u),\mu\}$. Third, we show that maximal loss $\bar \lambda^L$ is strictly decreasing in $b$ when $\min\{\hat x(v,u),\mu\}\le b$. When $b<\mu$, then the maximization of loss $\lambda^{L1}$ is a constrained maximization problem $\max_{x_1}\lambda^{L1}(x_1,u)$ subject to $b\le x_1\le \mu$. The envelope theorem implies
\begin{equation*}
	\frac{d\bar\lambda^L(b)}{db}\biggr|_{b<\mu} = \frac{\partial}{\partial b}\lambda^{L1}(x_1,u) - \nu_1 (-x_1+b) - \nu_2(x_1-\mu) = -\nu_1,
\end{equation*}
where $\nu_1$ and $\nu_2$ are the respective Lagrange multipliers in the optimal solution. The derivative is negative as the constraint $x_1\ge b$ is binding so that the multiplier $\nu_1$ must be strictly positive by complementary slackness \citep{simon-blume}. For $\mu\le b$, the derivative of $\bar\lambda^L(b)=(v-b)(1-\pi(l,b)^{n-1})$ with respect to $b$ is
\begin{equation*}
	\frac{d\bar\lambda^L(b)}{db}\biggr|_{\mu\le b} = -(1-\pi(l,b)^{n-1}) -(n-1)(v-b)\pi(l,b)^{n-2}\pi_2(l,b)<0.
\end{equation*}
Finally, the loss of bidding too low is 0 when bidding value, that is, $\bar \lambda^L(v)=0.$

\paragraph{Minimization of maximal loss.}
Maximal loss is the upper envelope of $\bar\lambda^L(b)$ and $\bar\lambda^H(b)$, i.e., $\max\{\bar\lambda^L(b),\bar\lambda^H(b)\}$. Recall that $\bar\lambda^L(l)>0$ and $\bar\lambda^L(b)$ decreases in $b$. Moreover, $\bar\lambda^H(l)=0$ and $\bar\lambda^H(b)$ increases in $b$. Hence, maximal loss is minimized by the bid that equalizes the two expressions. Such a bid exists in $[l,v]$ as maximal loss is continuous in $b$ and $\bar\lambda^L(v)=0$. 

We now show that certain cases related to the loss of bidding too low are irrelevant at the minimax bid. First, we argue that the minimax bid $b^*$ is never such that $b^*<\hat x(v,u)\le \mu$. Let $v\le u$ so that $\hat x(v,u)\le v$. Suppose $l<\hat x(v,u)\le \mu$ so that $\bar \lambda^L(b)$ is constant in $b$ for $l\le b\le \hat x(v,u)$. We show $\bar \lambda^L(b)>\bar \lambda^H(b)$ for all $b\in[l,\hat x(v,u)]$. Hence, the two functions can only intersect at $b> \hat x(v,u)$. Recall that for relatively low bids $\bar\lambda^L(b)> \bar\lambda^H(b)$. Let $b\in[l,\hat x(v,u)]$. Recall that $x_1=\hat x(v,u)$ maximizes $\pi(x_1,u)^{n-1}(v-x_1)$. In particular, for $x_1'=\max\{l,\hat x(b,u)\}$ we have $\bar \lambda^L(b) = \bar \lambda^L(\hat x(v,u))\ge \pi(x_1',u)^{n-1}(v-x_1')$. Hence, we have
\begin{equation*}
	\bar \lambda^L(b)\ge \pi(x_1',u)^{n-1}(v-x_1') > \pi(x_1',u)^{n-1}(b-x_1')= \bar \lambda^H(b),
\end{equation*}
where the second inequality follows from $v>\hat x(v,u)\ge b$. The equality holds as $x_1'$ maximizes $\pi( x_1,u)^{n-1}(b-x_1)$.

We show that the minimax bid $b^*$ is never such that $b^*<\mu$ and $\bar\lambda(b^*)=v-\mu$. By way of contradiction, suppose this was true. Then we have either $u<v$ or $v\le u$ and $\hat x(v,u)>\mu$. Let $b<\mu$ and observe that $\lambda^{L1}(x_1,u)$ increases in $x_1$ so that the highest loss of bidding too low is $v-\mu$. Note that $\max\{l,\min\{\hat x(\mu,u),\mu\}\}<\mu$. We note that $\bar \lambda^H(b)\le\bar \lambda^H(\mu)$ and
\begin{align*}
	\bar \lambda^H(\mu)&=\frac{\left(u-\mu\right)^{n-1}}{\left(u-\max\{l,\min\{\hat x(\mu,u),\mu\}\}\right)^{n-1}} (\mu-\max\{l,\min\{\hat x(\mu,u),\mu\}\})\\
	&< \frac{\left(u-\mu\right)^{n-1}}{\left(u-\max\{l,\min\{\hat x(\mu,u),\mu\}\}\right)^{n-1}} (v-\max\{l,\min\{\hat x(\mu,u),\mu\}\})\\
	&=\lambda^{L1}(\max\{l,\min\{\hat x(\mu,u),\mu\}\},u)\le \bar \lambda^L(b)=v-\mu.
\end{align*}
The first line states the formula for the maximal loss when the bid $\mu$ is too high. The second line uses the fact that $\mu< v$ must hold. The third line observes that the second equals $\lambda^{L1}$ and that $\lambda^{L1}$ increases in $x_1$. We conclude that $\bar\lambda^H(b)=\bar\lambda^L(b)$ can be true only for $b\ge \mu$ if $\lambda^{L1}$ increases in $x_1$ for $b\le \mu$.

The previous two paragraphs show that $\bar \lambda^H$ and $\bar \lambda^L$ never intersect in the interior of the region on which $\bar \lambda^L(b)$ is constant in $b$. The worst-case loss of bidding too low $\bar \lambda^L$ simplifies to the relevant expression
\begin{equation*}
	\bar \lambda^L(b) =
	\begin{cases}
		\left(\frac{u-\mu}{u-b} \right)^{n-1}(v-b)&\text{for } b<\mu\\
		(v-b)\left(1-\left(\frac{b-\mu}{b-l} \right)^{n-1}\right)&\text{for } \mu\le b\le v.
	\end{cases}
\end{equation*}

	Equating the worst-case losses $\bar \lambda^H$ and $\bar \lambda^L$ gives the minimax bid in Proposition~\ref{prop:individual-minimax-bdding-function}.

	Let $n=2$. The first cutoff value $\hat v_1$ is such that $\beta^*(\hat v_1)=m$, and the second cutoff value is such that $\beta^*(\hat v_2)=u$. 


\subsection{Proof of Proposition~\ref{prop:comparative_statics_individual_moment_belief}}

The proof is analogous to the proof of Proposition~\ref{prop:comparative_statics_aggregate_moment_belief}. The minimization of maximal loss takes the form of finding $b$ such that $F(v,\mu,b)=\bar \lambda^L(b)-\bar\lambda^H(b)=0$. Note that there is a unique $b$ that solves the equation for each $v$ and $\mu$.

The proof of Proposition~\ref{prop:individual-minimax-bdding-function} implies that $\partial F/\partial b<0$ as $\bar \lambda^L$ decreases in $b$ and $\bar \lambda^H$ increases in $b$. 

The IFT implies that the minimax bid $\beta$ increases in $v$ as $\frac{\partial \beta}{\partial v} = -{\frac{\partial F}{\partial v}}\left({\frac{\partial F}{\partial b}}\right)^{-1}$ and $\frac{\partial F}{\partial v} = \frac{\partial \bar \lambda^L(b)}{\partial v}>0$.

We show that the minimax bid increases in $\mu$. The minimax bid is independent of the mean $\mu$ for values $v$ such that $\beta(v|l,\mu,u)\le \mu$. In general, the minimax bid increases in $\mu$ if $\partial F/\partial \mu \ge 0$. For bids above the mean, the maximal loss associated with bidding too low increases in $\mu$ as
\begin{equation*}
	\frac{\partial\bar\lambda^L}{\partial \mu}\biggr|_{\mu\le b} = (n-1)\frac{v-b}{b-l}\left(\frac{b-\mu}{b-l} \right)^{n-2}>0.
\end{equation*}
We now show that the maximal loss from bidding too high $\bar\lambda^H$ decreases in $\mu$. If $\max\{l,\min\{\hat x_1(b,u),\mu\}\}\neq \mu$, then the maximal loss of bidding too high decreases in $\mu$ as
\begin{equation*}
	\frac{\partial \bar\lambda^H}{\partial \mu}\biggr|_{\hat x_1(b,u)\le\mu} = -(n-1)\pi(\max\{l,\hat x_1(b,u)\},u)^{n-2}\frac{b-\max\{l,\hat x_1(b,u)\}}{u-\max\{l,\hat x_1(b,u)\}}<0. 
\end{equation*}
If $\mu=\max\{l,\min\{\hat x_1(b,u),\mu\}\}$, then $\bar\lambda^H(b)=b-\mu$ so that the derivative with respect to $\mu$ is $-1$. As a result, $\partial F/\partial \mu > 0$ and $\partial \beta/\partial \mu\ge 0$.

The first derivative of $\bar\lambda^H$ with respect to $n$ is
\begin{equation*}
	\frac{\partial \bar\lambda^H}{\partial n} =
	\begin{cases}
		(b-l) (\frac{u-\mu}{u-l})^{n-1} \log (\frac{u-\mu}{u-l})&\text{if }\tilde x_1 = l\\
		\frac{(n-1) (u-b)^2 }{(n-2)^2 (u-\mu)}\left(\frac{(n-2) (u-\mu)}{(n-1) (u-b)}\right)^n \log \left(\frac{(n-2)
   (u-\mu)}{(n-1) (u-b)}\right)&\text{if }\tilde x_1=\hat x(b,u)\\
   0&\text{if }\tilde x_1=\mu.
	\end{cases} 
\end{equation*}
The partial derivative is negative.

For bids above the mean, the partial derivative of $\bar \lambda^L$ with respect to $n$ is
\begin{equation*}
	\frac{\partial \bar\lambda^L}{\partial n} \biggr|_{\mu\le b}=- (v-b) \left(\frac{b-\mu}{b-l}\right)^{n-1} \log \left(\frac{b-\mu}{b-l}\right)>0.
\end{equation*}
Hence, for bids above the mean $\partial (\bar\lambda^L-\bar\lambda^H)/\partial n >0$. 

For bids below the mean $\mu$, we consider the inverse minimax bid $\beta^{-1}$. The partial derivative of $\beta^{-1}$ with respect to $n$ is
\begin{equation*}
		\frac{\beta^{-1}}{\partial n} \biggr|_{b<m}=
		\begin{cases}
			(b-l) (\frac{u-b}{u-l})^{n-1} \log (\frac{u-b}{u-l}) &\text{if } l\ge \hat x(b,u)\\
			\frac{u-b}{n-2}(\frac{n-2}{n-1})^{n-1}  \log (\frac{n-2}{n-1}) &\text{else}.
		\end{cases}
\end{equation*}
The derivative is negative in both cases. Note that taking the total derivative of $\beta ^{-1}(\beta(v|l,\mu,u)|l,\mu,u)=v$ with respect to $\mu$ on both sides and reformulating leads to 
\begin{equation*}
	\frac{\partial \beta(v|l,\mu,u)}{\partial \mu}=-\frac{\partial \beta ^{-1}(\beta(v|l,\mu,u)|l,\mu,u)}{\partial \mu} \left(\frac{\partial \beta ^{-1}(\beta(v|l,\mu,u)|l,\mu,u)}{\partial b}\right)^{-1}.
\end{equation*}
The inverse bid increases in $b$ and the inverse bid decreases in $n$. Thus, the minimax bid increases in $n$.



\subsection{Proof of Theorem~\ref{thm:identification_aggregate_meq}}
Let $l<m<u$. The inverse bidding function is strictly increasing on $[l,u]$ as $l<m<u$. We use the strictly increasing inverse bidding function $v^\AGG(\cdot|l,m,u)$ to construct the value distribution $F$ by defining $v=v^\AGG(b|l,m,u)$ for all $b\in[l,u]$. Clearly, $F$ is a c.d.f. with support $[\underline v,\overline v]$, where $\underline v = v^\AGG(l|l,m,u) = l$ and $\overline v =v^\AGG(u|l,m,u)$. Note that as $v^\AGG$ is strictly increasing in $b$, so is its inverse $\beta^\AGG$ in $v$, and $\beta^\AGG(v^\AGG(b))=b.$ Hence, $F(v) =G(\beta^\AGG(v|l,m,u))$ for all $v\in[\underline v,\overline v]$. The equilibrium is separating by construction and has consistent range beliefs by construction. To see that $G$ is the bid distribution of a moment equilibrium, observe that 
\begin{equation*}
	\int_{\underline v}^{\overline v}\beta^\AGG(v| l, m, u)dF^n(v) = \int_{l}^{u}\beta^\AGG(v^\AGG(b|l,m,u)|l,m,u)dG^n(b)=\int_{l}^{u}bdG^n(b)=m.
\end{equation*}
Hence, the moment condition holds, and $(\beta^\AGG,l,m,u)$ is a separating aggregate moment equilibrium bid distribution $G$ when the value distribution is $F$ as constructed.

Let there be a value distribution $F$ and $G$ a corresponding bid distribution of a separating aggregate moment equilibrium. Separation and $\underline v<\overline v$ implies $l<m<u$.

\subsection{Proof of Proposition~\ref{prop:affine}}

To prove that if $b$ is a minimax bid for $l_0,\mu_0,u_0$, and $v$, then $p\cdot b + q$ is a minimax bid for $p\cdot l_0 + q, p\cdot \mu_0 + q, p\cdot u_0 + q$, and $p\cdot v + q$, one plugs the affine transformations into equations~\eqref{eq:ind_minimax_bid_low_value} and \eqref{eq:ind_minimax_bid_high_value} and observes that $p$ and $q$ drop out.

To see that $(\beta^\IND_\alpha,\alpha(l_0;\mathbf y),\alpha(\mu_0;\mathbf y),\alpha(u_0;\mathbf y))$ is a moment equilibrium for distribution $F_\alpha$, observe that
\begin{align*}
	\int_{\underline v}^{\bar v} \beta(\alpha(v;\mathbf y)&|\alpha(l_0;\mathbf y),\alpha(\mu_0;\mathbf y), \alpha(u_0;\mathbf y)) dF_\alpha(\alpha(v;\mathbf y)) \\
	&= \int_{\underline v}^{\bar v}p(\mathbf y)\beta(v|l_0,\mu_0,u_0) + q(\mathbf y) dF(v) = \alpha(\mu_0;\mathbf y).
\end{align*}

An analogous proof works for aggregate moment beliefs.


\section{Incentives in Procurement Auctions}
\label{app:empirical}
In procurement auctions, the incentives are reversed compared to the usual bidder-as-buyer auction. As bidders are sellers in procurement auctions, the bidder with the lowest bid wins. The private information is the privately known cost of producing the good. Bidders incur this cost conditional on winning the auction. In a bidder-as-seller auction, the expected utility of bidder $i$ when bidding $b_i$ is
\begin{align*}
	\mathbb P(b_i<\min b_{-i})(b_i-c_i) &= (1-\mathbb P(\min b_{-i}< b_i))(b_i-c_i) \\
	&= (1 - (1 - (1 - B(b_i))^{n-1})) (b_i-c_i)\\
	&= (1 - B(b_i))^{n-1}(b_i-c_i),
\end{align*}
where the bids of the other bidders are drawn from distribution $B$.

We first consider the incentives with aggregate moment beliefs. Let $l$ be the lowest bid and $u$ the highest bid of the other bidders. Let $m$ denote the expected winning bid, i.e., the expected minimum of $n$ independent bids. The worst-case bid distribution has the elements $x_1$ and $x_2$ in the support, $l\le x_1 \le m \le x_2\le u$. The moment constraint implies $q\cdot x_1 + (1-q)\cdot x_2 = m$, so $q(x_1,x_2) = \frac{x_2 - m}{x_2 - x_1} $. The lowest bid of $n$ independent bids is $x_1$ if not all bid $x_2$, i.e., $q = 1 - \mathbb P(b_j = x_2)^n  = 1 - (1-\mathbb P(b_j = x_1))^n$. The probability that a competing bidder bids $x_1$ is therefore $\mathbb (b_j = x_1) = 1 - (1-q)^\frac{1}{n} $. The probability that $x_1$ is the lowest of $n-1$ independent draws is then $1 - (1 - q)^{\frac{n-1}{n}} $. The probability that a bid $b\in (x_1,x_2)$ is winning is then $(1-q)^\frac{n-1}{n}.$

The maximal loss conditional on bidding too high is $\max\{(1-q(l,b))^\frac{n-1}{n}(b-c),m-c\}$ if $m< b$ and $(b-c)(1- (1-q(b,u))^\frac{n-1}{n})$ if $b\le m$. The maximal loss conditional on bidding too low is $(1-q(l,u))^\frac{n-1}{n}(u-b)$. The minimax bid equates the maximal loss conditional on bidding too high and the maximal loss conditional on bidding too low. The inverse minimax bidding function $c^\AGG(b|l,m,u)$ is well defined in a separating aggregate moment equilibrium.

We now consider individual moment beliefs in a procurement context. If other bidders only use two bids, $x_1$ and $x_2$, then the probability that bid $b\in (x_1,x_2)$ wins if all $n-1$ other bidders bid $x_2$. In a procurement auction, the bid is too low if a higher bid would not decrease the probability of winning. The worst-case bid distribution is then $\pi[l]+(1-\pi)x_2^P$, where $x_2^P = \min\{u,\max\{m,\hat x_2^P\}\}$ and
\begin{equation*}
	\hat x_2^P=
	\begin{cases}
		\frac{(n-1)b-l}{n-2}&\text{if }n\ge 3\\
		u&\text{if }n=2 .
	\end{cases}
\end{equation*}
Maximal loss conditional on bidding too high equals $\left(\frac{m-l}{x_2^P-l} \right)^{n-1}(x_2^P-b)$. The worst case when bidding too high comes from bidding marginally too high. A lower bid would have won. If $m< b$, then the worst-case bid distribution is $\pi [l]+(1-\pi)[b-\epsilon]$ so that maximal loss is $\left(\frac{m-l}{b-l} \right)^{n-1}(b-c)$. If $b\le m$, then the worst-case bid distribution is $\pi[b-\epsilon]+(1-\pi)[u]$. Maximal loss equals $(b-c)(1-\left(\frac{m-b}{u-b} \right)^{n-1})$. In a separating equilibrium, the (inverse) minimax bidding function equalizes the maximal loss from bidding too high with the maximal loss from bidding too low and equals
\begin{equation}
	c^\IND (b|l,m,u) = 
	\begin{cases}
		b - \left(\frac{b-l}{x_2^P-l} \right)^{n-1}(x_2^P-b) &\text{if } m<b\\
		b - \left(\frac{(m-l)(u-b)}{x_2^P-l} \right)^{n-1} \frac{x_2^P-b}{(u-b)^{n-1}-(m-b)^{n-1}}&\text{if }b\le m 
	\end{cases}
	\label{eq:minimax-bid-procurement}
\end{equation}
for $b\in[l,u]$.

The upper bound belief $u$ must equal the highest cost $\bar c$ in a moment equilibrium.

\singlespacing

\renewcommand{\bibname}{References} 
\bibliographystyle{chicago}
\bibliography{robust-moment}

\begin{thebibliography}{}

\bibitem[\protect\citeauthoryear{Athey and Haile}{Athey and
  Haile}{2007}]{athey-haile-2007-handbook}
Athey, S. and P.~Haile (2007).
\newblock {\em Nonparametric Approaches to Auctions}, Volume~6A of {\em
  Handbook of Econometrics}, Chapter~60, pp.\  3847--3965.
\newblock Elsevier.

\bibitem[\protect\citeauthoryear{Bajari, Houghton, and Tadelis}{Bajari
  et~al.}{2014}]{BHT-AER-2014}
Bajari, P., S.~Houghton, and S.~Tadelis (2014, April).
\newblock Bidding for incomplete contracts: An empirical analysis of adaptation
  costs.
\newblock {\em American Economic Review\/}~{\em 104\/}(4), 1288--1319.

\bibitem[\protect\citeauthoryear{Battigalli, Cerreia-Vioglio, Maccheroni, and
  Marinacci}{Battigalli et~al.}{2015}]{10.1257/aer.20130930}
Battigalli, P., S.~Cerreia-Vioglio, F.~Maccheroni, and M.~Marinacci (2015,
  February).
\newblock Self-confirming equilibrium and model uncertainty.
\newblock {\em American Economic Review\/}~{\em 105\/}(2), 646--77.

\bibitem[\protect\citeauthoryear{Battigalli, Cerreia-Vioglio, Maccheroni, and
  Marinacci}{Battigalli et~al.}{2016}]{BATTIGALLI201640}
Battigalli, P., S.~Cerreia-Vioglio, F.~Maccheroni, and M.~Marinacci (2016).
\newblock Analysis of information feedback and selfconfirming equilibrium.
\newblock {\em Journal of Mathematical Economics\/}~{\em 66}, 40 -- 51.

\bibitem[\protect\citeauthoryear{Bergemann and Morris}{Bergemann and
  Morris}{2005}]{BM05}
Bergemann, D. and S.~Morris (2005, November).
\newblock Robust mechanism design.
\newblock {\em Econometrica\/}~{\em 73\/}(6), 1771--1813.

\bibitem[\protect\citeauthoryear{Bergemann and Schlag}{Bergemann and
  Schlag}{2011}]{BERGEMANN20112527}
Bergemann, D. and K.~H. Schlag (2011).
\newblock Robust monopoly pricing.
\newblock {\em Journal of Economic Theory\/}~{\em 146\/}(6), 2527--2543.

\bibitem[\protect\citeauthoryear{Carrasco, Luz, Kos, Messner, Monteiro, and
  Moreira}{Carrasco et~al.}{2018}]{Carrasco-et-al-2018}
Carrasco, V., V.~F. Luz, N.~Kos, M.~Messner, P.~Monteiro, and H.~Moreira
  (2018).
\newblock Optimal selling mechanisms under moment conditions.
\newblock {\em Journal of Economic Theory\/}~{\em 177}, 245 -- 279.

\bibitem[\protect\citeauthoryear{Carroll}{Carroll}{2015}]{Carroll15}
Carroll, G. (2015, February).
\newblock Robustness and linear contracts.
\newblock {\em American Economic Review\/}~{\em 105\/}(2), 536--563.

\bibitem[\protect\citeauthoryear{Carroll and Segal}{Carroll and
  Segal}{2019}]{Carroll-Segal-resale-2016}
Carroll, G. and I.~Segal (2019).
\newblock Robustly optimal auctions with unknown resale opportunities.
\newblock {\em Review of Economic Studies\/}~{\em 86\/}(4), 1527--1555.

\bibitem[\protect\citeauthoryear{Chen, Katu{\v s}{\v c}{\'a}k, and
  Ozdenoren}{Chen et~al.}{2007}]{chen2007}
Chen, Y., P.~Katu{\v s}{\v c}{\'a}k, and E.~Ozdenoren (2007).
\newblock Sealed bid auctions with ambiguity: Theory and experiments.
\newblock {\em Journal of Economic Theory\/}~{\em 136\/}(1), 513 -- 535.

\bibitem[\protect\citeauthoryear{Donald and Paarsch}{Donald and
  Paarsch}{2002}]{donald-paarsch-2002}
Donald, S.~G. and H.~J. Paarsch (2002).
\newblock Superconsistent estimation and inference in structural econometric
  models using extreme order statistics.
\newblock {\em Journal of Econometrics\/}~{\em 109}, 305--340.

\bibitem[\protect\citeauthoryear{Dow and Werlang}{Dow and
  Werlang}{1994}]{DOW1994305}
Dow, J. and S.~R. D.~C. Werlang (1994).
\newblock Nash equilibrium under knightian uncertainty: Breaking down backward
  induction.
\newblock {\em Journal of Economic Theory\/}~{\em 64\/}(2), 305 -- 324.

\bibitem[\protect\citeauthoryear{Eichberger and Kelsey}{Eichberger and
  Kelsey}{2000}]{EICHBERGER2000183}
Eichberger, J. and D.~Kelsey (2000).
\newblock Non-additive beliefs and strategic equilibria.
\newblock {\em Games and Economic Behavior\/}~{\em 30\/}(2), 183 -- 215.

\bibitem[\protect\citeauthoryear{Elstrodt}{Elstrodt}{2018}]{Elstrodt}
Elstrodt, J. (2018).
\newblock {\em Ma{\ss}- und Integrationstheorie\/} (8 ed.).
\newblock Springer.

\bibitem[\protect\citeauthoryear{Gilboa and Schmeidler}{Gilboa and
  Schmeidler}{1989}]{GILBOA1989141}
Gilboa, I. and D.~Schmeidler (1989).
\newblock Maxmin expected utility with non-unique prior.
\newblock {\em Journal of Mathematical Economics\/}~{\em 18\/}(2), 141 -- 153.

\bibitem[\protect\citeauthoryear{Gretschko and Mass}{Gretschko and
  Mass}{2018}]{Gretschko-Mass}
Gretschko, V. and H.~Mass (2018).
\newblock Worst-case subjective-belief equilibria in first-price auctions.
\newblock {\em Working paper\/}.

\bibitem[\protect\citeauthoryear{Guerre, Perrigne, and Vuong}{Guerre
  et~al.}{2000}]{GPV}
Guerre, E., I.~Perrigne, and Q.~Vuong (2000).
\newblock Optimal nonparametric estimation of first-price auctions.
\newblock {\em Econometrica\/}~{\em 68\/}(3), 525--574.

\bibitem[\protect\citeauthoryear{Guo and Shmaya}{Guo and
  Shmaya}{2019}]{RobustMonopolyRegulation}
Guo, Y. and E.~Shmaya (2019).
\newblock Robust monopoly regulation.
\newblock {\em Working paper\/}.

\bibitem[\protect\citeauthoryear{Guo and Shmaya}{Guo and
  Shmaya}{2021}]{guo-shmaya-project-choice}
Guo, Y. and E.~Shmaya (2021).
\newblock Project choice from a verifiable proposal.
\newblock {\em Working paper\/}.

\bibitem[\protect\citeauthoryear{Hong, Haile, and Shum}{Hong
  et~al.}{2003}]{hong-haile-shum}
Hong, H., P.~Haile, and M.~Shum (2003).
\newblock Nonparametric tests for common values in first-price sealed-bid
  auctions.
\newblock {\em Working paper\/}.

\bibitem[\protect\citeauthoryear{Horta{\c c}su, Luco, Puller, and Zhu}{Horta{\c
  c}su et~al.}{2019}]{10.1257/aer.20172015}
Horta{\c c}su, A., F.~Luco, S.~L. Puller, and D.~Zhu (2019, December).
\newblock Does strategic ability affect efficiency? evidence from electricity
  markets.
\newblock {\em American Economic Review\/}~{\em 109\/}(12), 4302--42.

\bibitem[\protect\citeauthoryear{Hyafil and Boutilier}{Hyafil and
  Boutilier}{2004}]{hyafill-boutilier}
Hyafil, N. and C.~Boutilier (2004).
\newblock Regret minimizing equilibria and mechanisms for games with strict
  type uncertainty.
\newblock In {\em Proceedings of the Twentieth Annual Conference on Uncertainty
  in Artificial Intelligence}, pp.\  268--277.

\bibitem[\protect\citeauthoryear{Kasberger and Schlag}{Kasberger and
  Schlag}{2022}]{kasberger-schlag}
Kasberger, B. and K.~H. Schlag (2022).
\newblock Robust bidding in first-price auctions: How to bid without knowing
  what others are doing.
\newblock {\em Working paper\/}.

\bibitem[\protect\citeauthoryear{Klibanoff}{Klibanoff}{1996}]{Klibanoff1996}
Klibanoff, P. (1996).
\newblock Uncertainty, decision, and normal form games.
\newblock {\em MIT Working Paper\/}.

\bibitem[\protect\citeauthoryear{Krasnokutskaya}{Krasnokutskaya}{2011}]{Krasnokutskaya-20110-restud}
Krasnokutskaya, E. (2011, 01).
\newblock {Identification and Estimation of Auction Models with Unobserved
  Heterogeneity}.
\newblock {\em The Review of Economic Studies\/}~{\em 78\/}(1), 293--327.

\bibitem[\protect\citeauthoryear{Lehrer}{Lehrer}{2012}]{10.1257/mic.4.1.70}
Lehrer, E. (2012, February).
\newblock Partially specified probabilities: Decisions and games.
\newblock {\em American Economic Journal: Microeconomics\/}~{\em 4\/}(1),
  70--100.

\bibitem[\protect\citeauthoryear{Linhart and Radner}{Linhart and
  Radner}{1989}]{LINHART1989152}
Linhart, P. and R.~Radner (1989).
\newblock Minimax-regret strategies for bargaining over several variables.
\newblock {\em Journal of Economic Theory\/}~{\em 48\/}(1), 152 -- 178.

\bibitem[\protect\citeauthoryear{Lo}{Lo}{1996}]{LO1996443}
Lo, K.~C. (1996).
\newblock Equilibrium in beliefs under uncertainty.
\newblock {\em Journal of Economic Theory\/}~{\em 71\/}(2), 443 -- 484.

\bibitem[\protect\citeauthoryear{Lo}{Lo}{1998}]{10.2307/25055108}
Lo, K.~C. (1998).
\newblock Sealed bid auctions with uncertainty averse bidders.
\newblock {\em Economic Theory\/}~{\em 12\/}(1), 1--20.

\bibitem[\protect\citeauthoryear{Marinacci}{Marinacci}{2000}]{MARINACCI2000191}
Marinacci, M. (2000).
\newblock Ambiguous games.
\newblock {\em Games and Economic Behavior\/}~{\em 31\/}(2), 191 -- 219.

\bibitem[\protect\citeauthoryear{Oll{\'a}r and Penta}{Oll{\'a}r and
  Penta}{2017}]{10.1257/aer.20151462}
Oll{\'a}r, M. and A.~Penta (2017, August).
\newblock Full implementation and belief restrictions.
\newblock {\em American Economic Review\/}~{\em 107\/}(8), 2243--77.

\bibitem[\protect\citeauthoryear{Osborne and Rubinstein}{Osborne and
  Rubinstein}{2003}]{OSBORNE2003434}
Osborne, M.~J. and A.~Rubinstein (2003).
\newblock Sampling equilibrium, with an application to strategic voting.
\newblock {\em Games and Economic Behavior\/}~{\em 45\/}(2), 434 -- 441.
\newblock Special Issue in Honor of Robert W. Rosenthal.

\bibitem[\protect\citeauthoryear{Perakis and Roels}{Perakis and
  Roels}{2008}]{Perakis-Roels-2008}
Perakis, G. and G.~Roels (2008).
\newblock Regret in the newsvendor model with partial information.
\newblock {\em Operations Research\/}~{\em 56\/}(1), 188--203.

\bibitem[\protect\citeauthoryear{Popescu}{Popescu}{2005}]{Popescu2005}
Popescu, I. (2005).
\newblock A semidefinite programming approach to optimal-moment bounds for
  convex classes of distributions.
\newblock {\em Mathematics of Operations Research\/}~{\em 30\/}(3), 632--657.

\bibitem[\protect\citeauthoryear{Renou and Schlag}{Renou and
  Schlag}{2010}]{RENOU2010264}
Renou, L. and K.~H. Schlag (2010).
\newblock Minimax regret and strategic uncertainty.
\newblock {\em Journal of Economic Theory\/}~{\em 145\/}(1), 264 -- 286.

\bibitem[\protect\citeauthoryear{Roughgarden}{Roughgarden}{2020}]{roughgarden-eip-1559}
Roughgarden, T. (2020).
\newblock Transaction fee mechanism design for the ethereum blockchain: An
  economic analysis of eip-1559.
\newblock {\em Working paper\/}.

\bibitem[\protect\citeauthoryear{Salant and Cherry}{Salant and
  Cherry}{2020}]{doi:10.3982/ECTA17105}
Salant, Y. and J.~Cherry (2020).
\newblock Statistical inference in games.
\newblock {\em Econometrica\/}~{\em 88\/}(4), 1725--1752.

\bibitem[\protect\citeauthoryear{Savage}{Savage}{1951}]{savage1951}
Savage, L.~J. (1951).
\newblock The theory of statistical decision.
\newblock {\em Journal of the American Statistical Association\/}~{\em
  46\/}(253), 55--67.

\bibitem[\protect\citeauthoryear{Savage}{Savage}{1954}]{savage1954}
Savage, L.~J. (1954).
\newblock {\em The Foundations of Statistics}.
\newblock Wiley.

\bibitem[\protect\citeauthoryear{Schlag and Zapechelnyuk}{Schlag and
  Zapechelnyuk}{2020}]{Schlag-Zapechelnyuk-best-compromise}
Schlag, K.~H. and A.~Zapechelnyuk (2020).
\newblock Compromise, don't optimize: A prior-free alternative to bayesian nash
  equilibrium.
\newblock {\em Working paper\/}.

\bibitem[\protect\citeauthoryear{Schmeidler}{Schmeidler}{1989}]{schmeidler1989}
Schmeidler, D. (1989).
\newblock Subjective probability and expected utility without additivity.
\newblock {\em Econometrica\/}~{\em 57\/}(3), 571--587.

\bibitem[\protect\citeauthoryear{Simon and Blume}{Simon and
  Blume}{1994}]{simon-blume}
Simon, C.~P. and L.~Blume (1994).
\newblock {\em Mathematics for Economists}.
\newblock W. W. Norton \& Company.

\bibitem[\protect\citeauthoryear{Smith}{Smith}{1995}]{Smith1995}
Smith, J.~E. (1995).
\newblock Generalized chebychev inequalities: Theory and applications in
  decision analysis.
\newblock {\em Operations Research\/}~{\em 43\/}(5), 807--825.

\bibitem[\protect\citeauthoryear{Tukey}{Tukey}{1977}]{tukey}
Tukey, J.~W. (1977).
\newblock {\em Exploratory Data Analysis}.
\newblock Addison-Wesley.

\bibitem[\protect\citeauthoryear{Wilson}{Wilson}{1987}]{wilson-critique}
Wilson, R. (1987).
\newblock Game-theoretic analyses of trading processes.
\newblock In T.~F. Bewley (Ed.), {\em Advances in Economic Theory}, Fifth World
  Congress, pp.\  33--70. Cambridge University Press.

\bibitem[\protect\citeauthoryear{Winkler}{Winkler}{1988}]{Winkler_1988}
Winkler, G. (1988).
\newblock Extreme points of moment sets.
\newblock {\em Mathematics of Operations Research\/}~{\em 13\/}(4), 581--587.

\end{thebibliography}

\end{document}